\documentclass[10pt,usenames,dvipsnames]{article}
\usepackage{enumerate}
\usepackage{amsmath,amssymb}
\usepackage{natbib}
\usepackage{caption}
\usepackage{enumitem}
\usepackage{comment}
\usepackage[usenames]{color}
\usepackage{booktabs} 
\usepackage{bm}
\usepackage{ying}
\usepackage{multirow}
\usepackage{rotating}
\usepackage{xspace}
\usepackage{fullpage}
\usepackage{authblk}
\usepackage{wrapfig}
\usepackage{mathtools}
\mathtoolsset{showonlyrefs}
\usepackage{geometry}
\usepackage{parskip}
\usepackage{float}
\usepackage{subcaption}
\usepackage{tikz,tabularx}
\usetikzlibrary{patterns}
\usetikzlibrary{arrows}
\usetikzlibrary{tikzmark, positioning, fit, shapes.misc}

\usetikzlibrary{decorations.pathreplacing, calc}
\tikzset{brace/.style={decorate, decoration={brace}},
  brace mirrored/.style={decorate, decoration={brace,mirror}},
}

\newcolumntype{g}{>{\columncolor{red}}c}

\newcommand{\mname}{\textsc{CPL}\xspace}

\usepackage{graphicx,amssymb}
\usepackage[table]{xcolor}
\definecolor{lightgray}{gray}{0.9}

\usepackage[colorlinks,
            linkcolor=red,
            anchorcolor=blue,
            citecolor=blue
            ]{hyperref}
\usepackage{algorithm}
\usepackage{algorithmic}

\def \iid {\stackrel{\text{i.i.d.}}{\sim}}

\def \calib {\textrm{calib}}
\def \train {\textrm{train}}

\def \test {\textrm{test}}
\def \obs {\textrm{obs}}

\def \welfare {{\textnormal{welfare}}}
\def \Welfare {{\textnormal{Welfare}}}
\def \fna {\gamma}
\def \power {{\textnormal{power}}}
\providecommand{\keywords}[1]
{{ 
  \fontsize{9}{12}\selectfont
  \textbf{\textit{Keywords:}} #1
}}
\renewcommand{\widehat}{\hat}

\theoremstyle{plain}

\allowdisplaybreaks

\usepackage{hyperref}
\def\@#1\@{\begin{align}#1\end{align}}
\def\$#1\${\begin{align*}#1\end{align*}}

\usepackage{color-edits}
\addauthor[Ying]{ying}{magenta}

\title{Conformal Policy Learning with \\ Distribution-Free Safety Guarantees}
\author[1]{Ying Jin\thanks{The reproduction code is in the GitHub repository \url{https://github.com/ying531/conformal-policy-learning}. Email: \url{yjinstat@wharton.upenn.edu}}} 
\author[2]{Naoki Egami}
\affil[1]{Department of Statistics and Data Science, University of Pennsylvania} 
\affil[2]{Department of Political Science \& Statistics and Data Science Center, \hspace{0.5in} Massachusetts Institute of Technology}
\date{}

\begin{document}

\maketitle

\begin{abstract}
Policy learning aims to determine who should be treated based on individual characteristics. In high-stakes settings such as medicine and public policy where safety is a central concern, improving the average outcomes alone may not be sufficient: decision makers may also seek to protect individuals from harm, in line with the Hippocratic principle of ``do no harm.'' 

In this paper, we propose \textit{conformal policy learning} (CPL), a policy learning procedure with a new distribution-free safety guarantee that controls the probability of assigning treatment to an individual who would be harmed relative to control. CPL views each treatment decision as testing a hypothesis of counterfactual harm and assigns treatment by thresholding conformal p-values. These p-values use observable proxies and selective calibration to address the challenge that the potential outcomes under comparison are never simultaneously observed. For randomized experiments, under standard exchangeability conditions, CPL provides finite-sample safety guarantee at a user-specified level, without imposing any outcome modeling assumptions. Moreover, when the outcome model is consistently estimated, CPL achieves asymptotically optimal welfare subject to the safety constraint. 
In observational studies, CPL with learn-then-balance weights achieves doubly robust safety guarantees. We evaluate CPL through extensive simulations and apply it to an empirical study of AI-powered interventions designed to reduce conspiracy beliefs.
  \end{abstract}
% \fontsize{10}{12}\selectfont

% !TEX root = main.tex

\section{Introduction}

%% Policy learning is everywhere. 
Policy learning, also known as the treatment choice problem, aims to learn a rule that automatically assigns future treatment options based on individual characteristics~\citep{manski2004statistical,hirano2009asymptotics,kitagawa2018should,athey2021policy}.
% deciding who should receive treatment and who should receive control. 
It has been the foundation for data-driven decision-making in various domains spanning precision medicine~\citep{murphy2003optimal,qian2011performance,zhao2012estimating}, online advertising and recommendation systems~\citep{li2010contextual,dudik2011doubly}, political campaigns~\citep{imai2011estimation}, and criminal justice~\citep{kleinberg2018human}, among others. 

Policy learning methods often aim to maximize the average welfare (expectation of the realized outcome) within a policy class~\citep{manski2004statistical,kitagawa2018should,athey2021policy}. 
While welfare maximization is widely useful, it can be insufficient in high-stakes domains such as medicine and public policy where individual safety is of concern. In such applications, decision makers may seek not only on-average improvement of the outcomes, but also controlling the number of individuals harmed by the intervention~\citep{gadbury2004individual,kallus2022s, richens2022counterfactual, ben2025safe}, i.e., ``do no harm.''  
As an example, suppose a new intervention benefits $60\%$ of a group while harming the remaining $40\%$ by the same magnitude.  
A decision-maker who maximizes average welfare would decide to treat the group; this would harm 40\% of the population, which can be unacceptable if individual safety is of primary concern.

%% The usual goal is the welfare maximization. 

%% However, it is often not enough to maximize the welfare on average. 

In this paper,  we study policy learning with a distribution-free safety guarantee. Formally, assume access to observed data $\{(X_i, T_i, Y_i)\}_{i=1}^n$ where $Y_i\in \{0,1\}$ is a binary outcome, $T_i\in\{0,1\}$ is the binary treatment, and $X_i\in \cX$ is the observed features for each unit $i$. Under the potential outcome framework with SUTVA (formalized in Section~\ref{subsec:setup}), the outcome is $Y_i=Y_i(T_i)$, where $(Y_i(1),Y_i(0))$ are the potential outcomes under treatment and control, respectively. 
For a new test point with observed features $X_{n+1}$ and unknown potential outcomes $(Y_{n+1}(0),Y_{n+1}(1))$, our goal is to learn a policy $\hat\pi: \mathcal{X} \rightarrow \{0, 1\}$ that maps features to treatment assignments with the following safety guarantee: for a pre-specified level $\alpha\in (0,1)$,  
\begin{equation}\label{eq:def_safety_intro}
    \PP\bigl(Y_{n+1}(\hat\pi(X_{n+1})) < Y_{n+1}(0) \bigr) \leq \alpha.
\end{equation} 
This safety guarantee~\eqref{eq:def_safety_intro} means the probability of the realized outcome being worse than the ``status-quo'' outcome under control is no greater than $\alpha$, thereby limiting the risk of harming the new individual. When this policy is implemented on $m$ new individuals,~\eqref{eq:def_safety_intro} implies that the expected number of units harmed by the treatment is no greater than $m \alpha$.  

Such guarantees are important for two related reasons. First, negative treatment effects are substantively costly in many settings. Second, treatment rules are increasingly embedded in (semi-)automated decision systems: once a learned rule is deployed, it may assign interventions repeatedly and at scale. In such settings, trust in the system requires explicit control of the probability that it harms the individuals it chooses to treat. 
We highlight three applications that motivate such explicit harm-rate control.

\vspace{-0.5em}
\paragraph{Example 1 (AI-powered Interventions):} Generative AI is emerging as a new class of intervention in the social sciences, with applications designed to change attitudes and behaviors through scalable, personalized interactions \citep[e.g.,][]{costello2024durably, bai2025llm}. At the same time, recent empirical studies highlight an important risk: while such AI interventions may benefit many individuals and tasks, they may also harm others. For example, a randomized experiment in a global consulting firm \citep{dell2023navigating} found that access to AI can harm the productivity of high-skilled consultants when working on challenging intellectual tasks. Similarly, \cite{bastani2025generative} found that generative AI without guardrails can harm the learning of high school math students. As interventions are powered by black-box AI, controlling the harm is fundamental for safety, public trust, and efficient deployment of AI-powered treatment.

\vspace{-0.5em}
\paragraph{Example 2 (Precision Medicine):} Precision medicine---prevention and treatment strategies that account for individual variability---is central in modern medicine \citep{kosorok2019precision}. The problem of controlling individual risk while maximizing benefit has long been recognized. %. Safety concerns for medications often arise, since 
It is especially relevant when efficacious medications may also lead to a higher risk for certain individuals, such as opioid treatment of chronic pain~\citep{laber2018identifying} and type-2 diabetes with insulin therapies~\citep{wang2018learning}.  % estimate the policy learning algorithm for treating type 2 diabetes patients with insulin therapies. 

\vspace{-0.5em}
\paragraph{Example 3 (Criminal Justice):} In the US criminal justice system, how to safely use risk scores to evaluate an individual's likelihood of reoffending and identify their criminogenic needs is a major topic of interest \citep{skeem2020impact}. For example, \cite{ben2025safe} analyze a field experiment to estimate the causal effect of algorithmic recommendations on judges' decisions at a criminal first appearance hearing. Here, a treatment rule with the safety guarantee can  prevent arrestees from committing a new crime or failing to appear in court, while avoiding unnecessarily harsh decisions.

\vspace{0.5em}
 
We aim to achieve the safety guarantee in a model-agnostic fashion---meaning that it holds without strong modeling assumptions on the data distribution or the learning algorithms---and tightly in finite samples, so this framework is widely applicable to various high-stakes settings. 

% \vspace{0.25em}
\subsection{Overview of contributions}

We develop \textit{conformal policy learning} (CPL) to achieve the safety guarantee~\eqref{eq:def_safety_intro}. 
Distinct from standard policy learning methods that maximizes empirical welfare within a  policy class, our starting point is to view each treatment decision as testing a counterfactual harm event $Y_{n+1}(1) < Y_{n+1}(0)$. A p-value \(p_{n+1}\) satisfying
\[
  \mathbb P\bigl(Y_{n+1}(1)<Y_{n+1}(0),\,p_{n+1}\leq \alpha\bigr)\leq \alpha
\]
yields the desired safety guarantee by   assigning treatment only when \(p_{n+1}\leq\alpha\). 
This connects safe policy learning to conformal inference~\citep{vovk2005algorithmic, lei2021conformal, jin2023selection}. 

The key challenge here---which necessitates novel constructions of powerful conformal p-values---is that the ``harm'' event involves both potential outcomes and is therefore unobserved even for the labeled data. 
CPL addresses this difficulty through two techniques. First, it uses observable proxy labels to construct valid conformal p-values.  Second, to reduce conservativeness, it selects the labeled observations used in p-value calibration according to which treatment arm provides a sharper proxy for harm. We call this strategy ``selective calibration.''
The resulting p-value compares the test conformity score with the selected calibration scores computed using the proxy labels. The CPL workflow is in Figure~\ref{fig:intro}.

\begin{figure}
    \centering
    \includegraphics[width=\linewidth]{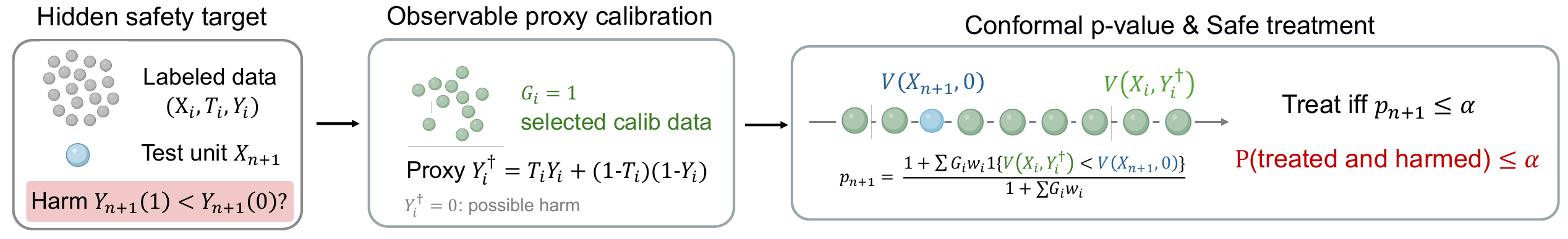}
    \caption{\textbf{Workflow of Conformal Policy Learning (CPL).} Given labeled data $\{(X_i,T_i,Y_i)\}_{i=1}^n$, the goal is to assign safe treatment for a new test unit $X_{n+1}$ with harm rate target $\alpha\in(0,1)$. CPL constructs proxy label $Y_i^\dagger$ and selects a subset of data for calibration. It then computes a p-value by comparing a test conformity score to the calibration scores. The treatment is determined by whether the conformal p-value is below $\alpha$, which controls the harm rate below $\alpha$.}
    \label{fig:intro}
\end{figure}

% Theoretical Properties 
For randomized experiments, CPL provides finite-sample, distribution-free safety guarantees. Under the exchangeability conditions ensured by random assignment, we show that CPL controls the harm rate exactly below $\alpha$ in finite samples, without making any assumption about the outcome models used in the conformal p-values and the selective calibration rule. 
We first establish this result for balanced randomized experiments and then extend it to stratified experiments using weighted conformal inference~\citep{TibshiraniBCR19}. 

We next study the sharpness and optimality of CPL. Because individual-level harm, \(Y(1)<Y(0)\), is not identifiable from observed data, we formulate the optimality under partial identification~\citep{kallus2022s,li2023trustworthy}. We characterize the optimal policy that maximizes the power (probability of treatment) and average welfare subject to the safety constraint for every joint potential-outcome distribution compatible with the observed data. 
Then, we show that as long as the score function used in the p-values and selective calibration rule converge to oracle ones, CPL attains these global optima asymptotically. 

We further extend CPL to observational studies. In this setting, the treatment assignment mechanism is unknown, and the calibration weights must be estimated.  
We develop a learn-then-balance procedure that estimates these weights and enforces balance on functions tailored to the thresholding decisions and the partially identified harm target. Under suitable regularity conditions, the resulting policy has a doubly robust asymptotic safety guarantee: the excess harm rate vanishes when either the propensity-score model or the outcome models, but not necessarily both, are consistently estimated. 
Moreover, if both components converge at standard slow nonparametric rates, the excess harm rate is of a parametric order. 
With consistent outcome models, CPL with observational data attains the same welfare and power optima as in the randomized case.  
Our excess harm rate bound does not pay the price of policy class complexity common in policy learning.  

Finally, we evaluate CPL through simulations and an empirical study. In simulations across randomized experiments, stratified experiments, and observational studies, CPL shows robust harm-rate control and high power and welfare. In the empirical study,  CPL assigns an AI-powered intervention designed to reduce conspiracy beliefs while tightly controlling the harm rate. 

The rest of the paper is organized as follows. 
Section~\ref{sec:setup} introduces the problem setup and connects the safety guarantee to conformal hypothesis testing. 
Section~\ref{sec:rct} develops CPL for balanced randomized experiments and studies its finite-sample validity and asymptotic optimality; the framework is then extended to stratified experiments in 
Section~\ref{sec:stratified} and observational studies in Section~\ref{sec:obs}. 
Section~\ref{sec:sim} presents simulation studies, and Section~\ref{sec:real} applies CPL to an empirical study of AI-powered interventions. 
We close the paper with a discussion on extensions and future directions in Section~\ref{sec:ext}.

\subsection{Related Work}
%% Policy Learning in general 

This article lies in the intersection of causal inference, policy learning, and conformal inference. We summarize several important lines of related work below. 

% \noindent\textbf{Policy learning with safety considerations.}
Our work is connected to the established literature on policy learning \citep[e.g.,][]{murphy2003optimal, manski2004statistical, hirano2009asymptotics, zhao2012estimating, kitagawa2018should, athey2021policy, jin2025policy}, which aims to select, among a given class of policies, the one that maximizes an objective (such as average welfare) while optionally respecting a constraint (such as harm rate). 
Within this literature, this work is closely related to a small but emerging literature on policy learning with safety considerations (where the exact meaning of safety varies). One line of work develops safe policy learning algorithms in settings that necessitate extrapolating beyond the observed labeled data, and the safety refers to not performing worse than a ``status quo'' policy~\citep{zhang2022safe, ben2025safe, jia2025bayesian, wu2025safe}. Our safety notion of individual harm is conceptually related since it measures the harm relative to the status quo of no treatment, but distinct enough to yield completely different techniques.
In addition, \cite{ben2024policy} studies policy learning when the objective involves counterfactuals, providing doubly robust algorithms and regret bounds for the learned policy; the dual form of a special case in their framework (with an unknown Lagrange parameter) coincides with the welfare maximization problem subject to harm rate control; this connects with our setting and the method of~\cite{li2023trustworthy}. 
As standard in policy learning, these methods select a policy that maximizes an empirical objective within a policy class, whose performance is often measured by the regret (the gap between the true objective of the learned policy from the optimal), which is typically bounded by a statistical error term that scales with the complexity of the policy class (such as the VC-dimension). In contrast, CPL leverages conformal inference to achieve finite-sample, distribution-free safety guarantee (without a high-probability excess error bound term) when propensity scores are known; moreover, when propensity scores are unknown and estimated so inexact harm control is inevitable, our excess harm rate bound does not pay the price of the policy class complexity.
 
Our safety notion follows from a literature in causal inference and policy learning that bounds, estimates, and controls the same harm rate notion as~\eqref{eq:safety_guarantee}. A line of work studies its bounds under various  assumptions such as monotonicity~\citep{huang2012assessing} and certain conditional independence conditions~\citep{shen2013treatment,yin2018assessing}. Relatedly, due to the non-identifiability, several work focuses on establishing bounds on the conditional harm rate  for binary outcomes, including~\cite{gadbury2004individual} without covariates, \cite{zhang2013assessing} with covariates, \cite{kallus2022s} on the sharp identification bounds, and~\cite{wu2024quantifying} using a sensitivity model for the correlation between the potential outcomes. In addition, a recent independent work of~\cite{scauda2026counterfactual} studies the population-level optimal welfare subject to harm rate control. Since we aim to control the harm rate without strong modeling assumptions, this work is implicitly tied to the Frech\'et--Hoeffding bound characterized in~\cite{zhang2013assessing,kallus2022s}. 
We show CPL attains the optimal welfare subject to safety guarantee among the worst-case distributions in these work.

Our method builds on the conformal p-values proposed in~\cite{jin2023selection} for i.i.d.~data and~\cite{jin2023model} for covariate shift settings, which were extended to model selection~\citep{bai2024optimized} and online settings~\citep{xu2024online}. In that literature, the p-values quantify the confidence in a large, \emph{ordinary} outcome (i.e., no potential outcomes) exceeding a \emph{known} threshold. 
While we borrow the high-level intuitions, the technical route in constructing such p-values for our problem---which is the key contribution here---is sharply distinct since the two outcomes under comparison are never simultaneously observed. The resulting optimality and robustness properties are likewise quite different. 

The conformal inference approach also connects our work with a line of work on conformal inference for individual treatment effects (ITE) $Y_{n+1}(1)-Y_{n+1}(0)$~\citep{lei2021conformal,jin2023sensitivity,yin2024conformal}. These work typically focuses on constructing a prediction set for the counterfactual outcome for a unit who has already received treatment or control, where inference for the ITE of a new test point with two unknown outcomes appears particularly challenging. The latter case (with binary outcomes) is the setting we address, and we study the ``decision'' problem rather than prediction set construction. This involves a thresholding decision rule which necessitates distinct calibration and theoretical analysis techniques.

% !TEX root = main.tex

\section{Problem Setup and Conceptual Framework}
\label{sec:setup}

\subsection{Problem Setup}
\label{subsec:setup}

We assume access to a set of labeled data $\{(X_i,T_i,Y_i)\}_{i=1}^n$, where $X_i\in \cX$ is the feature, $T_i\in\{0,1\}$ is the treatment assignment, and $Y_i\in \{0,1\}$ is the binary outcome. We define the potential outcomes $Y_i(t)$ for $t \in \{0, 1\}$ and assume the triplets $\{(X_i,Y_i(1),Y_i(0))\}_{i=1}^n$ are independent and identically distributed (i.i.d.) from an unknown super-population $\PP_{X,Y(1),Y(0)}$. 
Assuming the Stable Unit Treatment Value Assumption (SUTVA) \citep{rubin1980randomization}, the observed outcome is given by $Y_i=Y_i(T_i)$. The joint super-population $\PP_{X,Y(1),Y(0)}$  
is unidentifiable from data because researchers observe only one potential outcome for each unit~\citep{holland1986statistics}. 
Finally, the distribution of the labeled data $\{(X_i,T_i,Y_i)\}_{i=1}^n$ is induced by the unknown super-population and the treatment assignment mechanism and denoted as $\PP_{X,Y,T}$.

Throughout, we make the unconfoundedness assumption for the treatment assignment mechanism, which is standard in the causal inference literature~\citep{imbens2015causal}. 
\begin{assumption}[Unconfoundedness]\label{assump:uncon}
    The treatment assignments $\{T_i\}_{i=1}^n$ are mutually independent and obey $T_i\indep (Y_i(0),Y_i(1))\given X_i$, with  $\PP(T_i=1\given X_i=x) \coloneq e(x)$, and we call $e(x)\in (0,1)$ the propensity score. 
\end{assumption}
Our framework covers both randomized experiments (Sections~\ref{sec:rct} and~\ref{sec:stratified}) and observational studies (Section~\ref{sec:obs}). In randomized experiments, Assumption~\ref{assump:uncon} is satisfied by design and $e(x)$ is known. In observational studies, Assumption~\ref{assump:uncon} should be evaluated based on domain knowledge; even though it holds, the propensity score $e(x)$ is unknown and needs to be estimated.

Policy learning uses the labeled data to determine the treatment for a new individual (which we call the test unit/point) with observed feature $X_{n+1}$ and unobserved potential outcomes $Y_{n+1}(1)$ and $Y_{n+1}(0)$. 
Following the literature~\citep[e.g.,][]{manski2004statistical, hirano2009asymptotics, zhao2012estimating, kitagawa2018should, athey2021policy}, we assume it is from the same super-population independently of the labeled data.

\begin{assumption}[IID]\label{assump:iid}
    The test point is drawn from $(X_{n+1},Y_{n+1}(1),Y_{n+1}(0))\sim \PP_{X,Y(1),Y(0)}$  independently of the labeled data. 
\end{assumption}
% Researchers only observe $X_{n+1}$ and do not observe either $Y_{n+1}(1)$ or $Y_{n+1}(0)$. 
Our framework can be naturally extended to relax the i.i.d.~assumption and  allow for covariate shift between the labeled data $i \in \{1, \ldots, n\}$ and the test point; see Section~\ref{sec:ext} for a discussion.

\subsection{Harm Rate}
\label{subsec:harm_rate_discuss}

Recall that our goal is to develop a policy learning algorithm producing a rule $\hat\pi: \cX \rightarrow \{0,1\}$ that maps features $X$ to treatment assignment $\hat\pi(X)$ with the \textit{counterfactual} safety guarantee:
\@\label{eq:safety_guarantee}
\PP\big(Y_{n+1}(\hat\pi(X_{n+1})) < Y_{n+1}(0)\big) \leq \alpha 
\@
for a pre-specified confidence level $\alpha\in (0,1)$. 
Here, the probability is over the labeled data (based on which $\hat\pi$ is built) and the test point. We call the left-handed side probability in~\eqref{eq:safety_guarantee} the ``harm rate'' following~\cite{zhang2013assessing}; the same quantity is also referred to as the ``fraction of negatively affected'' in the literature~\citep{kallus2022s,li2023trustworthy,wu2024quantifying}. 
Throughout, we focus on binary outcomes; this setting is most extensively studied~\citep{zhang2013assessing,kallus2022s,wu2024quantifying}.

The safety guarantees can be rewritten as
\@ 
\PP\big(Y_{n+1}(\hat\pi(X_{n+1})) < Y_{n+1}(0)\big) \leq \alpha ~~ \Longleftrightarrow ~~ \PP\big(Y_{n+1}(1) < Y_{n+1}(0), ~\hat\pi(X_{n+1}) = 1 \big) \leq \alpha. \label{eq:eqivalence}
\@
Thus, harm occurs when the individual treatment effect on the test unit is negative and one decides to treat the unit. Thus, a trivial way to achieve this is to treat the test unit with probability $\alpha$; however, such a policy is clearly suboptimal due to limited welfare and power (whose meaning will be made precise later). Instead, we aim for a procedure that achieves this with reasonable, and sometimes optimal, welfare and power. 
Meanwhile, the above result shows that to control the harm rate, we should detect and avoid treating units that can be harmed by the treatment. This is the intuition for CPL.

\subsection{Connecting Safety Guarantees to Hypothesis Testing}

Our technical route differs from standard policy learning approaches, so it is helpful to begin with the conceptual framework: connect the safety guarantees with hypothesis testing. % We shall connect to their methodology later on.

Equation~\eqref{eq:eqivalence} suggest that to make safe treatments, one needs to identify units that are likely \emph{not to} be harmed by the treatment, i.e., $Y_{n+1}(1)\geq Y_{n+1}(0)$. 
We view this as testing a random null hypothesis $H_0\colon Y_{n+1}(1) < Y_{n+1}(0)$.  
Suppose we can construct a p-value $p_{n+1}$ such that  
\@\label{eq:p-value}
\PP\big(Y_{n+1}(1) < Y_{n+1}(0), ~ p_{n+1} \leq \alpha \big) \leq \alpha. 
\@
This is a non-conventional definition of p-values because the truth of $H_0$ is itself random, but this notion is sufficient for the desired safety guarantee:  
setting $\hat\pi(X_{n+1}) = \ind\{p_{n+1} \leq \alpha\}$, the validity~\eqref{eq:p-value} directly implies the safety guarantee~\eqref{eq:safety_guarantee} via the equivalence relationship in~\eqref{eq:eqivalence}. Intuitively,  
a small p-value informs strong evidence against the null, i.e., the test unit is unlikely to be harmed and thus can safely receive the treatment.

The remaining of the paper focuses on addressing two questions. First, how can we construct p-values that satisfy~\eqref{eq:p-value} in finite sample without making strong modeling assumptions? We address this question by expanding the recent literature of conformal p-values \citep{vovk2005algorithmic, bates2021testing, jin2023selection} to causal inference contexts~\citep{imbens2015causal,lei2021conformal}, while respecting the partial-identification nature of
the individual-level causal effects~\citep[e.g.,][]{heckman1997making}. Second, is thresholding p-values optimal in any sense? This is important as it is not clear a priori why our approach may be desired, even if it might achieve the safety guarantee. We will show that with specifically designed p-values, our method achieves optimal power and welfare among all safe policies at the population level.

% !TEX root = main.tex

\section{Conformal Policy Learning with Balanced Randomized Experiments}
\label{sec:rct}
To fix ideas, in this section, we introduce our framework in the simple setting of balanced randomized experiments with $e(x)=1/2$; this will be gradually generalized in Sections~\ref{sec:stratified} and~\ref{sec:obs}.  
In Section~\ref{subsec:pre}, we start with a single-arm construction that provides distribution-free safety guarantees. In Section~\ref{subsec:cpl_rct}, we introduce observable proxies and selective calibration that sharpens the procedure while preserving the same guarantee. In Section~\ref{subsec:rct_optimality}, we characterize the welfare-optimal policy under partial identification, followed by an explanation on the sharpness of selective calibration in Section~\ref{subsec:discuss_sharp}.

Throughout this section, all learned functions are fitted on a training sample independent of the labeled observations used for calibration and the test unit. In practice, this can be achieved by sample splitting~\citep{lei2018distribution}. For notational simplicity, we use \(\{(X_i,T_i,Y_i)\}_{i=1}^n\) to denote the labeled observations available for calibration; selective calibration introduced in Section~\ref{subsec:cpl_rct} may retain only a subset of these observations. Conditional on the training sample, the learned functions are treated as fixed.

\subsection{Warm-up: Direct Application of Existing Conformal p-values}
\label{subsec:pre}

We begin with a simple adaptation of the existing idea to show how conformal p-values can serve as the foundation for safe policy learning. Given that we focus on binary outcomes, we can rewrite~\eqref{eq:p-value} as finding a p-value $p_{n+1}$ obeying
\@\label{eq:safety_equiv_bin}
\PP\big( Y_{n+1}(1) = 0, Y_{n+1}(0) = 1,~ p_{n+1} \leq \alpha  \big) \leq \alpha.
\@
That is, we would like to find p-values that quantify the confidence in a small treated outcome \emph{and} a large control outcome. 
The conformal selection (CS) framework~\citep{jin2023selection} provides a natural starting point. Given labeled data $\{(X_i,Y_i)\}_{i=1}^n$ for an ordinary outcome $Y\in \RR$ (i.e., no potential-outcome framing) and a known threshold $c\in \RR$, CS proposes conformal p-values with a similar null property $\PP(Y_{n+1}\leq c, p_{n+1}\leq \alpha)\leq \alpha$. 
Setting aside the issue that we now have two potential outcomes to deal with, one simple idea is to 
% Because the two potential outcomes are never observed simultaneously, 
% one simple idea is to 
develop a conservative p-value $p^{(1)}_{n+1}$ that satisfy
\@\label{eq:safety_equiv_bin_1}
\PP\big(Y_{n+1}(1) = 0, ~ p^{(1)}_{n+1} \leq \alpha  \big) \leq \alpha,
\@
which directly implies~\eqref{eq:safety_equiv_bin}. For treated outcome, it means we will use the treated units for calibrating p-values. 
The conformal p-value from CS takes the form  
\@\label{eq:cs_pval}
p_{n+1}^{(1)} = \frac{1+\sum_{i=1}^n T_i  \cdot \ind\{V(X_i,Y_i)\leq V(X_{n+1},0)\}}{1 + \sum_{i=1}^n T_i},
\@ 
where $V$ is any score function $V\colon \cX\times\cY\to \RR$ obeying $V(x,y)\leq V(x,y')$ whenever $y\leq y'$ for any $x\in \cX$. 

\paragraph{Example (clipped score function).}
One example is the clipped score function~\citep{jin2023selection}, $V(x,y) = M\ind\{y>0\} + (1 - \widehat{\mu}_1(x))$, where $\widehat{\mu}_1(x)$ is an estimator for the conditional expectation function $\E[Y_i \mid T_i = 1, X_i = x]$ and $M>2\sup_x|\hat\mu_1(x)|$ is a sufficiently large constant so that~\eqref{eq:cs_pval} reduces to 
\@\label{eq:cs_pval_clip}
p_{n+1}^{(1)} = \frac{1+\sum_{i=1}^n T_i  \cdot \ind\{Y_i= 0\} \cdot \ind\{ \widehat{\mu}_1(X_i)\geq  \widehat{\mu}_1(X_{n+1})\}}{1 + \sum_{i=1}^n T_i}.
\@ 
Intuitively, this p-value focuses on potentially unsafe cases (i.e., $Y_i(1) = 0$) among treated units and examines whether the test unit is extreme with respect to $\widehat{\mu}_1(x)$. When this p-value is small, the test unit is unlikely to come from the distribution of potentially unsafe units, and thus is likely to be safe. \qed  \\

To see why the conformal p-value~\eqref{eq:cs_pval} satisfies~\eqref{eq:safety_equiv_bin_1} for any score function, 
we outline the theoretical arguments from~\cite{jin2023selection}. 
Let $n_1 = |\cI_1|$ be the number of treated units, where $\cI_1 = \{i\in [n]\colon T_i=1\}$.  
Consider the ``oracle'' p-value  
\@\label{eq:cs_pval_oracle}
p_{n+1}^{\ast (1)} = \frac{1+\sum_{i=1}^n T_i  \cdot \ind\{V(X_i,Y_i(1))\leq V(X_{n+1},Y_{n+1}(1))\}}{1 + n_1} 
% =  \frac{1+\sum_{i\in \cI_1}   \ind\{V(X_i,Y_i(1))\leq V(X_{n+1},Y_{n+1}(1))\}}{n_1+1},
\@ 
which is not computable since $Y_{n+1}(1)$ is unobserved. The only difference between~\eqref{eq:cs_pval} and~\eqref{eq:cs_pval_oracle} is that $Y_{n+1}(1)$ in the oracle p-value is replaced by $0$, and we have $Y_i(1)=Y_i$ for $i\in \cI_1$ in~\eqref{eq:cs_pval_oracle}. First, conditional on $\{T_i\}_{i=1}^n$, in this randomized experiment, the data $\{(X_i,Y_i(1)\}_{i\in \cI_1} \cup\{X_{n+1},Y_{n+1}(1)\}$ are exchangeable, which implies $\PP(p_{n+1}^{*(1)}\leq \alpha)\leq \alpha$~\citep{vovk2005algorithmic}. 
Second, on the event that $Y_{n+1}(1)=0$, the two p-values coincide. 
The two facts imply $\PP(p_{n+1}^{(1)}\leq \alpha,~Y_{n+1}(1)=0) = \PP(p_{n+1}^{*(1)}\leq \alpha,~Y_{n+1}(1)=0)\leq \PP(p_{n+1}^{*(1)}\leq \alpha) \leq \alpha$.

Of course, one may also leverage the control samples and consider the null hypotheses $H_0^{(0)}\colon Y_{n+1}(0)=1$. We can similarly construct the p-value (using a monotone function $V$ as before)
\@\label{eq:cs_pval_0}
p_{n+1}^{(0)} = \frac{1+\sum_{i=1}^n \ind\{T_i = 0\} \cdot \ind\{V(X_i,1-Y_i)\leq V(X_{n+1},0)\}}{1 + \sum_{i=1}^n \ind\{T_i = 0\}}.
\@ 
Following exactly the same rationales as above, this p-value is valid in the same sense.

\subsection{Proposed Method: Conformal Policy Learning}
\label{subsec:cpl_rct}

While both p-values in the last section are feasible and valid, they are conservative since $H_0^{(1)}: Y_{n+1}(1) = 0$ and $H_0^{(0)}: Y_{n+1}(0) = 1$ are both strict implications of $H_0: Y_{n+1}(1) = 0, Y_{n+1}(0) = 1$. 
In this section, we propose the general CPL procedure that improves and  
subsumes the previous two options as special cases.  

Following the conformal inference literature, we call the subset of the labeled data used to compute the p-values the ``\textit{calibration data}''. Then, both p-values in Section~\ref{subsec:pre} only use either the treatment or the control group as the calibration data, which can be substantially generalized.   

First, consider any inclusion function $\hat{g} \colon \cX\times\{0,1\}\to [0,1]$ whose training process is independent of the labeled data and the test point.  
Conditional on data, we  draw independent inclusion indicators $G_i\sim \textrm{Bernoulli}(\hat{g}(X_i,T_i))$ for each $i\in [n]$, which defines the calibration set
\@\label{eq:def_Dcal}
\cD_{\calib} = \{(X_i,T_i,Y_i)\}_{i\in \cI_\calib},\quad \text{where} \quad \cI_{\calib} = \{i\in [n]\colon G_i=1\}.
\@
The inclusion into $\cD_\calib$ can depend both on $T_i$ and $X_i$, and is allowed to be random, although later we shall see a binary $\hat{g}$ suffices for optimality. 
Since not all labeled data are used in calibration, we call this technique ``selective calibration''.

Second, to allow involving data from both arms in calibration, we combine information from the two treatment arms. Define the computable proxy outcome $Y_i^\dagger = T_iY_i+(1-T_i)(1-Y_i)$, so $Y_i^\dagger=0$ whenever $Y_i=0$ for a treated sample or $Y_i=1$ for a control sample. 
Generalizing the previous ideas, $Y_i^\dagger=0$ captures potentially unsafe units (i.e., units for which $H_0$ might happen).

With the two techniques in hand, we construct the p-value 
\@\label{eq:pval_general}
p_{n+1} = \frac{1+ \sum_{i=1}^n G_i \times \ind\{V(X_i,Y_i^\dagger)\leq V(X_{n+1},0)\}}{1+\sum_{i=1}^n G_i}.
\@  
This covers the single-arm p-values in Section~\ref{subsec:pre}: it recovers $p_{n+1}^{(1)}$ (resp.~$p_{n+1}^{(0)}$)
when $G_i = T_i$ (resp.~$G_i = 1 - T_i$). 
Another simple case is to take $G_i = 1$ which includes all the labeled data as the calibration set.  

With this p-value, our conformal policy learning simply produces a thresholding rule $\hat\pi(X_{n+1}) \coloneqq \ind\{p_{n+1} \leq \alpha\}$. Algorithm~\ref{algo:rct} summarizes the procedure. 

% \yingcomment{do we want formal algorithms for later cases?} \nec{How much formal are you thinking about?}

\begin{algorithm}
    \small
    \captionsetup{font=small}
    \caption{Conformal Policy Learning (Randomized Experiments)}
    \label{algo:rct}
\begin{algorithmic}[1]
    \REQUIRE{Labeled data $\{(X_i, T_i,Y_i)\}_{i=1}^n$, test data $X_{n+1}$, target level $\alpha \in (0,1)$, score $V(\cdot, \cdot)$, inclusion function $\hat{g}(\cdot,\cdot)$.} \\[0.5ex]
    \STATE Draw independent inclusion indicators $G_i \sim \text{Bern}(\hat{g}(X_i,T_i))$ for $i=1,\dots,n$.
    \STATE Compute $Y_i^\dagger = T_iY_i+(1-T_i)(1-Y_i)$ for $i=1,\dots,n$.
    \STATE Compute p-value $p_{n+1}$ via~\eqref{eq:pval_general}.
    \STATE Compute  $T_{n+1}=\ind\{p_{n+1}\leq \alpha\}$.  
    \ENSURE{Safe treatment $T_{n+1}$.}
\end{algorithmic}
\end{algorithm}

% \subsection{Finite-sample safety guarantee}

Theorem~\ref{thm:validity_rct} establishes the distribution-free safety guarantee for any score and inclusion functions whose inclusion probability in two groups sums up to a constant. The proof is in Appendix~\ref{app:subsec_validity_rct}. 

\begin{theorem}\label{thm:validity_rct}
    Assume Assumption~\ref{assump:uncon} with $e(x) \equiv 0.5$ and Assumption~\ref{assump:iid}.
    Then, for any score function $V \colon \cX\times\cY \to \RR$ that is non-decreasing in the second argument and $\hat{g}\colon \cX\times\{0,1\}\to [0,1]$ whose training process is independent of $\{(X_i,T_i,Y_i)\}_{i=1}^n\cup \{X_{n+1}\}$, and obey $\hat{g}(x,1)+\hat{g}(x,0)=a$ for a constant $a>0$ for any $x\in \cX$, it holds that $\PP(Y_{n+1}(1) = 0, Y_{n+1}(0) = 1, p_{n+1} \leq \alpha) \leq \alpha$ for any $\alpha\in (0,1)$. That is, the finite-sample safety guarantee~\eqref{eq:safety_guarantee} holds for  $\hat\pi(X_{n+1}) =\ind\{p_{n+1}\leq \alpha\}$. 
\end{theorem}
Theorem~\ref{thm:validity_rct} states that, in principle, users can use any score function $V(\cdot, \cdot)$ and any inclusion function $\hat{g}(\cdot,\cdot)$, and CPL (Algorithm~\ref{algo:rct}) guarantees a controlled harm rate in finite samples. We emphasize that  the two simplifying assumptions, namely $e(x) = 1/2$ and $\widehat{g}(x,1) + \widehat{g}(x,0) = a$, are used only to streamline presentation, and will be eliminated in Section~\ref{sec:stratified}. 

This theorem  covers the special cases we discussed in Section~\ref{subsec:pre}. The one based on the treated (resp.~control) samples correspond to $\hat{g}(x,t) = t$  (resp.~$\hat{g}(x,t) = 1-t$),  both obeying the constant-sum condition with $a = 1$. As a third example, the conformal p-value  using all the labeled data corresponds to $\hat{g}(x,t) \equiv 1$ and thus it also satisfies the constant-sum condition with $a = 2$. 
Finally, while~\eqref{eq:pval_general} represents a restrictive class of policies, we shall see in the next subsection that specific choices of   $V$ and $\widehat{g}$ lead  to asymptotically optimal safe policy learning algorithm among all safe policies (i.e., beyond the policy class specified by CPL). 

The detailed proof of Theorem~\ref{thm:validity_rct}  is in Appendix~\ref{app:subsec_validity_rct}, and we provide a sketch here to clarify the intuitions. We again consider the ``oracle'' p-value 
\@\label{eq:cpl_pval_oracle}
p_{n+1}^{\ast} = \frac{1+\sum_{i=1}^n G_i \times \ind\{V(X_i,Y_i^\ast) \leq V(X_{n+1},Y_{n+1}^{\ast}\}}{1+\sum_{i=1}^n G_i},
\@ 
where $Y_{i}^\ast = \max\{Y_{i}(1),1-Y_{i}(0)\}$ so that $Y_{i}^\ast = 0$ if and only if $Y_{i}(1) = 0$ and $Y_{i}(0) = 1$, i.e., unit $i$ is harmed. This is again not computable because $Y_{i}^\ast$ is \emph{never} observed even for the labeled data. 
The first key fact is that this oracle is valid as long as $\hat{g}$ satisfies the constant-sum condition. Indeed, we can show that in balanced randomized experiments, the data $\{(X_i, Y_{i}^\ast)\}_{G_i = 1}$  and $(X_{n+1}, Y_{n+1}^\ast)$ are exchangeable conditional on $\{G_i\}_{i=1}^n$, which implies $\PP(p_{n+1}^\ast\leq \alpha\given \{G_i\}_{i=1}^n)\leq \alpha$. 
Second, the observed p-value~\eqref{eq:pval_general} replaces   $Y_{n+1}^\ast$ in~\eqref{eq:cpl_pval_oracle}   with the ``null'' value  $0$ and replaces the calibration label $Y_i^\ast$ with $Y_i^\dagger$, thereby upper bounding the oracle one on the ``unsafe'' null event. Specifically, the proxy outcome is conservative in the sense that $Y_i^\ast\geq Y_i^\dagger$; thus, whenever $Y_{n+1}^*=0$,   the monotonicity of $V$ implies $p_{n+1}\geq p_{n+1}^\ast$. 
This gives $\PP(Y_{n+1}^\ast =0, p_{n+1}\leq \alpha\given \{G_i\}_{i=1}^n) \leq \PP(Y_{n+1}^\ast =0, p^\ast_{n+1}\leq \alpha\given \{G_i\}_{i=1}^n) \leq \PP(p^\ast_{n+1}\leq \alpha\given \{G_i\}_{i=1}^n) \leq \alpha$.

\begin{remark}[Sharpness of conformal p-values]\label{rmk:conservative}
Our conformal p-value is subject to two sources of conservativeness compared with $p_{n+1}^*$: (i) the ``oracle'' labels $Y_i^*$ are replaced by the proxy labels $Y_i^\dagger$, and (ii) the test label $Y_{n+1}^*$ is replaced by the threshold $0$. As in conformal selection~\citep{jin2023selection}, the issue (ii) will be addressed by a tailored ``clipped'' score function. 
On the other hand, (i) is rooted in the partial identification nature of the harm rate, and specific choices of $\hat{g}$ makes our p-value sharp; we shall discuss this in more details once the optimality results in the next section are ready. 
\end{remark}

\subsection{Optimality of Conformal Policy Learning}
\label{subsec:rct_optimality}

Having established the finite-sample safety guarantees for CPL, we proceed to study the optimality.
We first analyze the optimal policy among all safe policies (i.e., including policies beyond our framework)  under partial identification. We then show that CPL with specific choices of the score and inclusion functions asymptotically achieves  global optimality.

To formalize the discussion, we define some additional notations. 
The unknown, true joint distribution over $(X,Y(1),Y(0))$ is denoted as $\PP_{X,Y(1),Y(0)}$, which induces the observable distribution  $(X_i,Y_i,T_i)\sim \PP_{X,Y,T}$. 
Also, we  denote the collection of distributions $P$ over $(X,Y(1),Y(0))$ that are compatible with the observable distribution $\PP_{X,Y,T}$ as 
\$
\cP = \big\{P_{X,Y(1),Y(0)}\colon P_{X,Y(T),T}=\PP_{X,Y(T),T} ~\text{for}~ P(T=1\given X,Y(1),Y(0))= e(X)\big\},
\$
that is, the induced distribution of $(X,Y,T)$ under $P$ coincides with $\PP_{X,Y,T}$ on the observables.

For any distribution $P$ and any treatment assignment rule $\pi\colon \cX\to \{0,1\}$, we write the harm rate as
\$
\text{Err}(\pi;P) := P( Y_{n+1}(\pi(X_{n+1}))< Y_{n+1}(0)) = P(Y_{n+1}(1)=0, Y_{n+1}(0)=1, \pi(X_{n+1})=1).
\$
By the tower property, letting $\EE_P$ denote the expectation under $P$, we know 
\$
\text{Err}(\pi;P) = \EE_P\big[  \pi(X_{n+1}) \cdot  f_P(X_{n+1})\big] \leq \alpha,\quad \text{where}\quad f_P(x) := P(Y_{n+1}(1) < Y_{n+1}(0)\given X_{n+1}=x).
\$
Here we call $f(x)$ the \emph{harm rate function}. 
Similar to the scalar marginal harm rate $\text{Err}(\pi,P)$, the function $f_P(\cdot)$ is not identifiable from data because $Y(1)$ and $Y(0)$ are never simultaneously observed. 
\cite{kallus2022s} derived the sharp partial-identification upper bound for $f_P(x)$,  given by 
\@\label{eq:def_bd_gamma}
\gamma(x):= \min\{1-\mu_1(x),\mu_0(x)\},
\@
where $\mu_t(x) = \EE[Y(t)\given X=x]$ for $t\in \{0,1\}$ is the conditional mean function of each potential outcome. Namely, for any super-population $P_{X,Y(1),Y(0)}$ whose observed distribution $P_{X,Y,T}$ (induced by the same treatment assignment mechanism) is equal to $\PP_{X,Y,T}$, its harm rate function obeys $f_P(x)\leq \gamma(x)$, and there exists one such distribution whose harm rate function coincides with $\gamma(x)$.   

The most common objective of policy learning is to maximize the welfare. Following the  policy learning literature~\citep[e.g.,][]{murphy2003optimal, manski2004statistical, hirano2009asymptotics, zhao2012estimating, kitagawa2018should, athey2021policy}, we define the welfare of a policy $\pi\colon \cX\to \{0,1\}$ as 
\$
\text{Welfare}(\pi; \PP) \coloneqq \EE\big[Y_{n+1}(\pi(X_{n+1}))\big],
\$
where the larger value of outcome $Y$ corresponds to the larger welfare. 
The welfare-optimal treatment rule subject to the safety guarantee can be defined as
\@\label{eq:def_opt_welfare_phi}
\pi_\welfare^* = \argmax_{\pi\colon \cX\to \{0,1\}} \quad & \Welfare  (\pi; \PP)  \\ 
\text{subject to}\quad & \max_{P\in \cP} ~\text{Err}(\pi;P) \leq \alpha. \notag
\@
Here, we optimize the welfare subject to the worst-case safety guarantee. The constraint, $\max_{P\in \cP} ~\text{Err}(\pi;P) \leq \alpha$, ensures that, regardless of the underlying data-generating process, the harm rate is upper bounded by $\alpha.$

We define the oracle welfare-optimal score function 
\@\label{eq:s_welfare_optimal}
s_{\welfare}(x)=- (\mu_1(x)-\mu_0(x))/\fna(x),
\@
where we use the convention $0/0=0$, $a/0=+\infty$ if $a>0$ and $a/0=-\infty$ if $a<0$ throughout. 
The following theorem formally gives the optimal solution $\pi^*_{\text{welfare}}$ under the worst-case safety constraint. It is an implication of the more general Theorem~\ref{thm:welfare_optimal_general} in Appendix~\ref{app:subsec_general_opt}.  

\begin{theorem}
\label{thm:welfare_optimal} 
    Assume $s_{\mathrm{welfare}}(X)$ has no point mass. Then an optimal solution to \eqref{eq:def_opt_welfare_phi} is 
    \[ 
    \pi^*_{\mathrm{welfare}}(x) = \mathbf{1}\{s_{\mathrm{welfare}}(x)\leq r^*\}, \quad
    \text{where}\quad r^* := \sup\left\{ \tilde r\leq 0: \mathbb{E}\!\left[ \gamma(X) \mathbf{1}\{s_{\mathrm{welfare}}(X)\leq \tilde r\} \right] \leq \alpha \right\}. 
    \] 
    Moreover, writing $\tau(x)=\mu_1(x)-\mu_0(x)$, if $\mathbb{E}[\gamma(X)\mathbf{1}\{\tau(X)>0\}]\leq\alpha$, then $r^*=0$ and $\pi^*_{\mathrm{welfare}}(x)=\mathbf{1}\{\tau(x)>0\}$ almost surely. Otherwise, $r^*<0$ and $\mathbb{E}[\gamma(X)\pi^*_{\mathrm{welfare}}(X)]=\alpha$.
\end{theorem}
Theorem~\ref{thm:welfare_optimal} shows that the optimal treatment rule for the welfare maximization is based on a cutoff on $s_{\text{welfare}}(X_{n+1})$ while staying in the region $\tau(X_{n+1})\geq 0$. The intuition is as follows. 
We can show that the problem~\eqref{eq:def_opt_welfare_phi} that defines $\pi_{\text{welfare}}^*$   can be rewritten as 
\@\label{eq:def_opt_welfare_phi_2}
\maximize_{\pi \colon \cX\to \{0,1\}} \quad & \E[\pi(X_{n+1})\tau(X_{n+1})] + \text{constant} \\ %\E[Y_{n+1}(0)] \\ 
\text{subject to}\quad & \E[\pi(X_{n+1}) \fna(X_{n+1})] \leq \alpha, \notag
\@
where $\tau(x) = \mu_1(x) - \mu_0(x)$ captures the conditional average treatment effect (CATE).   

The optimal solution thus seeks the feature regions that yield the largest welfare gain $\tau(X)$ relative to each unit of ``cost'' $\gamma(X)$. 

The next question is whether, and under what conditions, the conformal policy learning method can achieve this optimal welfare.  
Following Theorem~\ref{thm:welfare_optimal}, we 
define the conformal p-value
\begin{equation}
\label{eq:pval_welfare_optimal}
p^{\text{welfare}}_{n+1} = \frac{1+ \sum_{i=1}^n G_i \times \ind\{Y_i^\dagger = 0\} \times \ind\{\hat{s}(X_i) \leq \hat{s}(X_{n+1})\}}{1+\sum_{i=1}^n G_i}, 
\end{equation}
for the specific choices 
\@
& G_i = T_i \one\{1 - \widehat{\mu}_1(X_i) \leq \widehat{\mu}_0(X_{i})\} + (1-T_i) \one\{1 - \widehat{\mu}_1(X_i) > \widehat{\mu}_0(X_{i})\}, \label{eq:def_G_welfare_optimal}\\ 
&\text{and}\qquad \hat{s}(x) = - (\widehat{\mu}_1(x) - \widehat{\mu}_0(x))/\widehat{\fna}(x).\label{eq:def_s_welfare_optimal}
\@
This corresponds to taking a clipped score  $V(x,y) = M\ind\{y>0\} + \hat{s}(x)$ with a sufficiently large constant $M > 0$ in Algorithm~\ref{algo:rct}. 
Note that $\hat{g}(x,1) + \hat{g}(x,0) = 1$ and therefore it also satisfies the condition for the inclusion function specified in Theorem~\ref{thm:validity_rct}. The intuition about the expression of the inclusion function $G_i$ is given in the next subsection, where we discuss the sharpness of our results.

Consider the conformal policy learning algorithm with estimated welfare cutoff (if necessary) 
\$
\hat\pi_{\mathrm{welfare}}(X_{n+1}) := \ind\{p^{\mathrm{welfare}}_{n+1}\leq\alpha\} \ind\{\widehat\mu_1(X_{n+1}) > \widehat\mu_0(X_{n+1})\}.
\$ 
Theorem~\ref{thm:asymp_opt_welfare} shows that when conditional expectation functions of potential outcomes are correctly specified, CPL achieves the optimal welfare asymptotically. Its proof is included in Appendix~\ref{app:subsec_opt_welfare_ours}. 

\begin{theorem}\label{thm:asymp_opt_welfare}
    Suppose Assumption~\ref{assump:uncon} holds with $e(x) = 0.5$ and Assumption~\ref{assump:iid} holds. Suppose $\|\hat{\mu}_t(X) -\mu_t(X)\|_{L_2(\PP_{X})}=o_P(1)$ as $n\to \infty$ for $t \in \{0,1\}.$  
Additionally, assume $s_{\welfare}(X)$ has no point mass. 
Then, $\EE[Y(\hat\pi_{\welfare}(X_{n+1}))]\to \Welfare(\pi_\welfare^*;\PP)$ as $n\to \infty$ where $\pi_{\welfare}^*$ is the optimal solution in Theorem~\ref{thm:welfare_optimal}.
\end{theorem}
 
The consistency condition on the estimated conditional expectation functions could be achieved by flexible machine-learning-based methods.
Moreover, even if the consistency condition fails, the conformal policy learning algorithm based on $p^{\text{welfare}}_{n+1}$ still satisfies the finite-sample safety guarantee as per Algorithm~\ref{algo:rct}.  

In Appendix~\ref{app:subsec_power_opt}, we present analogous optimality results when researchers are interested in maximizing the ``power'', the probability of treating the test unit. The only difference is that we should now use a power-oriented score $\hat{s}(x) = \widehat{\fna}(x)$.

\subsection{Optimality, Sharpness, and Selective Calibration}
\label{subsec:discuss_sharp}

The optimality result suggests that, in the non-trivial setting where the constraint is binding,  conformal policy learning can \emph{exhaust} the harm rate budget under partial identification, matching the globally optimal rule whose worst-case harm rate is exactly $\alpha$. 
We now continue on Remark~\ref{rmk:conservative} to explain why this sharpness can be achieved. 
The use of clipped score essentially eliminated the conservativeness of using $0$ instead of $Y_i^*$. We now focus on the conservativeness arising from replacing the ``oracle'' outcome $Y_i^* = \max\{Y_i(1),1-Y_i(0)\}$ by the proxy label $Y_i^\dagger = Y_iT_i + (1-Y_i)(1-T_i)\leq Y_i^*$, and discuss how this is eliminated by selective calibration.

The use of $Y_i^\dagger$ is deeply rooted in the partial identification nature of the problem: even for the labeled data, the harm $Y_i^*$ is not observed. 
As such, the sharpness hinges on whether, on the population level, calibration with $Y_i^\dagger$ matches the worst-case harm-rate bound that relies on $\gamma(X_i)$ in~\eqref{eq:def_bd_gamma}.

CPL achieves so through the selective calibration mechanism. 
Assuming, for intuitions, that the conditional mean functions are perfect, the optimal conformal p-value uses the inclusion indicator $G_i = T_i \ind\{1 -  {\mu}_1(X_i) \leq  {\mu}_0(X_{i})\} + (1-T_i) \ind\{1 -  {\mu}_1(X_i) >  {\mu}_0(X_{i})\}$. Namely, a treated unit $X_i$ enters the calibration set if and only if $1 - \mu_1(X_i)\leq \mu_0(X_i)$, in which case 
\$
\gamma(X_i) = 1-\mu_1(X_i) = \EE[ \ind\{Y_i^\dagger=0\}\given X_i,T_i=1]
\$
since $Y_i^\dagger=Y_i(1)$ for $T_i=1$. That is, the proxy label $Y_i^\dagger$ provides ``sharp'' calibration for the worst-case harm rate $\gamma(X_i)$. A similar observation applies to control units.  

The above discussion shows that for sharp harm rate control, one should use a labeled unit for calibration if and only if the treatment status matches whether the outcome model is ``active'' in the worst-case harm rate function $\gamma(x)$, i.e., whether a treated unit obeys $\gamma(X_i)=1-\mu_1(X_i)$ or a control unit obeys $\gamma(X_i)=\mu_0(X_i)$. 
In the simulations, we shall evaluate the ``informativeness'' of labeled data and the power of CPL. 

Finally, we re-emphasize the model-free nature of conformal policy learning. Even though the arm informativeness---through the inclusion function $\hat{g}$---is estimated and therefore certainly imperfect, our method provides conservative harm rate control without any modeling assumptions.

% !TEX root = main.tex

\section{Conformal Policy Learning with Stratified Experiments}
\label{sec:stratified}
In this section, we use weighted conformal inference to generalize our previous results to allow for both arbitrary, known propensity score $e(x)$ and arbitrary inclusion function $\hat{g}$.  
It also serves as the foundation for CPL in observational studies, which we cover in the next section.

The selective calibration process remains the same as Section~\ref{subsec:cpl_rct}. Consider any pre-trained monotone score function $V\colon \cX\times\{0,1\}\to \RR$ and inclusion function $\widehat{g}\colon \cX\times\{0,1\}\to [0,1]$. 
The calibration set is the same as~\eqref{eq:def_Dcal} where   $G_i\sim \textnormal{Bern}(\hat{g}(X_i,T_i))$ are independently drawn inclusion indicators. 

When computing conformal p-values for a new test point $X_{n+1}\sim \PP_X$, conditional on $G_i=1$, there is a covariate shift between  $\cD_{\calib}$ and $X_{n+1}$ that arises from both the treatment assignment and the exogenous random inclusion into the calibration set. To address this shift, we define the function $w\colon \cX\to \RR^+$ by 
\@\label{eq:def_weight}
w(x) = %\frac{
\bigl(e(x)\hat{g}(x,1) + (1-e(x))\hat{g}(x,0)\bigr)^{-1}.
% }{\EE[ \bigl(e(X)\hat{g}(X,1) + (1-e(X))\hat{g}(X,0)\bigr)^{-1}]},
\@
Because $e(x) \coloneqq \Pr(T_i = 1 \mid X_i=x)$ and $\hat{g}(x,t) \coloneqq \Pr(G_i = 1 \mid X_i = x, T_i = t)$, 
one can show that $w(x) = \Pr(G_i = 1 \mid X_i = x)^{-1},$ the inverse probability of being included for calibration given $X_i = x$, regardless of the choice of $e(x)$ and $\hat{g}$.  
Given the weights, we compute
\@\label{eq:def_pval_weight}
p^{\text{str}}_{n+1} = \frac{w(X_{n+1}) +\sum_{i = 1}^n G_i w(X_i) \ind\{V(X_i,Y_i^\dagger)\leq V(X_{n+1},0)\}}{w(X_{n+1}) + \sum_{i = 1}^n G_i w(X_i)}.
\@ 
The p-value~\eqref{eq:pval_general} in Section~\ref{sec:rct} is a special  case of~\eqref{eq:def_pval_weight}, where the simplifying assumptions $e(x)=0.5$ and $\hat{g}(x,1) + \hat{g}(x,0)=a$ for a constant $a>0$ implied constant weights and  no reweighting is needed.

Theorem~\ref{thm:valid_weighted} states that CPL with the above weighted conformal p-value % by thresholding the p-value, $\hat\pi^{\text{str}}(X_{n+1}) \coloneqq \ind\{p^{\text{str}}_{n+1} \leq \alpha\}$, satisfies 
achieves the distribution-free safety guarantee, whose proof is in Appendix~\ref{app:subsec_weighted_validity}. 

\begin{theorem}\label{thm:valid_weighted}
    Suppose Assumption~\ref{assump:uncon} and Assumption~\ref{assump:iid} hold, and $e(x)$ is known. 
    Then, for any score function $V \colon \cX\times\cY \to \RR$ that is non-decreasing in the second argument and for any inclusion function $\hat{g}\colon \cX\times\{0,1\}\to [0,1]$ obeying $e(X)\hat{g}(X,1)+(1-e(X))\hat{g}(X,0)>0$, $P_X$-a.s., it holds that $\PP(Y_{n+1}(1) = 0, Y_{n+1}(0) = 1, p^{\text{str}}_{n+1} \leq \alpha) \leq \alpha$ for any $\alpha\in (0,1)$. That is,~\eqref{eq:safety_guarantee} holds for $\hat\pi^{\text{str}}(X_{n+1}) = \ind\{p^{\text{str}}_{n+1}\leq \alpha\}$. 
\end{theorem}

As in Section~\ref{sec:rct}, given consistent outcome modeling, CPL achieves optimality with
\begin{equation}
\label{eq:pval_power_optimal_weights}
p^{\text{str-opt}}_{n+1} = \frac{w(X_{n+1}) + \sum_{i=1}^n w(X_i)    G_i   \ind\{Y_i^\dagger = 0\} \ind\{\hat{s}(X_i) \leq \hat{s}(X_{n+1})\}}{w(X_{n+1}) +\sum_{i=1}^n w(X_i)  G_i}, 
\end{equation}
where $G_i = T_i \one\{1 - \widehat{\mu}_1(X_i) \leq \widehat{\mu}_0(X_{i})\} + (1-T_i) \one\{1 - \widehat{\mu}_1(X_i) > \widehat{\mu}_0(X_{i})\}$. The score functions remains the same as the preceding case: $\hat{s}(x) = - (\widehat{\mu}_1(x) - \widehat{\mu}_0(x))/\widehat{\fna}(x)$ for welfare maximization. This is the weighted version of equation~\eqref{eq:pval_welfare_optimal}. 
The proof of Theorem~\ref{thm:opt_stratified_experiments} is in Appendix~\ref{app:subsec_opt_stratified}.

\begin{theorem}\label{thm:opt_stratified_experiments}
    Suppose Assumption~\ref{assump:uncon} and Assumption~\ref{assump:iid} hold, and $e(x)$ is known. Suppose $\|\hat{\mu}_t(X) -\mu_t(X)\|_{L_2(\PP_{X})}\stackrel{P}{\to}0$ as $n\to \infty$ for $t \in \{0,1\}$ and $s_{\welfare}(X)$ has no point mass. 
    Let 
    \$
    \hat\pi_{\text{str-opt}}(X_{n+1}) \coloneqq \ind\{p_{n+1}^{\text{str-opt}} \leq \alpha\} \ind\{ \hat\mu_1(X_{n+1})>\hat\mu_0(X_{n+1})\}
    \$ 
    where we take the same $G_i$ as~\eqref{eq:def_G_welfare_optimal} and $\hat{s}(\cdot)$ as~\eqref{eq:def_s_welfare_optimal}.  
    Then $\EE[Y(\hat\pi_{\text{str-opt}}(X_{n+1}))]\to \Welfare(\pi_\welfare^*;\PP)$ as $n\to \infty$, where $\pi_\welfare^*$ is the optimal solution in Theorem~\ref{thm:welfare_optimal}. 
\end{theorem} 
Analogously, CPL with the same selection indicators $\{G_i\}$ as~\eqref{eq:def_G_welfare_optimal} and the estimated power-optimal score yields asymptotically optimal power; this result is deferred to Theorem~\ref{thm:power_opt_str} in Appendix~\ref{app:subsec_power_opt}.

% !TEX root = main.tex

\section{Conformal Policy Learning with Observational Data}
\label{sec:obs}
In this section, we further generalize CPL to observational studies. With observational data, the main challenge is that the propensity score, hence the weight function~\eqref{eq:def_weight}, is unknown and needs to be estimated. While a natural idea is to estimate the weights and plug them into the conformal p-value, it remains unclear how the safety guarantee changes with the estimation quality. 
We present a learn-then-balance procedure to obtain the estimated weights so the resulting policy learning algorithm enjoys a double robustness property.

\subsection{Conformal Policy Learning with Learn-then-Balance Weights}

We construct p-values of a similar form as~\eqref{eq:def_pval_weight} and~\eqref{eq:pval_power_optimal_weights}:
\@\label{eq:def_pval_weight_est}
p_{n+1}^{\text{obs}} = \frac{\hat{w}_{n+1} +\sum_{i=1}^n \hat{w}_i   G_i  \ind\{Y_i^\dagger = 0\}  \ind\{\hat{s}(X_i)\leq \hat{s}(X_{n+1})\}}{\hat{w}_{n+1} + \sum_{i=1}^n \hat{w}_i G_i},
\@
where the inclusion indicators are $G_i \sim \text{Bern}(\hat{g}(X_i,T_i))$ for a function $\hat{g}\colon \cX\times\{0,1\}\to [0,1]$ whose training process is independent of the labeled and unlabeled data. 
The only difference of~\eqref{eq:def_pval_weight_est} from~\eqref{eq:pval_power_optimal_weights} is that the weights are now estimated. Notably, here we work with a form analogous to~\eqref{eq:pval_power_optimal_weights} since our balancing weights are more easily stated in terms of the score function $\hat{s}(\cdot)$. It corresponds to~\eqref{eq:def_pval_weight} with the conformity score $V(x,y) = M\ind\{y>0\}+\hat{s}(x)$ for a sufficiently large constant $M>2\sup_x|\hat{s}(x)|$. 

Denoting $\cI_{\calib}=\{i\in[n]\colon G_i=1\}$, the estimated weights $\{\hat{w}_i\}_{i\in\cI_{\calib}\cup\{n+1\}}$ in~\eqref{eq:def_pval_weight_est} are obtained by a learn-then-balance approach.  
We begin with two learned functions, one preliminary weight function $\tilde{w}\colon \cX\to \RR^+$, and one estimated harm-rate function $\hat\fna\colon \cX\to [0,1]$. 
A natural choice is to define $\tilde{w}(x) = 1/[\hat{e}(x)\hat{g}(x,1)+(1-\hat{e}(x)\hat{g}(x,0)]$ where $\hat{e}(x)$ is an estimated propensity score function, and $\hat\fna(x) = \min\{1-\hat\mu_1(x),\hat\mu_0(x)\}$ for estimated outcome models $\{\hat\mu_t(\cdot)\}_{t\in\{0,1\}}$. 
For simplicity, we require the propensity score and outcome models to be trained independently of the calibration and test data (in practice, one can use sample splitting to fit the models~\citep{chernozhukov2018double} and the properties are analogously studied in~\cite{jin2025cross}). 
Then, define the balancing feature vector 
\$
\hat\phi(x) = \big(\hat{\fna}(x)\ind\{ \hat{s}(x)\leq \hat{t}\}, \tilde{w}(x)\big)\in \RR^2,\quad \text{where}\quad \hat{t} = \sup \bigg\{t\in \RR\colon \frac{1}{n}\sum_{i=1}^n \hat{\fna}(X_i) \ind\{\hat{s}(X_i)\leq t\} \leq \alpha \bigg\}.
\$
Finally, we let $\{\hat{w}_i\}_{i\in \cI_{\calib}}$ be the optimal solution to the following optimization program with $\delta_n=O(n^{-1/2})$:
    \@\label{eq:def_sbw_main}
    \argmin_{ w\succeq 0} \Bigg\{ \|w\|^2 : \, \bigg| \frac{1}{|\mathcal{I}_{\calib}|} \sum_{i\in \mathcal{I}_{\calib}} w_i \hat\phi_k(X_i) - \frac{1}{n} \sum_{i=1}^n \hat\phi_k(X_i) \bigg|\leq \delta_n, k \in \{1,2\}; \frac{1}{|\mathcal{I}_{\calib}|}\sum_{i\in \mathcal{I}_{\calib}} w_i = 1 \Bigg\}.~~~~
    \@ 
We shall see that the post-hoc processing to ensure the finite-sample approximate balancing~\eqref{eq:def_sbw_main} is essential for us to obtain the doubly robust safety guarantees even when the propensity model is misspecified. 
The first condition in~\eqref{eq:def_sbw_main} follows the balancing weights~\citep{hainmueller2012entropy,zubizarreta2015stable}, by enforcing a finite-sample balancing condition inspired by the desired population-level property: for the correct weights $w(\cdot)$, it should hold that $\frac{1}{|\cI_{\calib}|}\sum_{i\in \cI_{\calib}} w(X_i) f(X_i)\approx \frac{1}{n}\sum_{i=1}^n f(X_i)$ for any fixed function $f\colon \cX\to \mathbb R$. Here, we balance specifically-designed functions to obtain favorable statistical properties under general conditions on the learned functions.

\subsection{Doubly Robust Safety Guarantees}

We now formalize the safety guarantees of CPL~\eqref{eq:def_pval_weight_est} with weights $\{\hat{w}_i\}_{i\in \cI_{\calib}}$ from~\eqref{eq:def_sbw_main}. 
These results will require mild regularity conditions for the learn-then-balance
optimization, which we defer to Assumption~\ref{assump:bal_reg}. These conditions
require overlap and boundedness, local stability of the population
balancing solution, and regularity of the population cutoff. 

Under these
conditions, and given that either the preliminary weight function or the outcome models are consistent, we obtain asymptotic safety guarantee for CPL. The proof of Theorem~\ref{thm:dr_obs} is included in Appendix~\ref{app:subsec_dr_obs}, which  follows from our general theory with estimated weights in Appendix~\ref{app:subsec_dr_obs_general}.

\begin{theorem}[Model double robustness]\label{thm:dr_obs}
Suppose Assumption~\ref{assump:bal_reg} holds, and $\delta_n = O(n^{-1/2})$. Then,  
    $$\limsup_{n\to \infty}\PP(Y_{n+1}(\hat{\pi}_{\textnormal{obs}}(X_{n+1}))<Y_{n+1}(0))\leq \alpha$$ under either of the following conditions:
    \begin{enumerate}[label=(\roman*).]
        \item The preliminary weight function is consistent: $\|\tilde{w}-w\|_{L_2(\PP_X)}=o_P(1)$ for the true weight  $w(\cdot)$ in~\eqref{eq:def_weight}.
        \item The outcome models are consistent: $\|\hat\fna-\fna\|_{L_2(\PP_X)}=o_P(1)$.
    \end{enumerate}
\end{theorem}

\begin{remark}
    We remark that the techniques and conditions here differ from existing model double-robustness results in conformal prediction (e.g.~\cite{lei2021conformal}) where correct outcome model alone is enough to ensure validity; in our setting, because of the thresholding nature of the CPL policy, explicit finite-sample balance turns out to be an important element of our theoretical analysis for double robustness.
\end{remark}

Taking a step further, we show that the learn-then-balance weights lead to a product error rate structure, which further yields the $O_P(n^{-1/2})$ convergence of the harm rate of CPL given slow, nonparametric convergence rates of the estimated models. The proof of Theorem~\ref{thm:dr_obs_rate} is in Appendix~\ref{app:subsec_dr_obs_rate}. 

\begin{theorem}[Rate double robustness]\label{thm:dr_obs_rate} 
Suppose Assumption~\ref{assump:bal_reg} holds and $\delta_n = O(n^{-1/2})$. If 
$\|\tilde w-w\|_{L_\infty(\PP_X)}
=
O_P(n^{-1/4})$ and 
$\|\hat\gamma-\gamma\|_{L_1(\PP_X)}
=
O_P(n^{-1/4})$, 
then $\{\PP(Y_{n+1}(\hat{\pi}_{\text{obs}}(X_{n+1}))<Y_{n+1}(0)\given \cA_n) - \alpha\}_{+} = O_P(n^{-1/2})$, where $\cA_n$ is the $\sigma$-field for the randomness in calibration data, training, and selective inclusion.  
\end{theorem}

\subsection{Optimal Conformal Policy Learning with Observational Studies}
Since the population optimization problem only depends on the super-population of the potential outcomes and observed features, the population-level optimal solution remains the same as  Theorem~\ref{thm:welfare_optimal}.  
The following theorem shows that under mild conditions, our method with the optimally chosen score and inclusion function achieves the optimal welfare with observational data. Its proof is in Appendix~\ref{app:subsec_optimality_obs}.  
The results for power optimality are deferred to Appendix~\ref{app:subsec_power_opt}.

\begin{theorem}[Welfare optimality with observational data]
\label{thm:obs_welfare_opt}
Let $\hat\pi_{\obs}(X_{n+1})=\ind\{p^{\obs}_{n+1}\leq\alpha\}\ind\{\hat{s}(X_{n+1})<0\}$ where we take the same $G_i$ as~\eqref{eq:def_G_welfare_optimal} and   $\hat{s}(\cdot)$ as~\eqref{eq:def_s_welfare_optimal}.
Suppose Assumption~\ref{assump:bal_reg} holds, and $\delta_n = O(n^{-1/2})$. 
Furthermore, assume $\|\hat g-g^*\|_{L_2(\PP_{X,T})}+\|\hat\fna - \fna\|_{L_2(\PP_{X})}+\|\hat s-s_{\welfare}\|_{L_2(\PP_X)}=o_P(1)$ for $g^*(x,t)=t\ind\{1-\mu_1(x)\leq\mu_0(x)\}+(1-t)\ind\{1-\mu_1(x)>\mu_0(x)\}$, and 
$s_{\welfare}(X)$ defined in~\eqref{eq:s_welfare_optimal} has no point mass.  
Then
$ 
\EE[Y\{\hat\pi_{\obs}(X_{n+1})\}]
\to
\Welfare(\pi_{\welfare}^*;\PP)$ as $n\to\infty$, where $\pi_{\welfare}^*$ is the optimal solution in Theorem~\ref{thm:welfare_optimal}.
\end{theorem}

As in randomized experiments, optimality requires two  convergence conditions: $\hat g\to g^*$, which makes the proxy-label calibration sharp, and $\hat s\to s_{\mathrm{welfare}}$, which provides the optimal treatment ranking. When $\hat g$, $\hat s$, and $\hat\gamma$ are constructed from plug-in outcome model estimates, these
conditions follow from consistent outcome models under the usual regularity conditions. In this outcome-correct case, the balancing condition aligns the weighted calibration with the oracle harm-welfare curve, so the weight estimator need not be consistent. However, inheriting the model consistency of Theorem~\ref{thm:dr_obs}, one can show that the conclusion also continues to hold under consistent
weights, $\bar w\propto w$, without requiring
$\hat\gamma\to\gamma$, provided that $\hat g\to g^*$ and
$\hat s\to s_{\mathrm{opt}}$ still hold. Yet, we do not pursue it here since the latter two convergence conditions usually require consistency of the estimated outcome models which imply the consistency of the harm rate estimation.

% \newpage
% !TEX root = main.tex

\section{Simulation Studies}
\label{sec:sim} 

\subsection{Balanced Randomized Experiments}
\label{subsec:simu_rct}

We begin by evaluating the methods in balanced randomized experiments (Section~\ref{sec:rct}).  
The data in the randomized experiment settings are generated using the general framework below. 

\vspace{-0.75em}
\paragraph{Data generating processes.}  
Throughout, we generate the features $X\sim P_X$ for some distribution $P_X$ and  i.i.d.~treatment indicators $T\sim \text{Bern}(1/2)$, independent of everything else.  The potential outcomes follow $\PP(Y(1)=1\given X=x)=\mu_1(x)$ and $\PP(Y(0)=1\given X=x)=\mu_0(x)$ for some functions $\mu_1(\cdot)$ and $\mu_0(\cdot)$. The potential outcomes are coupled negatively, meaning that 
$\PP(Y(1)=0,Y(0)=1\given X=x) = \min\{1-\mu_1(x), \mu_0(x)\}$. 
Note that any coupling would lead to the same observed data distribution $P_{X,Y,T}$ (and thus the same results for any method), but the negative coupling leads to the worst-case harm rate.  
The specific choice of $P_X$ and $\mu_1(\cdot)$, $\mu_0(\cdot)$ is introduced in each set of simulations below. 
We first design four diverse settings to compare our methods and baselines, followed by two additional settings to investigate the performance of our methods, including (i)  the power of selective calibration by varying the informativeness of arms, and (ii) decisions by heterogeneous subgroups.

\subsubsection{Harm rate control in diverse settings} 
\label{subsubsec:simu_rct_main}

The first set of simulations evaluates the harm rate control across diverse settings and modeling choices. We set $X\in \RR^{20}$ with $P_X= \cN(0, I_d)$. The conditional mean functions are $\mu_t(x) = \exp(\eta_t(x))/(1+\exp(\eta_t(x))$, $t\in \{0,1\}$ for some functions $\eta_t(x)$. We design four settings  (see Appendix~\ref{app:subsec_simu_rct_details} for detailed DGPs):
\begin{itemize}
    \item Setting 1: approximately linear, where $\eta_t(x)=\beta_t^\top x$ for some $\beta_t\in \RR^{20}$ and fully-observed  $X\in \RR^{20}$;
    \item Setting 2: approximately linear DGP but the observed covariates are $X_{1:7}\in \RR^{7}$;
    \item Setting 3: nonlinear  $\eta_t(x)$ involves the first 5 entries of $x$, with fully observed $X\in \RR^{20}$;
    \item Setting 4: nonlinear  $\eta_t(x)$ involves the first 8 entries of $x$, but the observed covariates are $X_{1:5}\in \RR^{5}$.
\end{itemize}
In settings 2 and 4, missing some covariates typically reduces the accuracy of learned outcome models (finite-sample harm-rate control still holds as per Theorem~\ref{thm:validity_rct}). We vary the total labeled sample size $n_{\text{total}}\in\{500,1000,2000\}$, fixing the number of test data $m=1000$, and $\alpha=0.1$. We use $\cD_{\text{label}}$ to denote the labeled data and $\cD_{\test}$ to denote the test data.

\vspace{-0.75em}
\paragraph{Methods and evaluation.}

We compare the following  three baselines and two CPL methods:
\begin{enumerate}[label=(\arabic*)]\setlength\itemsep{-0.1em}
    \item \texttt{Li et al.}, \cite{li2023trustworthy} with doubly-robust estimator and constrained optimization among depth-2 decision trees via \texttt{econml} python library. The policy is learned on $\cD_{\text{label}}$ and evaluated on $\cD_{\test}$. This method requires the distributional assumption about the potential outcomes. In particular, it assumes that the potential outcomes are non-negatively correlated, which is violated in this setup.
    \item \texttt{Policy-Tree}, which maximizes welfare $\EE[Y(\pi(X_{n+1}))]$ among policies represented by depth-2 decision trees with \texttt{econml}. The policy is learned on $\cD_{\text{label}}$ and evaluated on $\cD_{\test}$. While this method is not designed for safety control, we include it as the baseline due to its popularity as a standard practice. 
    \item \texttt{Threshold}, which uses $\cD_{\text{label}}$ to obtain an estimator $\hat{\fna}(x)=\min\{1-\hat\mu_1(x),\hat\mu_0(x)\}$ and treat  $\{j\in[m]\colon \hat{\fna}(X_{n+j})\leq \tilde\gamma\}$ where $\tilde\gamma = \max\{\gamma\colon \sum_{j=1}^m \hat\fna(X_{n+j})\ind\{\hat{\fna}(X_{n+j})\leq \gamma \} \leq 0.1\cdot m\}$. This heuristic calibration is valid only when the estimator $\hat{\gamma}$ is sufficiently accurate.   

    \item \texttt{CPL-sel-welfare}, the CPL algorithm that maximizes the welfare, i.e., Algorithm~\ref{algo:rct} with $g(x,t)=t \ind\{1-\hat\mu_1(x)\leq \hat\mu_0(x)\} +(1-t)\ind\{1-\hat\mu_1(x)>\hat\mu_0(x)\}$ and a clipped score function $V(x,y) = M\ind\{y>0\} - (\hat{\mu}_1(x) - \hat{\mu}_0(x))/\hat{\fna}(x)$ with a sufficiently large constant $M > 0$. This method always satisfies the safety guarantee and asymptotically achieves the optimal welfare if $\mu_t(x)$ is consistently estimated.
    \item \texttt{CPL-sel-power}, the CPL algorithm that maximizes the power, i.e., Algorithm~\ref{algo:rct} with $g(x,t)=t \ind\{1-\hat\mu_1(x)\leq \hat\mu_0(x)\} +(1-t)\ind\{1-\hat\mu_1(x)>\hat\mu_0(x)\}$ and a clipped score function $V(x,y) = M\ind\{y>0\} + \hat{\fna}(x)$ with a sufficiently large constant $M > 0$. This method always satisfies the safety guarantee and asymptotically achieves the optimal power if $\mu_t(x)$ is consistently estimated (Appendix~\ref{app:subsec_power_opt}).
\end{enumerate}
\vspace{0.25em}

\begin{figure}[!t]
    \centering
    \includegraphics[width=\linewidth]{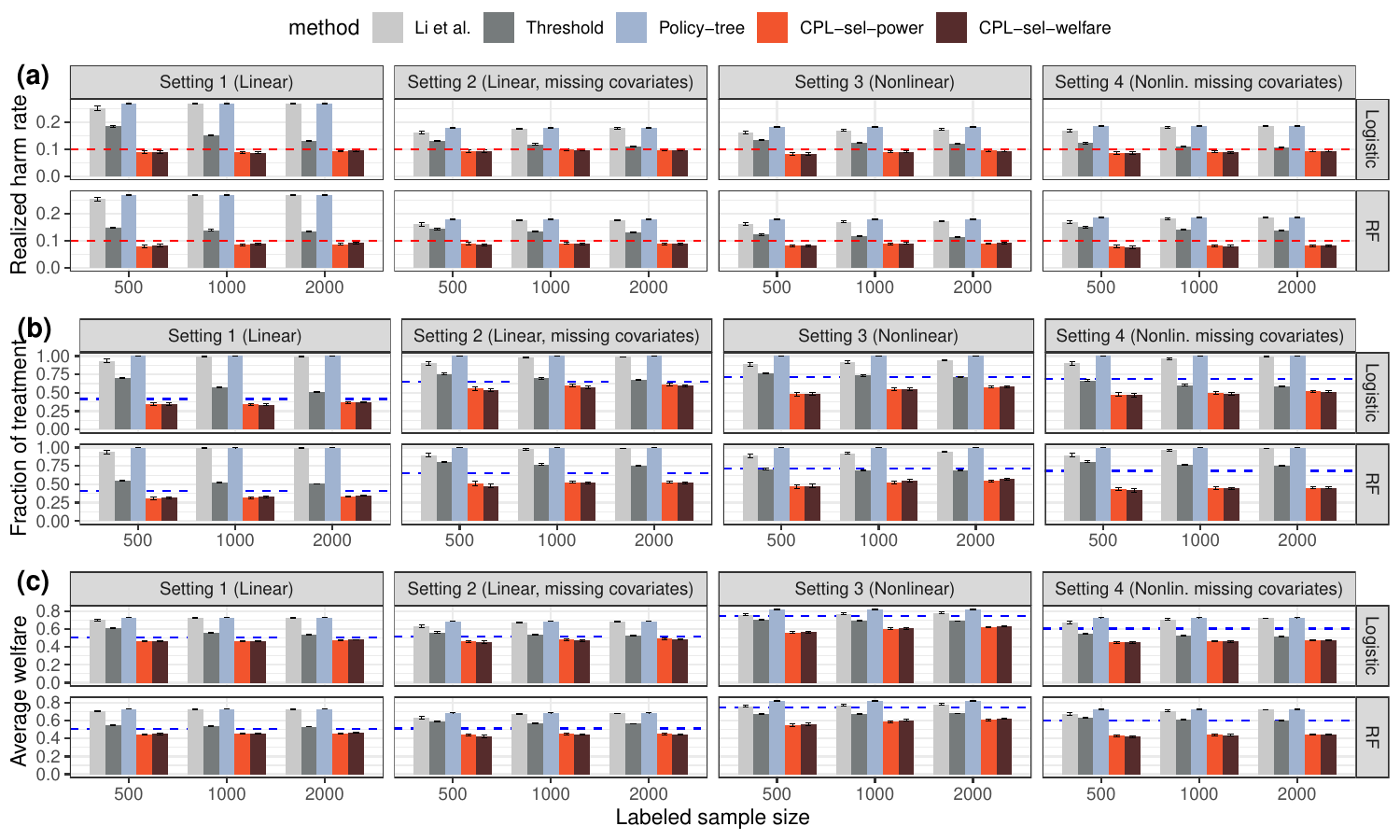}
    \caption{(a) Empirical harm rate, (b) power (probability of treatment in the test units), (c) average welfare of various methods at level $\alpha=0.1$ in the randomized experiment studies, averaged over $200$ independent simulation runs. Each row represents a different prediction model (logistic regression, random forest) for $(\hat\mu_1,\hat\mu_0)$, and each column represents a different data-generating process. The $x$-axis is the total labeled sample size $n$. The blue dashed lines in (b) and (c) show the optimal power and welfare of CPL based on oracle correct models $(\mu_1,\mu_0)$ without any estimation error.}
    \label{fig:rct_fna}
\end{figure}

For the two CPL variants, the labeled data $\cD_{\text{label}}$ is  randomly split into the training ($75\%$) and calibration ($25\%$) folds $\cD_\train$ and $\cD_\calib$. The training fold is used to fit two models $\hat\mu_1(x)$ and $\hat\mu_0(x)$ for $\mu_1(x)$ and $\mu_0(x)$, respectively. 
The function $\hat{\fna}(x)=\min\{1-\hat\mu_1(x),\hat\mu_0(x)\}$ then estimates the upper bound on the harm rate function.
We adopt two model classes for training the outcome models: logistic regression, and random forest, which are used in both \texttt{Threshold} and \mname methods. 
\texttt{Li et al.}~and \texttt{Policy-Tree} use random forests for nuisance component estimation to be consistent with the tree-based policy class. 
See Appendix~\ref{app:subsec_simu_rct_details} for additional details on method implementation.

Given the learned treatment $T_{n+j} \coloneqq \hat\pi(X_{n+j}) \in \{0,1\}$ for $j\in [m]$ produced by the methods, we evaluate three metrics: the harm rate by $\frac{1}{m}\sum_{j=1}^m T_{n+j}\ind\{Y_{n+j}(1)<Y_{n+j}(0)\}$, the welfare by $\frac{1}{m}\sum_{j=1}^m Y_{n+j}(T_{n+j})$, and the fraction of treatment by $\frac{1}{m}\sum_{j=1}^m T_{n+j}$. 
All metrics are averaged over $200$ independent  runs.

\vspace{-0.75em} 
\paragraph{Simulation results.} Figure~\ref{fig:rct_fna} presents the results for the five methods across various settings. First, the three baselines lead to a drastic violation of the target harm rate. The \texttt{Threshold} method relies on accurate outcome models to ensure consistency of $\hat\fna(\cdot)$, which is difficult to satisfy with limited labeled data. \texttt{Policy-tree}, which focuses on welfare maximization can incur large harm rates. Finally, \texttt{Li et al.}~achieves harm rate only when the two potential outcomes are non-negatively correlated given the features; as a result, it leads to exceedingly high harm rate due to (i) worst-case negative coupling and (ii) inconsistency due to missing covariates in settings 2 and 4. 

On the other hand, the two variants of \mname control the harm rate tightly at the target level, showing both the validity and sharpness of CPL. The two variants based on different score functions demonstrate negligible difference: in these settings, the rank of instances based on the two optimal scores does not drastically change the decisions by \mname. In practice, this means that researchers can approximately maximize both power and welfare together via CPL without worrying about the tradeoff between the two objectives. 
Finally, compared with the oracle power and welfare (blue dashed lines), the CPL methods achieve close-to-optimal performance, and the gap shows the impact of estimation error in $\hat\mu_t$ functions. Such gap seems to be moderate even for misspecified models (Logistic regression in settings 3 and 4).
 
\subsubsection{Effectiveness of selective calibration} 
We now use another set of experiments to dive deeper into the selective calibration mechanism. 
Following the discussion in Section~\ref{subsec:discuss_sharp}, the harm rate control by using the conservative proxy $Y_i^\dagger = T_iY_i +(1-T_i)(1-Y_i)$ is tight if a selected sample satisfies (i) $T_i=1$ and $\fna(X_i)=1-\mu_1(X_i)$, 
or (ii) $T_i=0$ and $\fna(X_i)=\mu_0(X_i)$.  
Our selective calibration method in Section~\ref{subsec:cpl_rct} aims to address this by adaptively selecting the informative arm. 
In this part, we vary the magnitudes of $\mu_1(x)$ and $\mu_0(x)$ to how this strategy contributes to the statistical efficiency.
We vary the proportion of samples obeying $\fna(X) = 1-\mu_1(X)$ (treated arm is informative) and those obeying $\fna(X)=\mu_0(X)$ (control arm is informative); see Appendix~\ref{app:subsec_rct_aux} for the detailed data-generating process. 
In addition to the two CPL variants evaluated in Section~\ref{subsubsec:simu_rct_main}, we evaluate two more procedures:
\begin{enumerate}[label=(\arabic*)]\setlength\itemsep{-0.1em} 
    \item[(6)] \texttt{CPL-treat}, our method in Section~\ref{subsec:pre} with treated samples in $\cD_\calib$ and $V(x,y)=My-\hat\mu_1(x)$ with a sufficiently large constant $M > 0$. 
    \item[(7)] \texttt{CPL-control}, our method in Section~\ref{subsec:pre} with control samples in $\cD_\calib$ and $V(x,y) = My +\hat\mu_0(x)$ with a sufficiently large constant $M > 0$. 
\end{enumerate}
\vspace{0.25em}
The procedures are evaluated in terms of harm rate, power (fraction of treatment), and welfare, with the metrics averaged over $200$ independent runs.

\begin{figure} 
    \centering
    \includegraphics[width=\linewidth]{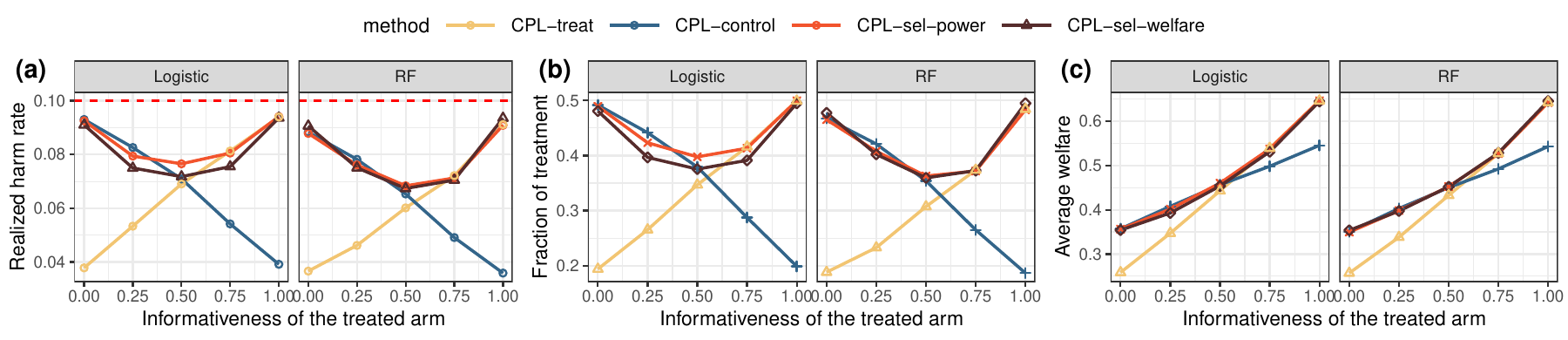}
    \caption{(a) Empirical harm rate, (b) power (probability of treatment) and (c) average welfare of the seven procedures in the arm-informativeness experiments at level $\alpha=0.1$ over $200$ independent runs. Within each panel, each column represents a prediction model (logistic regression, random forest) for $(\hat\mu_1,\hat\mu_0).$}
    \label{fig:rct_aux_arm}
\end{figure}

The results are summarized in Figure~\ref{fig:rct_aux_arm}, where the $x$-axis is the fraction of treated samples being informative (i.e., $\fna(X)=1-\mu_1(X)$). 
As before, all variants of \mname~control the harm rate below $\alpha=0.1$. 
\texttt{CPL-treat} and  \texttt{CPL-control}  demonstrate clear power tradeoffs: \texttt{CPL-treat} is more powerful when the treated arm is more informative ($x$-axis above $0.5$), and the opposite happens otherwise. 
Importantly, \texttt{CPL-sel-power} and \texttt{CPL-sel-welfare} are often comparable to the more powerful single-arm variants, showing the effectiveness of selective calibration. This justifies our recommendations for the asymptotically optimal variants (Section~\ref{subsec:rct_optimality}), since in practice it is often unknown  which arm might be more powerful.

\subsubsection{Performance under subgroup heterogeneity} 
% running. story: what happens when treatment benefits a subgroup and harms a subgroup
To further inspect the behavior of \mname~under treatment effect heterogeneity, we design a setting with three subgroups driven by the first two features in $X\in \RR^{20}$ (Figure~\ref{fig:rct_aux_subgroup} panel (a)). There is strong cross-group heterogeneity, but the units in the same group are largely similar. Studying the decisions in each group offers a zoom-in observation of ``who gets treated'' with different scoring functions $s(\cdot)$ in \mname.

The first group consists of half of the population, whose worst-case harm rate $\fna(X)$ is relatively large while the conditional average treatment effect $\tau(X)=\mu_1(X)-\mu_0(X)$ is also large. Groups 2 and 3 consist of $1/4$ of the population each, whose worst-case harm rate  $\fna(X)$ and conditional average treatment effect $\tau(X)$ are both small; the main difference is that the treated samples are more ``informative'' in Group 2 (i.e., $1-\mu_1(X)<\mu_0(X)$), while the control samples are more informative in Group 3. We follow the same procedures as before to evaluate the four variants of \mname.

\begin{figure}%[htbp]
    \centering
    \includegraphics[width=\linewidth]{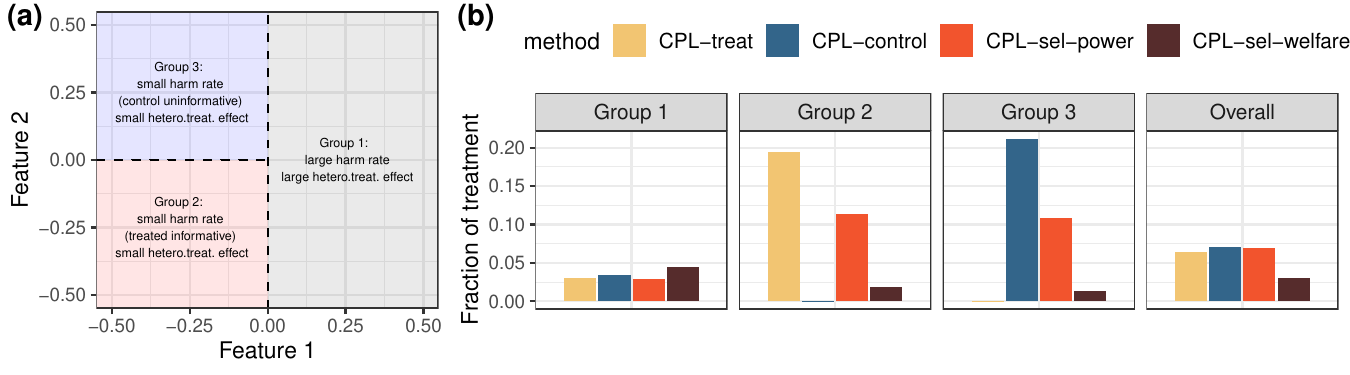}
    \caption{(a) Subgroup setup, (b)  Per-group fraction of $T_{n+1}=1$ for variants of \mname in the subgroup DGP at level $\alpha=0.01$. Panels (c-d) show random forests (RF) as the outcome model only.}
    \label{fig:rct_aux_subgroup}
\end{figure}

The results averaged over $200$ independent simulation runs are summarized in Figure~\ref{fig:rct_aux_subgroup}.
In panel (b), we show the fraction of $T_{n+j}$ within each group based on random forests predictions. Different score functions lead to different treatment prioritization patterns. \texttt{CPL-treat} concentrates the safe treatment budgets on group 2 (small $\fna(X)$ with treatment arm being informative) since it ranks instances based on $\hat\mu_1(x)$, while \texttt{CPL-control} concentrates the budgets on group 3. In contrast, \texttt{CPL-sel-power} distributes the budgets relatively uniformly across groups 2 and 3 (since their harm rates are comparably small) by using the score function $\hat\fna(x)$. \texttt{CPL-sel-welfare}, which ranks instances by balancing harm rate and treatment effect size, puts most budgets on Group 1 (large treatment effects) but assigns fewer treatments in general.

\subsection{Stratified Experiments and Observational Studies}
\label{subsec:simu_obs}

In this part, we proceed to evaluate \mname in stratified experiments and observational studies. The stratified experiments induce a \emph{known} covariate shift between the calibration and test data, while in observational studies this covariate shift needs to be estimated. 

\vspace{-0.75em}
\paragraph{Simulation settings.} We first sample the triplets $\{(X_i,Y_i(1),Y_i(0)\}_{i=1}^{n+m}$ using data generating processes to be specified later. 
In the labeled data, the treatment assignments are sampled by $T_i \mid X_i \sim \mathrm{Bernoulli}(e(X_i))$ independently with propensity score function $e(x)$ to be specified later. The observed outcome is $Y_i=Y_i(T_i)$. 

\vspace{-0.5em}

\paragraph{Methods and evaluation.} 
Fixing the confidence level at $\alpha=0.1$, we compare the following procedures:
\begin{enumerate}[label=(\arabic*)]\setlength\itemsep{-0.1em}
    \item \texttt{Li et al.}, \cite{li2023trustworthy} with a doubly-robust estimator and constrained optimization among depth-2 decision trees via \texttt{econml} python library, where the outcome models and propensity scores are estimated under two-fold cross-fitting. The policy is learned on $\cD_{\text{label}}$ and evaluated on $\cD_{\test}$. This method requires that the potential outcomes are non-negatively correlated, which is violated here.
    \item \texttt{Policy-Tree}, which maximizes welfare $\EE[Y(\pi(X_{n+1}))]$ among depth-2 decision trees with \texttt{econml}. The policy is learned with $\cD_{\text{label}}$ and evaluated on $\cD_{\test}$ with similar cross-fitting estimation of outcome and propensity score models. While this method is not designed for safety control, we include it as the baseline because of its popularity as a standard practice.
    \item \texttt{Threshold}, the same as in Section~\ref{subsec:simu_rct} which calibrates the threshold using cumulative estimated $\hat\fna(X)$ on the test data. This method is valid only when the estimator $\hat{\fna}$ is sufficiently accurate.  
    \item \texttt{CPL-sel-welfare}, the method in Section~\ref{sec:obs} using the same $g(x,t)$ and $s(x)$ functions as in the randomized experiment case. The weights are estimated using  learn-then-balance with features $\hat\phi(x) = (\hat{\fna}(x)\ind\{ {s}(x)\leq \hat{t}\}, \hat{w}(x))$, where $\hat{\fna}$ and $\hat{w}$ are estimated using logistic regression or random forests.
    \item \texttt{CPL-sel-power}, the method in Section~\ref{sec:obs} using the same $g(x,t)$ and $s(x)$ functions as in the randomized experiment case.  and the weights are estimated in the same way as \texttt{CPL-sel-welfare}.
    \item \texttt{Plugin-oracle-sel-welfare}, the method in Section~\ref{sec:stratified} using the correct weights and the same $g(x,t)$ and $s(x)$ functions as \texttt{CPL-sel-welfare}.
    \item \texttt{Plugin-oracle-sel-power}, the method in Section~\ref{sec:stratified} using the correct weights and the same $g(x,t)$ and $s(x)$ functions as \texttt{CPL-sel-power}.
\end{enumerate}
\vspace{0.25em}

Here, the last two plugin-oracle methods  essentially evaluate \mname in the stratified experiment setting where the propensity score $e(x)$, and hence the covariate shift weights, are known. Our theory implies their finite-sample harm rate control due to weighted exchangeability, no matter the accuracy of the scores and selective calibration functions. 
On the other hand, \texttt{CPL-sel-power} and \texttt{CPL-sel-welfare}  evaluate the robustness of CPL in observational studies with estimated weights.

\subsubsection{Harm rate control with known or estimated weights}

We first evaluate the harm rate control using 
the same data-generating process as in Section~\ref{subsec:simu_rct}, together with a linear or nonlinear propensity score. %nonlinear propensity score function $e(x) = \mathrm{logit}^{-1}\left\{
% 0.25\left[
% \sin(3\pi x_1) + \cos(2\pi x_8) + 0.1 x_3^2 - 0.3 x_4
% \right]
% \right\}$ truncated to lie in $[0.05, 0.95]$. 
The empirical harm rate, fraction of treatment, and average welfare, averaged over $200$ independent simulation runs, are summarized in Figure~\ref{fig:obs_fna}. 
First, the three baselines drastically violate the harm rate control due to similar reasons as in Section~\ref{sec:rct}, now with additional estimation challenges in observational settings. 
The two plug-in oracles, \texttt{Plugin-oracle-sel-power} and \texttt{Plugin-oracle-sel-welfare}, confirm the finite-sample harm-rate control of Theorem~\ref{thm:valid_weighted}. 
Finally, the two methods with estimated weights, \texttt{CPL-sel-power} and \texttt{CPL-sel-welfare}, also demonstrate robust harm-rate control below the target level. 
Although the four CPL methods rely on distinct scoring functions, they achieve similar performance in terms of fraction of treatment and average welfare, showing that the two objectives can align. 

\begin{figure}%[htbp]
    \centering
    \includegraphics[width=\linewidth]{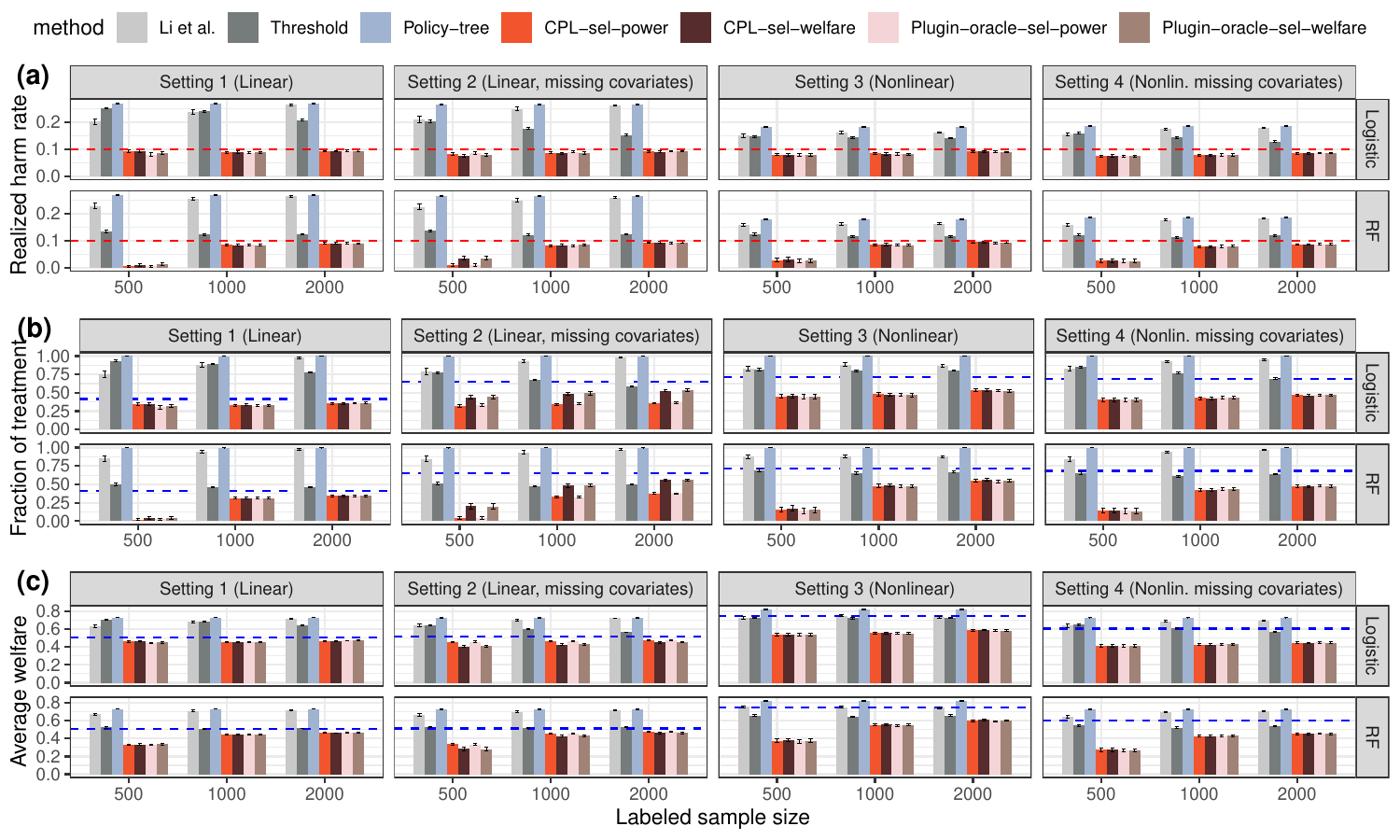}
    \caption{(a) Empirical harm rate, (b) power (probability of treatment in the test units), (c) average welfare at level $\alpha=0.1$ in observational studies, averaged over $N=200$ runs. Each row is a prediction model (logistic regression, random forest) for $(\hat\mu_1,\hat\mu_0)$, and each column is a data-generating process. The $x$-axis is the total labeled sample size $n$. The \texttt{Plugin} methods use the true weights.}
    \label{fig:obs_fna}
\end{figure}

\subsubsection{Double robustness}
 
We now additionally examine the double robustness property established in Section~\ref{sec:obs}. 
We design a data-generating process where the outcome models and propensity score model are all logistic in some nonlinear transformation of the raw features $X\in \RR^{8}$; see Appendix~\ref{app:subsec_simu_obs} for details. 
We sample observational data $\{(X_i,T_i,Y_i)\}_{i=1}^n$ as the labeled dataset, where  a random subset of 75\% is used as the training fold for training the models $\mu_1(\cdot)$, $\mu_0(\cdot)$, and $e(\cdot)$, while the remaining is used as the calibration data in \mname. 
We vary a parameter in the propensity score model to control the strength of confounding (the $x$-axis of Figure~\ref{fig:obs_aux}). 

The methods evaluated include \texttt{CPL-sel-power/welfare} based on estimated propensity scores, and two known-propensity baselines \texttt{Plugin-oracle-sel-power/welfare} used as oracle comparison only. The four procedures are implemented in the same way as in the preceding parts. 
Within each procedure, we vary the model classes of the outcome and propensity score models to demonstrate the double robustness property. 
A correct logistic model runs logistic regression over the nonlinearly-transformed features, 
while a misspecified logistic model runs logistic regression over the raw features. 
By Theorem~\ref{thm:dr_obs}, we expect \texttt{CPL-sel-power/welfare} to control the harm rate when either of them is correctly specified.  
Finally, we also consider that both outcome models and propensity scores are trained via random forests, which typically have slower-than-parametric convergence rates yet are less prone to model misspecification; we expect it to control the harm rate, especially when the labeled sample size is sufficiently large. 

The empirical harm rate averaged over $200$ independent simulation runs at nominal level $\alpha=0.05$ is summarized in Figure~\ref{fig:obs_aux}. In the first three columns, when either the outcome models or the propensity score model is well-specified, we observe tight harm-rate control by the two CPL variants, which is also close to the plugin-oracle ones. This validates the double robustness theory. 
When the outcome models are misspecified, we observe a larger slack in the harm rate than the other two configurations when sample size is small. 
On the other hand, when both models are misspecified (the fourth column), the two CPL variants can violate the harm rate control, yet the violation is moderate. 
We note that we design the settings deliberately to fail the both-misspecified procedures. In many other settings, misspecification can create a slack in the conservative calibration through the proxy outcome $Y_i^\dagger$, which often compensates the misspecification in the weights and keeps the harm rate below the budget even though both models are wrong. 
Finally, the nonparametric models in the fifth column yield tight harm rate control, showing the robustness to model misspecification and the quick convergence of the harm rate in \mname. 

\begin{figure}
    \centering
    \includegraphics[width=0.8\linewidth]{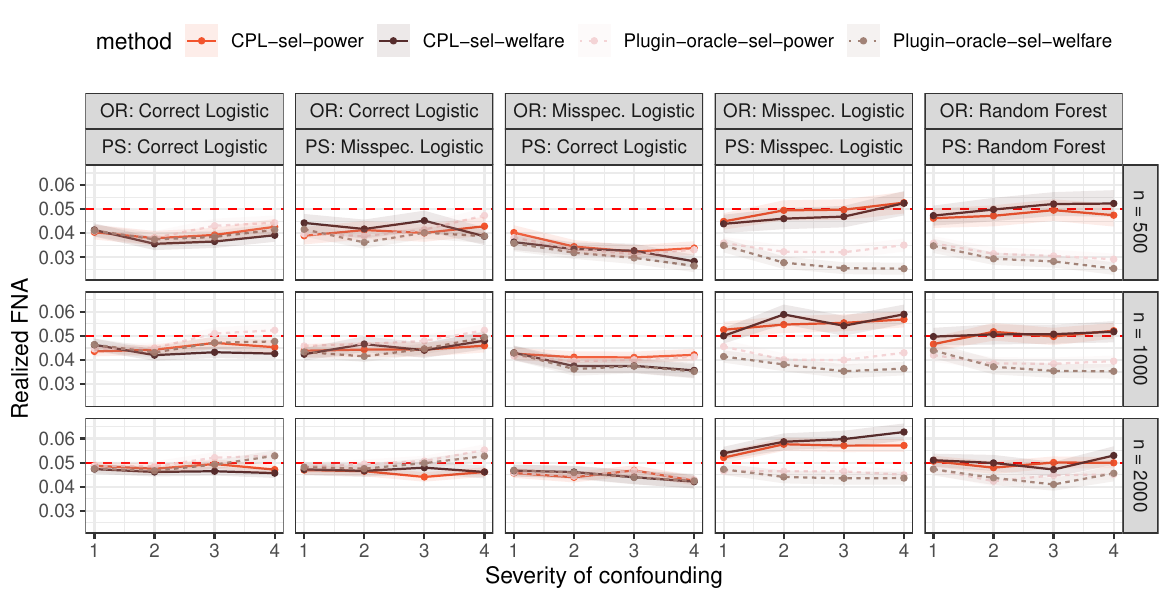}
    \caption{Empirical harm rate in the double robustness experiments; results are averaged over $N=200$ independent runs. Each column corresponds to one specification of outcome and propensity score models. Each row corresponds to a total labeled sample size. The $x$-axis is the strength of confounding.}
    \label{fig:obs_aux}
\end{figure}

% !TEX root = main.tex
\section{Empirical Application to AI-Powered Interventions}
\label{sec:real}

In this section, we apply \mname to a real-world dataset in the social sciences, where the AI model, ChatGPT, is used as an intervention to persuade participants out of some conspiracy beliefs, in which events are understood as being caused by secret, malevolent plots involving powerful conspirators \citep{costello2024durably}. The original study found that brief conversations with AI could reduce conspiracy beliefs by 20 percent on average, and the effect was durable for at least 2 months. Given the societal concern over widespread conspiracy theories, an increasing number of researchers and policymakers are evaluating similar AI interventions as a scalable solution. If policymakers scale up such AI interventions, it is of significant importance to consider safety, as the treatment effects of AI interventions are likely to be highly heterogeneous for several reasons. First, ``people believe a wide range of
conspiracies, and the specific evidence brought to bear in support of even a particular conspiracy
theory may differ substantially from believer to believer'' \citep[page 1,][]{costello2024durably}. Second, as AI chatbots treat people with natural texts as the intervention, the content of the treatment itself is heterogeneous and unpredictable. Here, to safely scale up these types of AI-powered interventions, we use CPL to decide who should be treated by AI by controlling the harm rate with the safety guarantee. 

In this study, before the experiment, the participants stated a conspiracy theory they believed in and reported a numerical score quantifying their belief in it. They are then randomly assigned to treated and control groups, where treated participants engage in a live conversation with a GPT model that is instructed to talk them out of the conspiracy, while the participants in the control condition engage with a neutral conversation with the GPT model. After the experiment, they again report a numerical score quantifying their belief in the same conspiracy theory. 

We take all participants in the raw dataset as the analysis population. We binarize the outcome $Y$ to indicate whether the post-experiment belief score is below $50$, a cutoff the authors originally used to define their analysis population. The participants are randomly split into training (40\%), calibration (40\%), and test (20\%) folds. There are 416 participants originally treated out of 667 participants in the test fold. The fraction of $Y=1$ in the treated group is $0.274$, while the fraction of $Y=1$ in the control group is $0.100$. We consider a stringent harm rate of $\alpha=0.025$. 

We build features based on participants' demographic covariates (education, age, gender, race, religion),   political and psychological covariates, AI-related covariates (familiarity and trust in generative AI),  baseline belief state variables and textual embedding for the stated conspiracy. 
The training fold is used to fit the outcome models and conditional treatment effect models, which are used in a similar way as in the simulations to build welfare-maximizing score functions (except that we truncate on extremely small estimated harm rate for stability). 
See Appendix~\ref{app:econml-welfare} for the detailed implementation.

\vspace{-0.5em}
\paragraph{Empirical welfare and power.} 
We report power (the fraction of treated units in the test data) and (estimated) empirical welfare of the welfare-maximizing and power-maximizing variants of \mname. We also compare the results against three baselines as references: the first is to treat everyone in the test data (All treat), the second is to treat no one in the test data (All control), and the third is to treat $2.5$\% of test units, which trivially satisfies the safety constraint. 

Because we observe the realized outcome in the test data, we can estimate the average welfare and harm rate of the policy as follows. Let $T_{n+j}^{\text{real}}\in \{0,1\}$ be the actual treatment assigned by the experiment, and $T_{n+j} = \hat\pi(X_{n+j})$ be the decision produced by \mname. 
We are interested in the average welfare $\text{Welfare}(\hat\pi):=\EE[Y_{n+j}(\hat\pi(X_{n+j}))]
= \EE[Y_{n+j}(1)\cdot\hat\pi(X_{n+j}) +Y_{n+j}(0)\cdot(1-\hat\pi(X_{n+j}))]$, for which an unbiased estimate is 
\@\label{eq:real_welfare_est}
\widehat{\text{Welfare}} = \hat\EE_{\test}\big[Y_{n+j}\cdot \hat\pi(X_{n+j})\given T_{n+j}^\text{real}=1\big]
+ 
\hat\EE_{\test}\big[Y_{n+j}\cdot (1-\hat\pi(X_{n+j}))\given T_{n+j}^\text{real}=0\big],
\@
where $\hat\EE_{\test}$ denotes the empirical average in the test fold.  
We also estimate the harm rate 
$
\PP(Y_{n+j}(1)<Y_{n+j}(0),\hat\pi(X_{n+j})=1)
$ 
by the following conservative estimator.
\@\label{eq:est_fna}
\widehat{\text{Harm}} =& \hat\EE_{\test}[(1-Y_{n+j})\ind\{1-\hat\mu_1(X_{n+j})\leq \hat\mu_0(X_{n+j})\}\hat\pi(X_{n+j})\given T_{n+j}^{\text{real}}=1\big] \notag \\ 
&\qquad + \hat\EE_{\test}[ Y_{n+j} \ind\{1-\hat\mu_1(X_{n+j}) > \hat\mu_0(X_{n+j})\}\hat\pi(X_{n+j})\given T_{n+j}^{\text{real}}=0\big] .
\@
Note that $\widehat{\text{Harm}}$ is unbiased and asymptotically normal for the population quantity $ \EE\big[ \{ (1- \mu_1(X_{n+j}))\ind\{1-\hat\mu_1(X_{n+j})\leq \hat\mu_0(X_{n+j})\} + \mu_0(X_{n+j})\ind\{1-\hat\mu_1(X_{n+j})> \hat\mu_0(X_{n+j})\} \} \cdot  \hat\pi(X_{n+j}) \big]$, which upper bounds the harm rate. This is a valid conservative estimator of the harm rate even if $\hat\mu_t(x)$ is misspecified, and this is a consistent estimator for the sharp upper bound of the harm rate when $\hat\mu_t(x)$ is consistently estimated. 

\paragraph{Results.} The main results are summarized in Table~\ref{tab:real_welfare}. Several points are worth noting. First, treating everyone (All treat) violates the safety constraint as its harm rate is $3.92$\% and exceeds $2.5$\%. In contrast, our proposed CPL methods achieve the tight control of harm rates at 2.5\% and 2.1\%, respectively. Second, while treating no one (All control) and treating only $2.5$\% of test units (Trivially-safe), of course, satisfy the safety constraint, they have extremely low power and welfare. In contrast, our proposed methods can treat more than 90\% of test units and achieve the high average welfare. While maintaining the safety constraint, our method simultaneously achieves high power and welfare. Finally, it is important to note that the difference between our welfare-maximizing and power-maximizing variants are minimal in practice. It is true that, consistent with our theory, our power-maximizing variant has a slightly higher fraction of treated test units, and our welfare-maximizing variant has a slightly higher average welfare. However, overall, their actual policy decision on who gets treated is similar, which implies that researchers can use either variant in practice and expect to approximately optimize both power and welfare in various applications.

\begin{table}
    \centering
    \begin{tabular}{c|c|c|c|c|c}
     \midrule
        & All treat & All control & Trivially-safe & \mname-welfare & \mname-power \\
        \hline
        Est.~harm rate & 0.0392 (0.0111) & 0 & 0.001 (0.0002) & 0.0247 (0.0094) & 0.0207 (0.0086) \\
        \hline
        Num.~treatment & 667 & 0 & 16 & 605 & 615 \\        
        \hline
        Est.~welfare & 0.274 (0.0218) & 0.0996 (0.0189) & 0.104 (0.0184) & 0.254 (0.0246) & 0.250 (0.0231) \\         
        \midrule
    \end{tabular}
    \caption{Estimated harm rate, number of treatment, and estimated welfare (standard deviation) on the test fold using different methods. ``All treat'' treats all test units;  ``All control'' treats no test units;  ``Trivially-safe'' randomly treats $\alpha$-fraction of test units. The welfare/harm rate estimates are based on~\eqref{eq:real_welfare_est} and~\eqref{eq:est_fna}.}
    \label{tab:real_welfare}
\end{table}

\vspace{-0.5em}
\paragraph{Safe treatments by \mname.} 
We now take a closer look at the decisions produced by the welfare-maximizing varinat of \mname. 
Figure~\ref{fig:real_rf_sel_scatterplot} visualizes the treatment decisions, where panel (a) plots the test points based on the predicted harm rate $\hat{\fna}(X)$ and predicted treatment effect $\hat\tau(X)$, while panel (b) plots the test points based on the predicted outcomes $\hat\mu_0(X)$ and $\hat\mu_1(X)$. 
\mname treats the test units with the largest values of $\hat\tau(X)/\hat{\fna}(X)$. The decision boundary is plotted in both panels, and the region not treated is indicated in light grey. While the actual harm $\ind\{Y_{n+j}(1)<Y_{n+j}(0)\}$ is not observed, we conservatively estimate it by checking the label $Y_{n+j}^\dagger = T_{n+j}^\text{real}Y_{n+j}+(1-T_{n+j}^\text{real})(1-Y_{n+j})$, where $T_{n+j}^\text{real}$ is the actual treatment received by the $j$-th test unit, among those with $\hat{g}(X_{n+j},T_{n+j}^\text{real})=1$, i.e., either $\hat\mu_1(X_{n+j})\leq \hat\mu_0(X_{n+j})$ and $T_{n+j}^\text{real}=1$ or $\hat\mu_1(X_{n+j}) > \hat\mu_0(X_{n+j})$ and $T_{n+j}^\text{real}=0$. The colored dots are those with $\hat{g}(X_{n+j},T_{n+j}^\text{real})=1$, among which the blue dots are those with $Y_{n+j}^\dagger=1$ and the red dots are those with $Y_{n+j}^\dagger=0$. 
We observe that very few red dots in the treated region can possibly be harmed. 

\begin{figure}[htbp]
    \centering
    \includegraphics[width=0.8\linewidth]{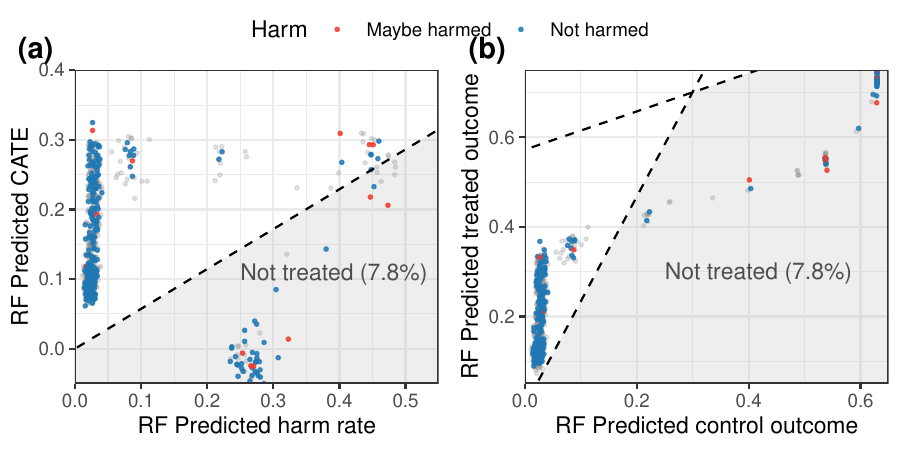}
    \caption{Visualization of treatment decisions by \mname using the welfare-maximizing variant, where outcome models are estimated by causal forests. Colored dots are those with $\hat{g}(X,T)=1$, which provide a conservative bound for harm; red dots are those who can possibly be harmed, with $Y^\dagger=0$.}
    \label{fig:real_rf_sel_scatterplot}
\end{figure}

We further examine how the treatment decision by \mname varies with participants. Figure~\ref{fig:real_rf_sel_by_pre_belief_wf} plots the moving average of predicted treated and control outcomes, as well as the fraction of safe treatments, for every value of pre-experiment belief scores (smoothed by a moving window of size $10$). 
In general, the two outcome curves indicate that a stronger pre-treatment belief makes it less likely to be persuaded out of the conspiracy in both conditions, but the effect of the AI intervention seems strong for participants with a strong pre-treatment belief. 
The welfare-maximizing variant of \mname prioritizes a treatment-effect-versus-harm-rate tradeoff. It mainly treats units with firm pre-treatment belief for whom the treatment is likely to make a huge difference (the gap between the two blue curves) while the estimated harm rate is relatively low.

\begin{figure}[!t]
    \centering
    \includegraphics[width=0.7\linewidth]{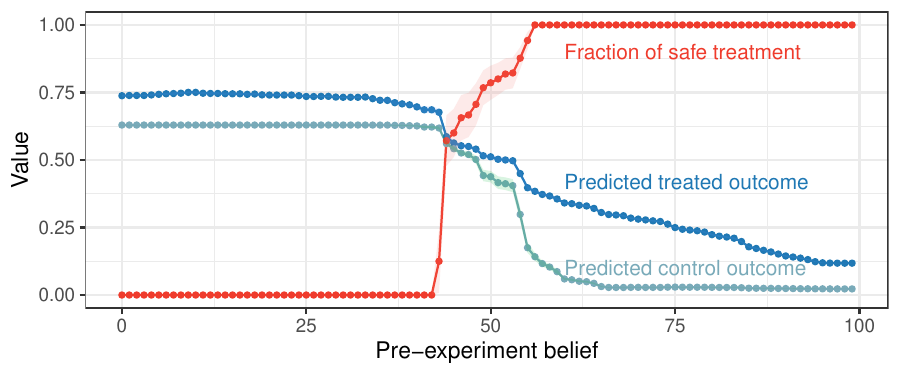}
    \caption{Fraction of safe treatment (red), average predicted treated (dark blue) and control (light blue) outcomes among test participants within a moving window of self-reported pre-experiment conspiracy belief.}
    \label{fig:real_rf_sel_by_pre_belief_wf}
\end{figure}

Figure~\ref{fig:real_rf_sel_by_fam_trust} similarly reports this information among participants with a specific GenAI familiarity score (panel a) and GenAI trust score (panel b). In this case, however, the outcomes and treatment decisions do not change significantly based on these features, indicating that AI intervention may be similarly safe for users with different familiarity with or trust in GenAI. 

\begin{figure}[!t]
    \centering
    \includegraphics[width=0.8\linewidth]{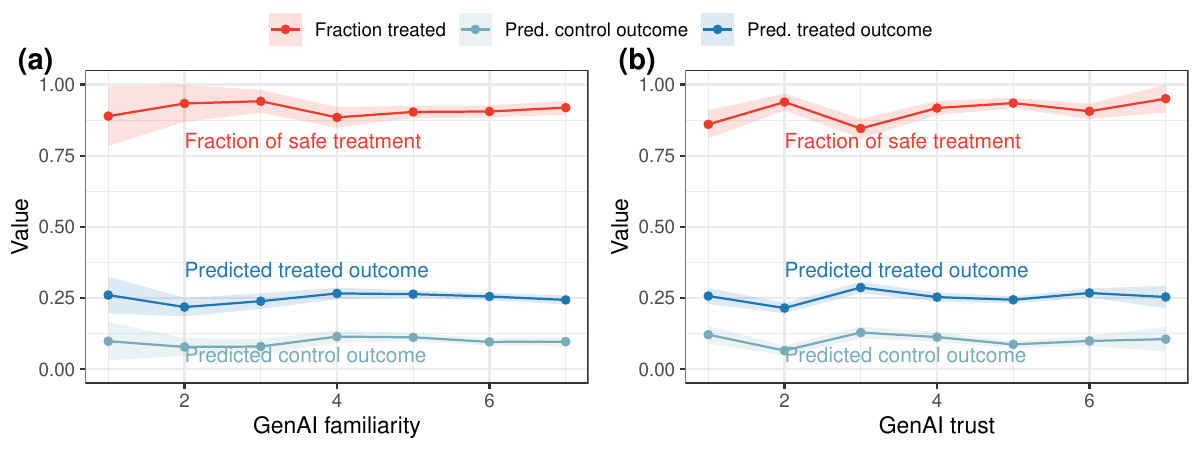}
    \caption{Fraction of treatment by the welfare-maximizing variant (red), average predicted treated (dark blue) and control (light blue) outcomes among test units stratified by self-reported GenAI-related variables.}
    \label{fig:real_rf_sel_by_fam_trust}
\end{figure}

\section{Discussion}
\label{sec:ext}
In this article, we developed the CPL approach that allows for individualized treatment assignment with a safety guarantee. We prove that when learning a treatment assignment rule from randomized experiments, CPL provides the safety guarantee in a finite sample without imposing any modeling assumption about potential outcomes. Additionally, when the outcome regression model is consistently estimated as assumed in many existing methods, CPL also asymptotically achieves the optimal welfare and power with appropriate choices of score $V$ and inclusion $g$ functions. We then extended our results to observational studies and derived novel asymptotic doubly robust safety guarantees for CPL.  

We now briefly discuss several natural extensions of our proposed conformal policy learning. The first concerns external validity settings where the test units may come from a different superpopulation than the labeled data and Assumption~\ref{assump:iid} is violated \citep[e.g.,][]{egami2023elements,jin2025beyond}. The most common and natural strategy is to relax the i.i.d. assumption to the covariate shift assumption, i.e., the distributions of the labeled data and the test data differ only in the covariate distribution. Under this setting, we can generalize CPL by simply multiplying the current conformal weights by additional weights $w_{Q/P}(x):=\ud \QQ_X/\ud \PP_X(x)$ that account for the covariate shift density ratio where $\PP$ denotes the distribution of the labeled data and $\QQ$ represents the distribution of the test data. Assuming both $e(x)$ and $w_{Q/P}(x)$ are known, CPL then proceeds in exactly the same way as in Section~\ref{sec:stratified}. When either of them is unknown, one may use similar strategies as in Section~\ref{sec:obs} to estimate and plug in these quantities.

The second natural extension concerns the control of harm rate among subgroups. In many problems where fairness and equity are stressed, it is desirable to maintain the harm rate control within each subgroup, such as those defined by demographic features~\citep{romano2020malice}. Following the setup in the main framework, our goal is to find the treatment assignment rule $\hat\pi(X_{n+1})\in \{0,1\}$ such that, for a group indicator $\cG(\cdot)$,  
$\PP( Y_{n+1}(\pi(X_{n+1})) <Y_{n+1}(0)\given \cG(X_{n+1})=1)$. 
Given a new sample with $\cG(X_{n+1})=1$, this can be achieved by taking the calibration data from the same subgroup, i.e.,  $\{(X_i,T_i,Y_i)\colon \cG(X_i)=1\}$. 
These data are induced by the full data in the subgroup $\{(X_i,Y_i(1),Y_i(0))\colon \cG(X_i)=1\}$ which is exchangeable with the new test sample conditional on the group indicator. The same \mname procedures can then be performed within the subgroup for randomized experiments and observational studies.

\section*{Acknowledgments}

The authors thank Eli Ben-Michael and Molly Offer-Westort for helpful discussions at the Online Causal Inference Seminar and the Political Methodology summer meeting, respectively. We also thank seminar participants at Yale Economics, Harvard Applied Stats Workshop, University of Tokyo Applied Stats Seminar, and the Political Methodology summer meeting. Y.J.~is partially supported by NSF DMS-2610282.

\newpage 
\bibliographystyle{apalike}
\bibliography{reference}

@article{wu2025safe,
  title={Safe individualized treatment rules with controllable harm rates},
  author={Wu, Peng and Jiang, Qing and Luo, Shanshan and Geng, Zhi},
  journal={arXiv preprint arXiv:2505.05308},
  year={2025}
}

@article{richens2022counterfactual,
  title={Counterfactual harm},
  author={Richens, Jonathan and Beard, Rory and Thompson, Daniel H},
  journal={Advances in Neural Information Processing Systems},
  volume={35},
  pages={36350--36365},
  year={2022}
}

@article{skeem2020impact,
  title={Impact of risk assessment on judges’ fairness in sentencing relatively poor defendants.},
  author={Skeem, Jennifer and Scurich, Nicholas and Monahan, John},
  journal={Law and human behavior},
  volume={44},
  number={1},
  pages={51},
  year={2020},
  publisher={Educational Publishing Foundation}
}

@article{kosorok2019precision,
  title={Precision medicine},
  author={Kosorok, Michael R and Laber, Eric B},
  journal={Annual review of statistics and its application},
  volume={6},
  number={1},
  pages={263--286},
  year={2019},
  publisher={Annual Reviews}
}

@article{laber2018identifying,
  title={{Identifying optimal dosage regimes under safety constraints: An application to long term opioid treatment of chronic pain}},
  author={Laber, Eric B and Wu, Fan and Munera, Catherine and Lipkovich, Ilya and Colucci, Salvatore and Ripa, Steve},
  journal={Statistics in medicine},
  volume={37},
  number={9},
  pages={1407--1418},
  year={2018},
  publisher={Wiley Online Library}
}

@article{wang2018learning,
  title={{Learning optimal personalized treatment rules in consideration of benefit and risk: with an application to treating type 2 diabetes patients with insulin therapies}},
  author={Wang, Yuanjia and Fu, Haoda and Zeng, Donglin},
  journal={Journal of the American Statistical Association},
  volume={113},
  number={521},
  pages={1--13},
  year={2018},
  publisher={Taylor \& Francis}
}

@article{jin2025beyond,
  title={{Beyond reweighting: On the predictive role of covariate shift in effect generalization}},
  author={Jin, Ying and Egami, Naoki and Rothenh{\"a}usler, Dominik},
  journal={Proceedings of the National Academy of Sciences},
  volume={122},
  number={45},
  pages={e2427181122},
  year={2025},
  publisher={National Academy of Sciences}
}

@article{ben2025safe,
  title={Safe policy learning through extrapolation: Application to pre-trial risk assessment},
  author={Ben-Michael, Eli and Greiner, D James and Imai, Kosuke and Jiang, Zhichao},
  journal={Journal of the American Statistical Association},
  volume={120},
  number={551},
  pages={1386--1399},
  year={2025},
  publisher={Taylor \& Francis}
}

@article{zhang2022safe,
  title={Safe policy learning under regression discontinuity designs with multiple cutoffs},
  author={Zhang, Yi and Ben-Michael, Eli and Imai, Kosuke},
  journal={arXiv preprint arXiv:2208.13323},
  year={2022}
}

@article{jia2025bayesian,
  title={Bayesian safe policy learning with chance constrained optimization: Application to military security assessment during the vietnam war},
  author={Jia, Zeyang and Ben-Michael, Eli and Imai, Kosuke},
  journal={Journal of the Royal Statistical Society Series A: Statistics in Society},
  pages={qnaf122},
  year={2025},
  publisher={Oxford University Press UK}
}

@article{costello2024durably,
  title={{Durably Reducing Conspiracy Beliefs through Dialogues with AI}},
  author={Costello, Thomas H and Pennycook, Gordon and Rand, David G},
  journal={Science},
  volume={385},
  number={6714},
  pages={eadq1814},
  year={2024},
  publisher={American Association for the Advancement of Science}
}

@article{jin2025policy,
  title={{Policy Learning “without” Overlap: Pessimism and Generalized Empirical Bernstein’s Inequality}},
  author={Jin, Ying and Ren, Zhimei and Yang, Zhuoran and Wang, Zhaoran},
  journal={The Annals of Statistics},
  volume={53},
  number={4},
  pages={1483--1512},
  year={2025},
  publisher={Institute of Mathematical Statistics}
}

@article{zhao2012estimating,
  title={{Estimating Individualized Treatment Rules using Outcome Weighted Learning}},
  author={Zhao, Yingqi and Zeng, Donglin and Rush, A John and Kosorok, Michael R},
  journal={Journal of the American Statistical Association},
  volume={107},
  number={499},
  pages={1106--1118},
  year={2012},
  publisher={Taylor \& Francis}
}

@article{murphy2003optimal,
  title={{Optimal Dynamic Treatment Regimes}},
  author={Murphy, Susan A},
  journal={Journal of the Royal Statistical Society Series B: Statistical Methodology},
  volume={65},
  number={2},
  pages={331--355},
  year={2003},
  publisher={Oxford University Press}
}

@article{hirano2009asymptotics,
  title={{Asymptotics for Statistical Treatment Rules}},
  author={Hirano, Keisuke and Porter, Jack R},
  journal={Econometrica},
  volume={77},
  number={5},
  pages={1683--1701},
  year={2009},
  publisher={Wiley Online Library}
}

@article{manski2004statistical,
  title={{Statistical Treatment Rules for Heterogeneous Populations}},
  author={Manski, Charles F},
  journal={Econometrica},
  volume={72},
  number={4},
  pages={1221--1246},
  year={2004},
  publisher={Wiley Online Library}
}

@article{heckman1997making,
  title={{Making the Most Out of Programme Evaluations and Social Experiments: Accounting for Heterogeneity in Programme Impacts}},
  author={Heckman, James J and Smith, Jeffrey and Clements, Nancy},
  journal={The Review of Economic Studies},
  volume={64},
  number={4},
  pages={487--535},
  year={1997},
  publisher={Wiley-Blackwell}
}

@article{rubin1980randomization,
  title={{Randomization Analysis of Experimental Data: The Fisher Randomization Test Comment}},
  author={Rubin, Donald B},
  journal={Journal of the American statistical association},
  volume={75},
  number={371},
  pages={591--593},
  year={1980},
  publisher={JSTOR}
}

@article{egami2023elements,
  title={{Elements of external validity: Framework, design, and analysis}},
  author={Egami, Naoki and Hartman, Erin},
  journal={American Political Science Review},
  volume={117},
  number={3},
  pages={1070--1088},
  year={2023},
  publisher={Cambridge University Press}
}

@article{bai2025llm,
  title={{LLM-generated Messages can Persuade Humans on Policy Issues}},
  author={Bai, Hui and Voelkel, Jan G and Muldowney, Shane and Eichstaedt, Johannes C and Willer, Robb},
  journal={Nature Communications},
  volume={16},
  number={1},
  pages={6037},
  year={2025},
  publisher={Nature Publishing Group UK London}
}

@article{bastani2025generative,
  title={{Generative AI without Guardrails can Harm Learning: Evidence from high School Mathematics}},
  author={Bastani, Hamsa and Bastani, Osbert and Sungu, Alp and Ge, Haosen and Kabakc{\i}, {\"O}zge and Mariman, Rei},
  journal={Proceedings of the National Academy of Sciences},
  volume={122},
  number={26},
  pages={e2422633122},
  year={2025},
  publisher={National Academy of Sciences}
}

@article{dell2023navigating,
  title={{Navigating the Jagged Technological Frontier: Field Experimental Evidence of the Effects of AI on Knowledge Worker Productivity and Quality}},
  author={Dell'Acqua, Fabrizio and McFowland III, Edward and Mollick, Ethan R and Lifshitz-Assaf, Hila and Kellogg, Katherine and Rajendran, Saran and Krayer, Lisa and Candelon, Fran{\c{c}}ois and Lakhani, Karim R},
  journal={Organization Science},  
  year={2026}
}

@article{holland1986statistics,
  title={{Statistics and Causal Inference}},
  author={Holland, Paul W.},
  journal={Journal of the American Statistical Association},
  volume={81},
  number={396},
  pages={945--960},
  year={1986},
  publisher={Taylor \& Francis},
  doi={10.1080/01621459.1986.10478354}
}

@book{imbens2015causal, 
  place={Cambridge}, 
  title={Causal Inference for Statistics, Social, and Biomedical Sciences: An Introduction}, 
  DOI={10.1017/CBO9781139025751}, 
  publisher={Cambridge University Press}, 
  author={Imbens, Guido W. and Rubin, Donald B.}, 
  year={2015}
}

@inproceedings{TibshiraniBCR19,
  author    = {Ryan J. Tibshirani and
               Rina Foygel Barber and
               Emmanuel J. Cand{\`{e}}s and
               Aaditya Ramdas},
  title     = {{Conformal Prediction Under Covariate Shift}},
  booktitle = {Advances in Neural Information Processing Systems 32},
  pages     = {2526--2536},
  year      = {2019},
  url       = {https://proceedings.neurips.cc/paper/2019/hash/8fb21ee7a2207526da55a679f0332de2-Abstract.html},
}

@article{bates2021testing,
  title={Testing for outliers with conformal p-values},
  author={Bates, Stephen and Cand{\`e}s, Emmanuel and Lei, Lihua and Romano, Yaniv and Sesia, Matteo},
  journal={arXiv preprint arXiv:2104.08279},
  year={2021}
}

@book{vovk2005algorithmic,
  title={Algorithmic learning in a random world},
  author={Vovk, Vladimir and Gammerman, Alexander and Shafer, Glenn},
  year={2005},
  publisher={Springer Science \& Business Media}
}

@article{lei2018distribution,
  title={Distribution-free predictive inference for regression},
  author={Lei, Jing and G’Sell, Max and Rinaldo, Alessandro and Tibshirani, Ryan J and Wasserman, Larry},
  journal={Journal of the American Statistical Association},
  volume={113},
  number={523},
  pages={1094--1111},
  year={2018},
  publisher={Taylor \& Francis}
}

@inproceedings{xu2024online,
  title={Online multiple testing with e-values},
  author={Xu, Ziyu and Ramdas, Aaditya},
  booktitle={International Conference on Artificial Intelligence and Statistics},
  pages={3997--4005},
  year={2024},
  organization={PMLR}
}

@article{jin2023selection,
  title={Selection by prediction with conformal p-values},
  author={Jin, Ying and Cand{\`e}s, Emmanuel J},
  journal={Journal of Machine Learning Research},
  volume={24},
  number={244},
  pages={1--41},
  year={2023}
}

@article{jin2023model,
  title={Model-free selective inference under covariate shift via weighted conformal p-values},
  author={Jin, Ying and Cand{\`e}s, Emmanuel J},
  journal={arXiv preprint arXiv:2307.09291},
  year={2023}
}

@article{bai2024optimized,
  title={Optimized Conformal Selection: Powerful Selective Inference After Conformity Score Optimization},
  author={Bai, Tian and Jin, Ying},
  journal={arXiv preprint arXiv:2411.17983},
  year={2024}
}

@article{lei2021conformal,
  title={Conformal inference of counterfactuals and individual treatment effects},
  author={Lei, Lihua and Cand{\`e}s, Emmanuel J},
  journal={Journal of the Royal Statistical Society Series B: Statistical Methodology},
  volume={83},
  number={5},
  pages={911--938},
  year={2021},
  publisher={Oxford University Press}
}

@article{jin2023sensitivity,
  title={Sensitivity analysis of individual treatment effects: A robust conformal inference approach},
  author={Jin, Ying and Ren, Zhimei and Cand{\`e}s, Emmanuel J},
  journal={Proceedings of the National Academy of Sciences},
  volume={120},
  number={6},
  pages={e2214889120},
  year={2023},
  publisher={National Academy of Sciences}
}

@inproceedings{li2023trustworthy,
  title={Trustworthy policy learning under the counterfactual no-harm criterion},
  author={Li, Haoxuan and Zheng, Chunyuan and Cao, Yixiao and Geng, Zhi and Liu, Yue and Wu, Peng},
  booktitle={International Conference on Machine Learning},
  pages={20575--20598},
  year={2023},
  organization={PMLR}
}

@article{kallus2022s,
  title={What's the harm? sharp bounds on the fraction negatively affected by treatment},
  author={Kallus, Nathan},
  journal={Advances in Neural Information Processing Systems},
  volume={35},
  pages={15996--16009},
  year={2022}
}

@article{kitagawa2018should,
  title={Who should be treated? empirical welfare maximization methods for treatment choice},
  author={Kitagawa, Toru and Tetenov, Aleksey},
  journal={Econometrica},
  volume={86},
  number={2},
  pages={591--616},
  year={2018},
  publisher={Wiley Online Library}
}

@article{athey2021policy,
  title={Policy learning with observational data},
  author={Athey, Susan and Wager, Stefan},
  journal={Econometrica},
  volume={89},
  number={1},
  pages={133--161},
  year={2021},
  publisher={Wiley Online Library}
}

@article{zubizarreta2015stable,
  title={Stable weights that balance covariates for estimation with incomplete outcome data},
  author={Zubizarreta, Jos{\'e} R},
  journal={Journal of the American Statistical Association},
  volume={110},
  number={511},
  pages={910--922},
  year={2015},
  publisher={Taylor \& Francis}
}

@article{jin2025cross,
  title={Cross-Balancing for Data-Informed Design and Efficient Analysis of Observational Studies},
  author={Jin, Ying and Zubizarreta, Jos{\'e}},
  journal={arXiv preprint arXiv:2511.15896},
  year={2025}
}

@misc{chernozhukov2018double,
  title={Double/debiased machine learning for treatment and structural parameters},
  author={Chernozhukov, Victor and Chetverikov, Denis and Demirer, Mert and Duflo, Esther and Hansen, Christian and Newey, Whitney and Robins, James},
  year={2018},
  publisher={Oxford University Press Oxford, UK}
}

@article{hainmueller2012entropy,
  title={Entropy balancing for causal effects: A multivariate reweighting method to produce balanced samples in observational studies},
  author={Hainmueller, Jens},
  journal={Political analysis},
  volume={20},
  number={1},
  pages={25--46},
  year={2012},
  publisher={Cambridge University Press}
}

@article{romano2020malice,
  title={With malice toward none: Assessing uncertainty via equalized coverage},
  author={Romano, Yaniv and Barber, Rina Foygel and Sabatti, Chiara and Cand{\`e}s, Emmanuel},
  journal={Harvard Data Science Review},
  volume={2},
  number={2},
  pages={4},
  year={2020},
  publisher={MIT Press-Journals}
}

@article{wu2024quantifying,
  title={Quantifying individual risk for binary outcome},
  author={Wu, Peng and Ding, Peng and Geng, Zhi and Liu, Yue},
  journal={arXiv preprint arXiv:2402.10537},
  year={2024}
}

@article{yin2018assessing,
  title={Assessing the treatment effect heterogeneity with a latent variable},
  author={Yin, Yunjian and Liu, Lan and Geng, Zhi},
  journal={Statistica Sinica},
  pages={115--135},
  year={2018},
  publisher={JSTOR}
}

@article{huang2012assessing,
  title={Assessing treatment-selection markers using a potential outcomes framework},
  author={Huang, Ying and Gilbert, Peter B and Janes, Holly},
  journal={Biometrics},
  volume={68},
  number={3},
  pages={687--696},
  year={2012},
  publisher={Oxford University Press}
}

@article{shen2013treatment,
  title={Treatment benefit and treatment harm rate to characterize heterogeneity in treatment effect},
  author={Shen, Changyu and Jeong, Jaesik and Li, Xiaochun and Chen, Peng-Sheng and Buxton, Alfred},
  journal={Biometrics},
  volume={69},
  number={3},
  pages={724--731},
  year={2013},
  publisher={Oxford University Press}
}

@article{zhang2013assessing,
  title={Assessing the heterogeneity of treatment effects via potential outcomes of individual patients},
  author={Zhang, Zhiwei and Wang, Chenguang and Nie, Lei and Soon, Guoxing},
  journal={Journal of the Royal Statistical Society Series C: Applied Statistics},
  volume={62},
  number={5},
  pages={687--704},
  year={2013},
  publisher={Oxford University Press}
}

@article{gadbury2004individual,
  title={Individual treatment effects in randomized trials with binary outcomes},
  author={Gadbury, Gary L and Iyer, Hari K and Albert, Jeffrey M},
  journal={Journal of Statistical Planning and Inference},
  volume={121},
  number={2},
  pages={163--174},
  year={2004},
  publisher={Elsevier}
}

@article{ben2024policy,
  title={Policy learning with asymmetric counterfactual utilities},
  author={Ben-Michael, Eli and Imai, Kosuke and Jiang, Zhichao},
  journal={Journal of the American Statistical Association},
  volume={119},
  number={548},
  pages={3045--3058},
  year={2024},
  publisher={Taylor \& Francis}
}

@article{yin2024conformal,
  title={Conformal sensitivity analysis for individual treatment effects},
  author={Yin, Mingzhang and Shi, Claudia and Wang, Yixin and Blei, David M},
  journal={Journal of the American Statistical Association},
  volume={119},
  number={545},
  pages={122--135},
  year={2024},
  publisher={Taylor \& Francis}
}

@article{qian2011performance,
  title={Performance guarantees for individualized treatment rules},
  author={Qian, Min and Murphy, Susan A},
  journal={Annals of statistics},
  volume={39},
  number={2},
  pages={1180},
  year={2011}
}

@inproceedings{li2010contextual,
  title={A contextual-bandit approach to personalized news article recommendation},
  author={Li, Lihong and Chu, Wei and Langford, John and Schapire, Robert E},
  booktitle={Proceedings of the 19th international conference on World wide web},
  pages={661--670},
  year={2010}
}

@article{dudik2011doubly,
  title={Doubly robust policy evaluation and learning},
  author={Dud{\'\i}k, Miroslav and Langford, John and Li, Lihong},
  journal={arXiv preprint arXiv:1103.4601},
  year={2011}
}

@article{imai2011estimation,
  title={Estimation of heterogeneous treatment effects from randomized experiments, with application to the optimal planning of the get-out-the-vote campaign},
  author={Imai, Kosuke and Strauss, Aaron},
  journal={Political Analysis},
  volume={19},
  number={1},
  pages={1--19},
  year={2011},
  publisher={Cambridge University Press}
}

@article{kleinberg2018human,
  title={Human decisions and machine predictions},
  author={Kleinberg, Jon and Lakkaraju, Himabindu and Leskovec, Jure and Ludwig, Jens and Mullainathan, Sendhil},
  journal={The quarterly journal of economics},
  volume={133},
  number={1},
  pages={237--293},
  year={2018},
  publisher={Oxford University Press}
}

@article{scauda2026counterfactual,
  title={Counterfactual Optimization of Policy Interventions: Lexical Ordering and Leapfrogging},
  author={Scauda, Martina and Freidling, Tobias and Zhao, Qingyuan},
  journal={arXiv preprint arXiv:2608.20505},
  year={2026}
}

% \newpage 
\appendix 
% !TEX root = main.tex

\section{Deferred discussion}

\subsection{Power Maximization} 
\label{app:subsec_power_opt}
We first define a natural notion of power, the probability of getting treated: 
\$
\text{Power}(\pi; P) := P(\pi(X_{n+1})=1).
\$ 
The optimal treatment rule under this power notion is defined as 
\@\label{eq:def_opt_power_phi}
\pi_{\text{power}}^* = \argmax_{\pi\colon \cX\to \{0,1\}} \quad & \text{Power}  (\pi; \PP_{X,Y(1),Y(0)})  \\ 
\text{subject to}\quad & \max_{P\in \cP} ~\text{Err}(\pi;P) \leq \alpha. \notag
\@
We can rewrite this optimization problem as follows. 
\@\label{eq:def_opt_power_phi_2}
\pi_{\text{power}}^* = \argmax_{\pi \colon \cX\to \{0,1\}} \quad & \E[\pi(X_{n+1})]  \\ 
\text{subject to}\quad & \E[\pi(X_{n+1}) \fna(X_{n+1})] \leq \alpha. \notag
\@
where $\fna(x) = \min\{1-\mu_1(x), \mu_0(x)\},$ which is the sharp upper bound for $\PP(Y_{n+1}(1) = 0, Y_{n+1}(0) = 1 \mid X_{n+1} = x).$ Intuitively, given the linear relaxation, the optimal treatment rule is a thresholding rule based on $\fna(X_{n+1})$, which acts as the ``cost'' in this optimization problem. 

Theorem~\ref{thm:global_opt}  formally establishes that the optimal solution $\pi^*_{\text{power}}$ under the worst-case harm constraint is based on a cutoff on $\fna(x)$. It is implied by a more general result in Theorem~\ref{thm:global_opt_general} in Appendix~\ref{app:subsec_general_opt} with proof in Appendix~\ref{app:subsec_proof_global_opt_general}; we thus omit the proof of Theorem~\ref{thm:global_opt} here.

\begin{theorem}[Power-optimal $\pi^*_{\text{power}}$] \label{thm:global_opt}    
    Assume $\fna(X)$ has no point mass. Then, the optimal solution~\eqref{eq:def_opt_power_phi} is $\pi^*_{\text{power}}(x) = \ind\{\fna(x)\leq \fna^*\}$ for the cutoff $\fna^* = \max\{\tilde{\fna} \in [0,1]\colon \EE[\fna(X)\ind\{\fna(X)\leq \tilde{\fna})\}]\leq \alpha\}$. 
\end{theorem}

To achieve the optimal power, we define the power-optimal conformal p-value (for   stratified experiments, covering the complete randomized experiments in Section~\ref{sec:rct} as a special case)
\@
p^{\text{str-power}}_{n+1} = \frac{w(X_{n+1})+ \sum_{i=1}^n w(X_i) G_i   \ind\{Y_i^\dagger = 0\}  \ind\{\widehat{\fna}(X_i)\leq \widehat{\fna}(X_{n+1})\}}{w(X_{n+1})+\sum_{i=1}^n G_iw(X_i)}, \label{eq:pval_power_optimal}
\@
where $G_i = T_i \one\{1 - \widehat{\mu}_1(X_i) \leq \widehat{\mu}_0(X_{i})\} + (1-T_i) \one\{1 - \widehat{\mu}_1(X_i) > \widehat{\mu}_0(X_{i})\}.$ 
The choice of the calibration inclusion indicators $G_i$ is the same as the welfare-optimal procedure. 
This coincides with Algorithm~\ref{algo:rct} with the clipped score $V(x,y) = M\ind\{y>0\} +  \widehat{\fna}(x)$ using  a sufficiently large constant $M > 0.$

Theorem~\ref{thm:power_opt_str} establishes the optimality of PCL with the power-oriented conformity score. Its proof is in Appendix~\ref{app:subsec_opt_power_ours}. 

\begin{theorem}\label{thm:power_opt_str}
    
Suppose Assumption~\ref{assump:uncon} and Assumption~\ref{assump:iid} hold, and $e(x)$ is known. Suppose $\|\hat{\mu}_t(X) -\mu_t(X)\|_{L_2(\PP_{X})}\stackrel{P}{\to}0$ as $n\to \infty$ for $t \in \{0,1\}$, and  $\fna(X)$ has no point mass. 
    Let $\hat\pi_{\text{str-power}}(X_{n+1}) \coloneqq \ind\{p_{n+1}^{\text{str-power}} \leq \alpha\}$ for the p-value in~\eqref{eq:pval_power_optimal}. Then $\EE[Y(\hat\pi_{\text{str-power}}(X_{n+1}))]\to \text{Power}(\pi_\text{power}^*;\PP)$ as $n\to \infty$, where $\pi_{\text{power}}^*$ is the power-optimal solution in Theorem~\ref{thm:global_opt}.
\end{theorem}

The theorem below establishes asymptotic power optimality under suitable conditions, in parallel to Theorem~\ref{thm:obs_welfare_opt}. The proof is in Appendix~\ref{app:subsec_obs_power_opt}. 

\begin{theorem}[Power optimality with observational data]
\label{thm:obs_power_opt}
Let $\hat\pi_{\obs}(X_{n+1})=\ind\{p^{\obs}_{n+1}\leq\alpha\}$,
where we take the same $G_i$ as~\eqref{eq:def_G_welfare_optimal} and
the score $\hat{s}(x) = \hat\fna(x) =\min\{1-\hat\mu_1(x),\hat\mu_0(x)\}$.  
Suppose Assumption~\ref{assump:bal_reg} holds, and $\delta_n = O(n^{-1/2})$. 
Furthermore, assume $\|\hat g-g^*\|_{L_2(\PP_{X,T})}+\|\hat\fna - \fna\|_{L_2(\PP_{X})}+\|\hat s-\fna\|_{L_2(\PP_X)}=o_P(1)$ for $g^*(x,t)=t\ind\{1-\mu_1(x)\leq\mu_0(x)\}+(1-t)\ind\{1-\mu_1(x)>\mu_0(x)\}$, and $\fna(X)$ has no point mass.
Then
$\EE[\hat\pi_{\obs}(X_{n+1})]
\to
\textnormal{Power}(\pi_{\power}^*;\PP)$, 
where $\pi_{\power}^*$ is the power-optimal solution in
Theorem~\ref{thm:global_opt}.
\end{theorem}

\subsection{General form of optimality}
\label{app:subsec_general_opt}

Theorem~\ref{thm:global_opt_general} is a general form of Theorem~\ref{thm:global_opt}, whose proof is in Appendix~\ref{app:subsec_proof_global_opt_general}. 

\begin{theorem}[Power-optimal $\phi^*$, general form] \label{thm:global_opt_general}
    Let $\fna(x) = \min\{\PP(Y(1)=0\given X=x), \PP(Y(0)=1\given X=x)\}$. 
    If $\mathbb{E}[\fna(X)]\leq\alpha$, then an optimal solution to~\eqref{eq:def_opt_power_phi} is $\phi^*_{\mathrm{power}}(x)\equiv1$. Otherwise, define  $\gamma^* := \inf\left\{ c\in[0,1]: \mathbb{E}\!\left[ \gamma(X)\mathbf{1}\{\gamma(X)\leq c\} \right] \geq\alpha \right\}$, and let $\eta^*\in[0,1]$ be chosen such that   $\mathbb{E}\!\left[ \gamma(X)\mathbf{1}\{\gamma(X)<\gamma^*\} \right] + \eta^* \mathbb{E}\!\left[ \gamma(X)\mathbf{1}\{\gamma(X)=\gamma^*\} \right] = \alpha$. Then an optimal solution to~\eqref{eq:def_opt_power_phi} is $\phi^*_{\mathrm{power}}(x) = \mathbf{1}\{\gamma(x)<\gamma^*\} + \eta^*\mathbf{1}\{\gamma(x)=\gamma^*\}$. In this case, the safety constraint is binding: $\mathbb{E}[\gamma(X)\phi^*_{\mathrm{power}}(X)]=\alpha$. 
    % Then, the optimal solution to~\eqref{eq:def_opt_power_phi} takes the form 
    % \$
    % \phi^*(x) = \begin{cases}
    %     1, &~ \fna(x)<\fna^*, \\ 
    %     \eta, & ~ \fna(x)=\fna^*, \\ 
    %     0, &~ \fna(x)>\fna^*,
    % \end{cases}
    % \$
    % so that $\phi^*(\cdot)$ becomes a randomized decision rule which gives the treatment with probability $\eta\in [0,1]$ when $\fna(X)=\fna^*$, and $\eta\in [0,1]$ is chosen such that $\EE[\fna(X)\phi^*(X)]=\alpha$. 
\end{theorem}

Theorem~\ref{thm:welfare_optimal_general} is a general form of Theorem~\ref{thm:welfare_optimal}, whose proof is in Appendix~\ref{app:subsec_proof_welfare_opt}. 

\begin{theorem} 
\label{thm:welfare_optimal_general}
Consider the randomized extension of problem~\eqref{eq:def_opt_welfare_phi}, where a policy is a measurable function $\phi:\mathcal X\to[0,1]$, and $\phi(x)$ denotes the probability of treatment conditional on $X=x$. For $r\leq 0$, define $H(r):=\mathbb{E}[\fna(X)  \ind\{s_{\mathrm{welfare}}(X)<r\}]$. If $H(0)\leq\alpha$, let $r^*=0$ and $\eta^*=0$. If $H(0)>\alpha$, let $r^* := \sup\left\{ r<0: H(r)\leq\alpha \right\}$, and choose $\eta^*\in[0,1]$ such that $H(r^*) + \eta^* \mathbb{E}\!\left[ \fna(X)  \ind\{s_{\mathrm{welfare}}(X)=r^*\} \right] = \alpha$. 
Then an optimal solution to the randomized extension of~\eqref{eq:def_opt_welfare_phi} is 
\[ 
\phi^*_{\mathrm{welfare}}(x) =  \ind\{s_{\mathrm{welfare}}(x)<r^*\} + \eta^*  \ind\{s_{\mathrm{welfare}}(x)=r^*\}. 
\] 
Moreover, $\mathbb{E}[\gamma(X)\phi^*_{\mathrm{welfare}}(X)]\leq\alpha$ and $r^* \cdot\left\{ \mathbb{E}[\gamma(X)\phi^*_{\mathrm{welfare}}(X)] -\alpha \right\} =0$. In particular, if $\mathbb{E}[\gamma(X) \ind\{\tau(X)>0\}]\leq\alpha$, then $r^*=0$ and $\phi^*_{\mathrm{welfare}}(x)= \ind\{\tau(x)>0\}$. Otherwise, $r^*<0$ and the safety constraint is exactly attained.
\end{theorem}

\subsection{A general theory for doubly robust guarantee with estimated weights}
\label{app:subsec_general_dr}

In this section, we present the general theory for CPL with estimated weights. 
We consider any estimated weights
$\{\hat{w}_i\}_{i=1}^{n+1}$, which may come from a pre-trained weighted function or whose estimation may depend on the data. 
We recall the definition of conformal p-values:
\@\label{eq:def_pval_weight_est_app}
p_{n+1}^{\text{obs}} = \frac{\hat{w}_{n+1} +\sum_{i=1}^n \hat{w}_i   G_i  \ind\{Y_i^\dagger = 0\}  \ind\{\hat{s}(X_i)\leq \hat{s}(X_{n+1})\}}{\hat{w}_{n+1} + \sum_{i=1}^n \hat{w}_i G_i},
\@
where the inclusion indicators are $G_i \sim \text{Bern}(\hat{g}(X_i,T_i))$ for a function $\hat{g}\colon \cX\times\{0,1\}\to [0,1]$ whose training process is independent of the labeled and unlabeled data. 

% \nec{@Ying, I did some cleanup of notations. Please edit if you prefer otherwise. I use $w(X_i)$ to denote the true weight function to keep the parallel notation in the stratified experiment case. I then use $\hat{w}(X_i)$ to denote our estimated weight function. I use $\bar{w}(X_i)$ as the limit of $\hat{w}(X_i)$ without assuming that it coincides with the true weight function. In particular, I do not use $\hat{w}_i$. We can use simplified notations $\widehat{w}_i$ in the proof in the appendix, but I want to use more expressive notation here in the main paper. }

We introduce several high-level conditions that the weight function and other nuisance functions need to satisfy in order to achieve the doubly robust safety guarantee. Importantly, we will later show that our learn-then-balance weights satisfy these conditions under mild model convergence conditions.

First, we require the weights to satisfy the following \emph{approximate balancing condition}. This will involve an estimator $\hat{\fna}(\cdot)$ of the worst-case harm rate function $\fna(x) = \min\{1-\mu_1(x),\mu_0(x)\}$ estimated with data independent of the calibration and test data, such as by plugging in independently estimated outcome models.

\begin{assumption}[Approximate balance]\label{assump:obs_bal}
Let $r_n>0$ be a deterministic sequence obeying $r_n=o(1)$. 
The weights obey  $\hat{w}_i\geq 0$, 
$\frac{1}{|\cI_{\calib}|} \sum_{i\in \cI_{\calib}} (\hat{w}_i - \hat{\omega}(X_i))^2 = O_P(r_n^2)$ for some function $\hat{\omega}(\cdot)\colon \cX\to \RR$ trained independently of the calibration data and test point, and $\frac{\hat{w}_{n+1}\vee \max_{i\in\cI_{\calib}}\hat w_i}{\sum_{i\in \cI_{\calib}}\hat{w}_i} = O_P(1/n)$. In addition, for $\hat{t} = \sup \{t\in \RR\colon \frac{1}{n}\sum_{i=1}^n \hat{\fna}(X_i) \ind\{\hat{s}(X_i)\leq \hat{t}\} \leq \alpha \}$, it holds that 
    \@\label{eq:approx_balance}
&\big|  \textstyle{\frac{1}{|\cI_{\calib}|} \sum_{i\in \cI_{\calib}}} \hat{w}_i \hat{\fna}(X_i)\ind\{\hat{s}(X_i)\leq \hat{t} \} - \textstyle{\frac{1}{n}\sum_{i=1}^n} \hat{\fna}(X_i) \ind\{\hat{s}(X_i)\leq \hat{t}\} \big| =  O_P(r_n),\quad \\ 
&\text{and}\qquad \textstyle{\frac{1}{|\cI_{\calib}|} \sum_{i\in \cI_{\calib}}}\hat{w}_i = 1+ O_P(r_n) .\notag 
\@ 
\end{assumption}

In Assumption~\ref{assump:obs_bal}, the first condition requires the estimated weights $\{\hat{w}_i\}$ to converge to any independently trained weight function with parametric rates. This is easily satisfied if one set $\hat{w}_i=\hat{w}(X_i)$; we shall see that our balancing program also satisfies this general condition.
The key balancing condition requires the estimated weights to approximately balance the capped score functions at the critical cutoff $\hat{t}$, %This reflects the requirements on the true weights $w(X_i)$ in~\eqref{eq:def_pval_weight}: 
% 
% \begin{assumption}[Approximate balance]\label{assump:obs_bal}
% The weights obey  
% $\frac{1}{|\cI_{\calib}|} \sum_{i\in \cI_{\calib}} (\hat{w}(X_i) - \bar{w}(X_i))^2 = O_P(1/n)$ for some function $\bar{w}(\cdot)\colon \cX\to \RR$ trained independently of the calibration data and test point, and $\frac{\hat{w}(X_{n+1})}{\sum_{i\in \cI_{\calib}}\hat{w}(X_i)} = o_P(1/\sqrt{n})$. In addition, for $\hat{t} = \sup \{t\in \RR\colon \frac{1}{n}\sum_{i=1}^n \hat{\fna}(X_i) \ind\{\hat{s}(X_i)\leq  {t}\} \leq \alpha \}$, it holds that 
%     \@\label{eq:approx_balance}
% & \textstyle{\frac{1}{|\cI_{\calib}|} \sum_{i\in \cI_{\calib}}} \hat{w}(X_i) \hat{\fna}(X)\ind\{\hat{s}(X)\leq \hat{t} \} = \textstyle{\frac{1}{n}\sum_{i=1}^n} \hat{\fna}(X_i) \ind\{\hat{s}(X_i)\leq \hat{t}\} + o_P(1/\sqrt{n}),\quad \\ 
% &\text{and}\qquad \textstyle{\frac{1}{|\cI_{\calib}|} \sum_{i\in \cI_{\calib}}}\hat{w}(X_i) = 1+ o_P(1/\sqrt{n}) .\notag 
% \@ 
% \end{assumption}
 % Intuitively, the ``correct'' weights $w(X_i)$ (up to proper normalization) 
% it should ensure~\eqref{eq:approx_balance} (and in fact the balance of any function of $X$ between two groups) due to the Markov's inequality. To protect against model mis-specification in weight estimation, the above finite-sample balance condition enforces the score-relevant function must be balanced in finite samples, 
following ideas in the balancing weights literature~\citep{hainmueller2012entropy,zubizarreta2015stable,jin2025cross} but with specific choice of the balancing features to yield favorable statistical properties. %We shall see in Proposition~\ref{prop:bal_conv} that this condition can be achieved by simply processing any pre-trained weight and score functions. 

% The second condition is a mild model convergence assumption that the estimated functions converge to any unknown, deterministic functions that are not necessarily the true counterparts. 

% \begin{assumption}[Convergence]\label{assump:obs_conv}
%    For the function $\hat{\omega}(\cdot)$ in Assumption~\ref{assump:obs_bal}, there exists some fixed function $\bar{w}\colon \cX\to \RR$ such that $\|\hat{\omega}(\cdot)-\bar{w}(\cdot)\|_{L_2(\PP_X)}=o_P(1)$.  
%     In addition, $\|\hat{\fna}(\cdot)- \bar{\fna}(\cdot)\|_{L_2(\PP_X)}=o_P(1)$  
%     %and $\|\hat{g} (\cdot,\cdot)-\bar{g}(\cdot,\cdot)\|_{L_2(\PP_{X,T})}=o_P(1)$ 
%     for some fixed function $\bar{\fna}\colon \cX\to \RR$. % and  $\bar{g}\colon \cX\times\{0,1\}\to [0,1]$. 
% \end{assumption}

The following theorem shows that, under Assumptions~\ref{assump:obs_bal}, the conformal policy learning $\hat{\pi}_{\text{obs}}(X_{n+1}) \coloneqq \ind\{p_{n+1}^{\text{obs}} \leq \alpha\}$ achieves the safety guarantees asymptotically if either the outcome conditional expectation function $\hat{\mu}_t(x)$ or the weight function $\hat{w}(x)$, but not necessarily both, is consistently estimated. The proof of Theorem~\ref{thm:dr_obs_general} is in Appendix~\ref{app:subsec_dr_obs_general}.  
Recall that $\cT_n$ is the $\sigma$-algebra of the training process.

\begin{theorem}[Model double robustness]\label{thm:dr_obs_general}
Suppose %$\EE[\bar{w}(X)\hat{g}(X,T)]>0$, and 
Assumption~\ref{assump:obs_bal} holds for some $r_n = o(1)$. Define $p_G:=\PP(G=1\mid\cT_n)$ and
$w^\circ(x):=p_Gw(x)$. Then, 
% Assume the following regularity conditions hold. (1)
%     $t\mapsto \bar{F}(t):=\frac{  \EE[\bar{w}(X) \hat{g}(X,T)\bar\gamma(X)\ind\{\hat{s}(X)\leq t\}]}{\EE[\bar{w}(X)\hat{g}(X,T)]}$ is strictly increasing at $\tau^*:=\sup\{t\in \RR\colon \bar{F}(t)\leq \alpha\}$. (2) $\EE[\bar{w}(X)\hat{g}(X,T)]>0$. (2) $t\mapsto \tilde{F}(t)=\EE[\hat{\fna}(X)\ind\{\hat{s}(X)\leq t\}]$ also has a positive local slope at $\hat{t}^* := \sup\big\{t\in \RR\colon \tilde{F}(t) \leq \alpha\big\}$. (3) $\sup_x |1/w(x)|<\infty$,  $\sup_x|\hat{\fna}(x)|\leq M_1$ for some constant $M_1>0$. (4) $\hat{s}(X)$ has no point mass, and the density of $\hat{s}(X)$ is uniformly upper bounded by a constant $M_2>0$.
    %, there exists some constant $c>0$ and $\delta>0$ such that $\tilde{F}(\hat{t}^*-u)\leq \alpha-cu$ and $\tilde{F}(\hat{t}^*+u)\geq \alpha + cu$. 
    % Then, under Assumptions~\ref{assump:obs_bal} and~\ref{assump:obs_conv}, we have 
    $$\limsup_{n\to \infty}\PP(Y_{n+1}(\hat{\pi}_{\textnormal{obs}}(X_{n+1}))<Y_{n+1}(0))\leq \alpha$$ under either of the following conditions:
    \begin{enumerate}[label=(\roman*).]
        \item The weight $\hat\omega(\cdot)$ obeys  $\|\hat\omega - w^\circ\|_{L_2(\PP_X)}=o_P(1)$, and there
exist constants $c,C >0$ such that $p_G\geq c$ and
$\|w^\circ\|_{L_2(\PP_X)}\leq C$ with probability tending to one.%  $\bar{w}(\cdot) \propto w(\cdot)$ for the true weight function $w(\cdot)$ in~\eqref{eq:def_weight}.
        \item The outcome models obey $\|\hat\fna - \fna\|_{L_2(\PP_X)}=o_P(1)$, and there
exist constants $c,C,\eta>0$ such that, with probability
tending to one,
$q(x)\geq c$, $c\leq\hat\omega(x)\leq C$, $x\in\cX$, and 
the function
$H_n(t):=
\EE[\hat\gamma(X)\ind\{\hat s(X)\leq t\}\mid\cT_n]$ 
satisfies
$H_n(t_n^\circ-u)\leq\alpha-cu$ and 
$H_n(t_n^\circ+u)\geq\alpha+cu$ for  $0<u\leq\eta$, 
where $t_n^\circ:=\sup\{t\in\RR:H_n(t)\leq\alpha\}$, and
$\PP\{a<\hat s(X)\leq b\mid\cT_n\}\leq C(b-a)$
for every $a<b$.
        % the following regularity conditions hold: $t\mapsto \bar{F}(t):=\frac{  \EE[\bar{w}(X) \hat{g}(X,T)\bar\gamma(X)\ind\{\hat{s}(X)\leq t\}]}{\EE[\bar{w}(X)\hat{g}(X,T)]}$ is strictly increasing at $\tau^*:=\sup\{t\in \RR\colon \bar{F}(t)\leq \alpha\}$, and $\hat{s}(X)$ has no point mass. 
    \end{enumerate}
\end{theorem}

Notably, the doubly robust safety guarantees do not require the score function $\hat{s}(\cdot)$ to converge to any true function, which inherits the model-free nature of conformal prediction. The asymptotic safety guarantee hinges on the convergence of $\hat{w}(\cdot)$ and/or $\hat{\fna}(\cdot)$. In particular, in the case where the weights do not converge to $w(\cdot)$, consistent outcome models and the balancing condition still ensure harm rate control.

We further show that CPL with balanced weights achieves rate double robustness: if both the outcome models and the weight functions are consistently estimated at slow, nonparametric rates $o_P(n^{-1/4})$, the excess harm rate above $\alpha$ is of the parametric order $O(1/\sqrt{n}).$ The proof of Theorem~\ref{thm:dr_obs_rate_general} is in Appendix~\ref{app:subsec_obs_dr_rate_general}. Recall that 
$p_G:=\PP(G=1\mid\cT_n)$ and 
$w^\circ(x):=p_Gw(x)$. 

% \begin{assumption}[Slow convergence to correct models]\label{assump:obs_conv_rate}
%     For the function $\hat{w}(\cdot)$ in Assumption~\ref{assump:obs_bal}, it holds that $\|\hat{w}(\cdot) - w(\cdot)\|_{L_\infty(\PP)}=O_P(n^{-1/4})$. In addition,  $\|\hat{\fna}(\cdot)-\fna(\cdot)\|_{L_1(\PP)}=O_P(n^{-1/4})$.
% \end{assumption}

\begin{assumption}[Slow convergence]
\label{assump:obs_slow}
Let $\hat\omega(\cdot)$ be the reference function in
Assumption~\ref{assump:obs_bal}, and 
Assume  
$\|\hat\omega-w^\circ\|_{L_\infty(\PP_X)}
=
O_P(n^{-1/4})$ and 
$\|\hat\gamma-\gamma\|_{L_1(\PP_X)}
=
O_P(n^{-1/4})$. 
\end{assumption}

% Finally, we introduce the technical definition of local slope to formally state the rate double robustness results. 
% \begin{definition}[Local slope]
%     A function $H\colon \RR\to \RR$ has local slope $c\in \RR$ at point $t$ if there exists some constant $\delta>0$ such that $H(t-u)\leq H(t)-cu$ and $H(t+u)\geq H(t)+cu$ for any $u\in (0,\delta)$. 
% \end{definition}

\begin{theorem}[Rate double robustness]
\label{thm:dr_obs_rate_general}
Suppose Assumption~\ref{assump:obs_bal} holds with
$r_n=O(n^{-1/2})$ and
Assumption~\ref{assump:obs_slow} holds. Suppose there exist
deterministic constants $c,C,\eta>0$ such that, with probability
tending to one, the following regularity conditions hold:
\begin{enumerate}
\item[(i)]
$q(x)\geq c$ and $0\leq\hat\gamma(x)\leq1$ for every $x\in\cX$;

\item[(ii)]
the function 
$H_n(t):=\EE[
\hat\gamma(X)\ind\{\hat s(X)\leq t\}
\given\cT_n]$ 
satisfies 
$H_n(t_n^\circ-u)\leq\alpha-cu$ and 
$H_n(t_n^\circ+u)\geq\alpha+cu$ for all 
$0<u\leq\eta$, 
where
$t_n^\circ:=\sup\{t\in\RR:H_n(t)\leq\alpha\}$;

\item[(iii)]
for every $a<b$, 
$\PP\{a<\hat s(X)\leq b\mid\cT_n\}
\leq C(b-a)$. 
\end{enumerate}
Then
\[
\left[
\PP\left\{
Y_{n+1}\bigl(\hat\pi_{\obs}(X_{n+1})\bigr)
<
Y_{n+1}(0)
\,\middle|\,
\cA_n
\right\}
-\alpha
\right]_+
=
O_P(n^{-1/2}).
\]
\end{theorem}

\subsection{Technical conditions for the balancing algorithm}

In this section, we provide the technical conditions and proofs that our learn-then-balance weights satisfy the needed conditions in the preceding parts, which together lead to the doubly robust safety guarantees in the main text. 
We first define some preparatory notations. 
Let $\cT_n$ denote the $\sigma$-field generated by the training
process, so that $\hat e$, $\hat g$, $\hat\gamma$, $\hat s$, and
$\tilde w$ are fixed conditional on $\cT_n$. 
Let $G$ be a generic inclusion indicator satisfying
$G\mid X,T,\cT_n\sim\operatorname{Bernoulli}\{\hat g(X,T)\}$. Define
$q(x):=e(x)\hat g(x,1)+\{1-e(x)\}\hat g(x,0)$,
$p_G:=\EE[q(X)\mid\cT_n]$, and $m_n:=|\cI_{\calib}|$.
For any $\cT_n$-measurable functions
$r\colon\cX\to[0,1]$ and $u\colon\cX\to\RR^+$, define
\[
h_t^r(x):=r(x)\ind\{\hat s(x)\leq t\},\qquad
\psi_t^{r,u}(x):=(1,h_t^r(x),u(x))^\top,
\]
and
\[
t^\circ(r)
:=
\sup\left\{
t\in\RR:
\EE[h_t^r(X)\mid\cT_n]\leq\alpha
\right\}.
\]
Also, let
\[
M_{r,u}(t)
:=
\EE[\psi_t^{r,u}(X)\psi_t^{r,u}(X)^\top
\mid G=1,\cT_n],
\qquad
b_{r,u}(t)
:=
\EE[\psi_t^{r,u}(X)\mid\cT_n].
\]
Whenever $M_{r,u}(t)$ is invertible, define
\[
\omega_t^{r,u}(x)
:=
\psi_t^{r,u}(x)^\top
M_{r,u}(t)^{-1}b_{r,u}(t),
\qquad
\mathcal W_n(r,u)
:=
\omega_{t^\circ(r)}^{r,u}.
\]
For the fitted functions, abbreviate
$t^\circ:=t^\circ(\hat\gamma)$,
$M(t):=M_{\hat\gamma,\tilde w}(t)$,
$b(t):=b_{\hat\gamma,\tilde w}(t)$, and
$\omega_t:=\omega_t^{\hat\gamma,\tilde w}$.
Also, write $\psi_t:=\psi_t^{\hat\gamma,\tilde w}$ and define
\[
\hat t
:=
\sup\Big\{
t\in\RR:
\frac1n\sum_{i=1}^n
h_t^{\hat\gamma}(X_i)\leq\alpha
\Big\},
\qquad
\hat M_n(t)
:=
\frac1{m_n}\sum_{i\in\cI_{\calib}}
\psi_t(X_i)\psi_t(X_i)^\top,
\]
and $\bar\psi_n(t):=n^{-1}\sum_{i=1}^n\psi_t(X_i)$. Let 
$B_n(t)
:=
\left\{
b\in\RR^3:
b_1=1,\ 
|b_k-\bar\psi_{n,k}(t)|\leq\delta_n,\ k\in\{2,3\}
\right\}$. 
Whenever $\hat M_n(t)$ is invertible, define, for any $x\in \cX$ and $t\in \mathbb R$, 
\[
\hat b_n(t)
\in
\argmin_{b\in B_n(t)}
b^\top\hat M_n(t)^{-1}b,
\qquad
\hat a_n(t):=\hat M_n(t)^{-1}\hat b_n(t),
\qquad
\hat\omega_t(x):=\psi_t(x)^\top\hat a_n(t).
\]
\begin{assumption}[Regularity of the balancing program]
\label{assump:bal_reg} 
Recall that $t^\circ=t^\circ(\hat\fna)$,
$M(t)=M_{\hat\fna,\tilde w}(t)$, and
$\omega_t=\omega_t^{\hat\fna,\tilde w}$. There exist deterministic
constants $c_0,c_1,c_2,C,\eta>0$ such that, with probability tending
to one over the training process, the following conditions hold.

\begin{enumerate}
\item[(i)] For every $x\in\cX$,
$q(x)\geq c_0$, $0\leq\hat\fna(x)\leq1$, and
$0<\tilde w(x)\leq C$.

\item[(ii)] Let $\cN:=\{t\in\RR:|t-t^\circ|\leq\eta\}$. The population
balancing problem is uniformly nondegenerate and its solution is
uniformly interior on $\cN$:
$\inf_{t\in\cN}\lambda_{\min}\{M(t)\} \geq c_1$, and 
$\inf_{ t\in\cN, x\in\cX}\omega_t(x)\geq c_1$. 

\item[(iii)] The population cutoff is locally regular: for every
$0<u\leq\eta$, it holds that 
$\EE[h_{t^\circ-u}(X)\mid\cT_n]
\leq \alpha-c_2u$, and 
$\EE[h_{t^\circ+u}(X)\mid\cT_n]
\geq \alpha+c_2u$.

\item[(iv)] The conditional distribution of $\hat s(X)$ has a uniformly
bounded density in the sense that, for every $a<b$,
$\PP\{a<\hat s(X)\leq b\given\cT_n\}\leq C(b-a)$.
\end{enumerate}
\end{assumption}

Under the above regularity conditions, the following lemma shows that
the learned weights $\hat w_i$ are close to the population balancing
rule $\mathcal W_n(\hat\gamma,\tilde w)(X_i)$. Moreover, if the
preliminary weight function consistently estimates the oracle weight
$w$, then the population balancing rule consistently estimates the
normalized oracle weight $w^\circ$; under an $O_P(n^{-1/4})$ uniform
rate, it inherits the same rate. 
The proof is in Appendix~\ref{app:subsec_lem_bal_weight_approx}.

\begin{lemma}[Approximation properties of the balancing weights]
\label{lem:bal_weight_approx}
Suppose $\delta_n=O(n^{-1/2})$ and
Assumption~\ref{assump:bal_reg} holds. Then, with probability tending
to one, the balancing program has the unique solution
$\hat w_i=\hat\omega_{\hat t}(X_i)$ for
$i\in\cI_{\calib}$. Define
$\hat w_{n+1}:=\hat\omega_{\hat t}(X_{n+1})$. Then
\[
|\hat t-t^\circ|=O_P(n^{-1/2}),
\quad
\frac1{m_n}\sum_{i\in\cI_{\calib}}
\left\{
\hat w_i-\mathcal W_n(\hat\gamma,\tilde w)(X_i)
\right\}^2
=O_P(n^{-1/2}),
\quad 
\frac{\hat w_{n+1}\vee
\max_{i\in\cI_{\calib}}\hat w_i}{
\sum_{i\in\cI_{\calib}}\hat w_i
}=O_P(n^{-1}).
\]

Moreover, if $\|\tilde w-w\|_{L_2(\PP_X)}=o_P(1)$, then
\[
\left\|
\mathcal W_n(\hat\gamma,\tilde w)-w^\circ
\right\|_{L_2(\PP_X)}
=o_P(1),
\]
where
$w^\circ(x)=p_Gw(x)$.
Finally, if
$\|\tilde w-w\|_{L_\infty(\PP_X)}
=O_P(n^{-1/4})$, then
\[
\frac1{m_n}\sum_{i\in\cI_{\calib}}
\left\{
\hat w_i-\mathcal W_n(\hat\gamma,\tilde w)(X_i)
\right\}^2
=O_P(n^{-1}),
\qquad
\left\|
\mathcal W_n(\hat\gamma,\tilde w)-w^\circ
\right\|_{L_\infty(\PP_X)}
=O_P(n^{-1/4}).
\]
\end{lemma}

\section{Technical proofs}

\subsection{Proof of Theorem~\ref{thm:validity_rct}}
\label{app:subsec_validity_rct}

\begin{proof}[Proof of Theorem~\ref{thm:validity_rct}]
We aim to show that 
\@\label{eq:target_prf}
\PP(p_{n+1}\leq \alpha, Y_{n+1}^* =0 \given \{G_i\}_{i=1}^n ) \leq \alpha. 
\@
Note that by definition, $Y_i^\dagger = T_iY_i+(1-T_i)(1-Y_i) \leq \max\{Y_i(1),1-Y_i(0)\}=Y_i^*$. Thus, by the monotonicity of $V(x,y)$ in $y$, on the event $\{Y_{n+1}^*=0\}$, it holds deterministically that 
\$
p_{n+1} \geq p_{n+1}^* :=  \frac{1+ \sum_{i=1}^n G_i   \cdot \ind\{V(X_i,Y_i^*) \leq V(X_{n+1},Y_{n+1}^*)\} }{1+\sum_{i=1}^n G_i}.
\$
We thus have 
\@\label{eq:monotone_bd}
\PP(p_{n+1}\leq \alpha, Y_{n+1}^* =0 \given \{G_i\}_{i=1}^n )
\leq \PP(p_{n+1}^*\leq \alpha, Y_{n+1}^* =0 \given \{G_i\}_{i=1}^n ) 
\leq \PP(p_{n+1}^*\leq \alpha  \given \{G_i\}_{i=1}^n ).
\@
Now, conditional on $\{G_i\}_{i=1}^n$, we denote $\cI:=\{i\in [n]\colon G_i=1\}$ as the index set of the selected calibration data. 
Let $P^{\sup}_{X,Y^*}$ be the (unknown) joint distribution of $(X,Y^*)$ induced by the unknown super-population $P^{\sup}$ of $(X_i,Y_i(1),Y_i(0))$. We then have $(X_{n+1},Y_{n+1}^*)\sim P^{\sup}_{X,Y^*}$. 
For any measurable subset $A$ of $\cX\times \{0,1\}$, by the Bayes' rule, we have  
\$
\frac{\PP((X_i,Y_i^*)\in A \given G_i=1)}{\PP((X_{n+1},Y_{n+1}^*)\in A)}
= \frac{\PP(G_i=1 \given (X_i,Y_i^*)\in A) }{  \PP(G_i=1)}.
\$
Furthermore, we know 
\$
\PP(G_i=1 \given (X_i,Y_i^*)\in A )
&= \PP(G_i=1, T_i=1 \given (X_i,Y_i^*)\in A ) + \PP(G_i=1,T_i=0\given (X_i,Y_i^*)\in A)) \\ 
&= \PP(G_i=1 \given   T_i=1, (X_i,Y_i^*)\in A ) \cdot \PP(T_i = 1\given (X_i,Y_i^*)\in A ) \\ 
&\qquad + \PP(G_i=1 \given   T_i=0, (X_i,Y_i^*)\in A ) \cdot \PP(T_i = 0\given (X_i,Y_i^*)\in A ) \\ 
&= 1/2\cdot \EE[\hat{g}(X_i,1)\given (X_i,Y_i^*)\in A] + 1/2\cdot \EE[\hat{g}(X_i,0)\given (X_i,Y_i^*)\in A].
\$
Since $\hat{g}(x,1)+\hat{g}(x,0)\equiv a$ for a constant $a>0$, we have 
\$
\PP(G_i=1 \given (X_i,Y_i^*)\in A ) = a/2, 
\$
and by the tower property, 
\$
\PP(G_i=1) = \EE\big[\PP(G_i=1\given X_i)\big] = \EE\big[\PP(G_i=1,T_i=1\given X_i) +\PP(G_i=1,T_i=0\given X_i) \big] =a/2.
\$
Putting things together, we know that 
\$
(X_i,Y_i^*)\given G_i=1 ~~ \stackrel{d}{=} ~~(X_{n+1},Y_{n+1}^*).
\$
Thus, conditional on $\{G_i\}_{i=1}^n$, the samples $\{(X_i,Y_i)\}_{i\in \cI_0\cup\{n+1\}}$ are exchangeable. 
Noting that 
\$
p_{n+1}^* = \frac{1+\sum_{i\in \cI_0}\ind\{V(X_i,Y_i^*)\leq V(X_{n+1},Y_{n+1}^*)\}}{1+|\cI_0|},
\$
we know that $\PP(p_{n+1}^*\leq \alpha \given \{G_i\}_{i=1}^n )\leq \alpha$ in~\eqref{eq:monotone_bd}; see, e.g.,~\cite{jin2023selection} or~\cite{vovk2005algorithmic}. 
Marginalizing over $\{G_i\}_{i=1}^n$, we complete the proof of~\eqref{eq:target_prf}. Finally, taking $T_{n+1}=\ind\{p_{n+1}\leq \alpha\}$, we have 
\$
\PP(Y_{n+1}(T_{n+1})<Y_{n+1}(0)) = \PP(T_{n+1}=1,~Y_{n+1}^*=0) = \PP(p_{n+1}\leq \alpha, ~Y_{n+1}^*=0)\leq \alpha,
\$
thereby concluding the proof of Theorem~\ref{thm:validity_rct}.
\end{proof}

\subsection{Proof of Theorem~\ref{thm:asymp_opt_welfare}}
\label{app:subsec_opt_welfare_ours}

\begin{proof}[Proof of Theorem~\ref{thm:asymp_opt_welfare}]

Write  $\tau(x):=\mu_1(x)-\mu_0(x)$, 
$\widehat\tau(x):=\widehat\mu_1(x)-\widehat\mu_0(x)$, and recall that $\gamma(x)=\min\{1-\mu_1(x),\mu_0(x)\}$ and $\widehat\gamma(x) =\min\{1-\widehat\mu_1(x),\widehat\mu_0(x)\}$. We simplify the notation to
\[ 
s^*(x):=-\frac{\tau(x)}{\gamma(x)} = s_{\mathrm{welfare}}(x), \qquad \widehat s(x):= -\frac{\widehat\tau(x)}{\widehat\gamma(x)}. 
\] Let $\mathcal T_n$ denote the sigma-field generated by the independent training process for $\widehat\mu_0$ and $\widehat\mu_1$. Conditional on $\mathcal T_n$, the estimated functions $\widehat\mu_t$, $\widehat\gamma$, $\widehat s$, and the inclusion rule in (3.10) are fixed. We first show that selective calibration consistently recovers the sharp worst-case harm-rate function. Define \[ q_1(x):=1-\mu_1(x),\qquad q_0(x):=\mu_0(x), \] and analogously $\widehat q_1(x):=1-\widehat\mu_1(x)$ and $\widehat q_0(x):=\widehat\mu_0(x)$. Let $\widehat a(x)\in\{0,1\}$ denote the arm selected by (3.10), so that $\widehat a(x)=1$ when $\widehat q_1(x)\leq\widehat q_0(x)$ and $\widehat a(x)=0$ otherwise. 
Define the true proxy-label risk of the selected arm by 
\[ 
\overline\gamma_n(x) := q_1(x) \ind\{\widehat a(x)=1\} + q_0(x) \ind\{\widehat a(x)=0\}. 
\] 
Since $\gamma(x)=\min\{q_0(x),q_1(x)\}$, we have $\overline\gamma_n(x)\geq\gamma(x)$. Moreover, if $a^*(x)\in\arg\min_{a\in\{0,1\}}q_a(x)$, the optimality of $\widehat a(x)$ for the estimated risks gives 
\$ 
0 \leq \overline\gamma_n(x)-\gamma(x) &= q_{\widehat a(x)}(x)-q_{a^*(x)}(x)\\ &\leq \left| q_{\widehat a(x)}(x) -\widehat q_{\widehat a(x)}(x) \right| + \left| \widehat q_{a^*(x)}(x)-q_{a^*(x)}(x) \right|\\ &\leq 2\left\{ |\widehat\mu_1(x)-\mu_1(x)| + |\widehat\mu_0(x)-\mu_0(x)| \right\}. 
\$
The assumed $L_2(P_X)$ convergence and the Cauchy--Schwarz inequality therefore imply \[ \|\overline\gamma_n-\gamma\|_{L_1(P_X)} = o_P(1). \] 

We next establish convergence of the estimated ranking. The outcome-model consistency implies $\|\widehat\tau-\tau\|_{L_2(P_X)}=o_P(1)$ and $\|\widehat\gamma-\gamma\|_{L_2(P_X)} =o_P(1)$, where the latter follows from the Lipschitz property of the minimum function. Under the ratio conventions in~\eqref{eq:s_welfare_optimal}, if $\gamma(X)=0$, then $s^*(X)$ must take one of the three values $-\infty$, $0$, or $+\infty$. By the no-point-mass assumption for $s^*(X)$, we have $\mathbb{P}\{\gamma(X)=0\}=0$. 
The continuous mapping theorem then gives $\widehat s(X)-s^*(X) = o_P(1)$,  where $X\sim P_X$ is independent of the training process. 
The preceding convergence also holds uniformly for threshold indicators. Indeed, for every $\varepsilon>0$, 
\$
&\sup_{t\in\mathbb R} \mathbb{P}\!\left(  \ind\{\widehat s(X)\leq t\} \neq  \ind\{s^*(X)\leq t\} \,\middle|\,\mathcal T_n \right)\\ &\qquad\leq \mathbb{P}\!\left( |\widehat s(X)-s^*(X)|>\varepsilon \,\middle|\,\mathcal T_n \right) + \sup_{t\in\mathbb R} \mathbb{P}\{|s^*(X)-t|\leq\varepsilon\}. 
\$
The first term is $o_P(1)$. The second term converges to zero as $\varepsilon\downarrow0$ because $s^*(X)$ has no point mass. It follows that 
\[ 
\sup_{t\in\mathbb R} \mathbb{P}\!\left(  \ind\{\widehat s(X)\leq t\} \neq  \ind\{s^*(X)\leq t\} \,\middle|\,\mathcal T_n \right) = o_P(1).
\] 
Now define the empirical  curve 
\[ 
F_n(t) := \frac{ 1+\sum_{i=1}^n G_i \ind\{Y_i^\dagger=0\}  \ind\{\widehat s(X_i)\leq t\} }{ 1+\sum_{i=1}^nG_i }, \qquad t\in\mathbb R, 
\] 
so that $p^{\mathrm{welfare}}_{n+1} =F_n\{\widehat s(X_{n+1})\}$. 
Conditional on $\mathcal T_n$, the summands are i.i.d. and bounded, thus Lemma~\ref{lem:wcdf_conv} gives
\[ 
\sup_{t\in \RR}\bigg| F_n(t) - \frac{ \mathbb{E}\!\left[ G \ind\{Y^\dagger=0\}  \ind\{\widehat s(X)\leq t\} \,\middle|\,\mathcal T_n \right] }{ \mathbb{E}[G\mid\mathcal T_n] } \bigg| = o_P(1). 
\] 
Under balanced randomization and the given inclusion rule, we know $\mathbb{E}[G\mid X,\mathcal T_n]=1/2$. Furthermore, conditional on $X=x$, the probability that the selected proxy label is zero is $q_1(x)$ when $\widehat a(x)=1$ and $q_0(x)$ when $\widehat a(x)=0$. Consequently, 
\[ 
\mathbb{E}\!\left[ G \ind\{Y^\dagger=0\}  \ind\{\widehat s(X)\leq t\} \,\middle|\,\mathcal T_n \right] = \frac12 \mathbb{E}\!\left[ \overline\gamma_n(X)  \ind\{\widehat s(X)\leq t\} \,\middle|\,\mathcal T_n \right]. 
\] 
Therefore, 
\[ 
\sup_{t\in\mathbb R} \left| F_n(t)-\overline H_n(t) \right| = o_{\mathbb P}(1), \qquad \overline H_n(t) := \mathbb{E}\!\left[ \overline\gamma_n(X)  \ind\{\widehat s(X)\leq t\} \,\middle|\,\mathcal T_n \right]. 
\] 
Define the oracle harm-cost curve \[ H(t) := \mathbb{E}\!\left[ \gamma(X) \ind\{s^*(X)\leq t\} \right]. \] Using $0\leq\gamma\leq1$, we have
\$ 
\sup_{t\in\mathbb R} |\overline H_n(t)-H(t)| &\leq \|\overline\gamma_n-\gamma\|_{L_1(P_X)} + \sup_{t\in\mathbb R} \mathbb{E}\!\left[ \gamma(X) \left|  \ind\{\widehat s(X)\leq t\} -  \ind\{s^*(X)\leq t\} \right| \,\middle|\,\mathcal T_n \right] = o_P(1). 
\$
Combining the preceding displays gives  
$\sup_{t\in\mathbb R}|F_n(t)-H(t)| = o_P(1)$.  
The function $H$ is continuous because its jump at any $t\in\mathbb R$ is $\mathbb{E}[\gamma(X) \ind\{s^*(X)=t\}]=0$. Thus, for the independent test point, 
\$ 
\left| p^{\mathrm{welfare}}_{n+1} - H\{s^*(X_{n+1})\} \right| &\leq \sup_{t\in\mathbb R}|F_n(t)-H(t)| + \left| H\{\widehat s(X_{n+1})\} - H\{s^*(X_{n+1})\} \right| = 0_P(1). 
\$ 
We next translate this convergence into convergence of treatment decisions. 
Since $\tau(X)=0$ implies $s^*(X)=0$, the no-point-mass condition gives $\mathbb{P}\{\tau(X)=0\}=0$. Hence,  $ \ind\{\widehat\tau(X_{n+1})>0\} -  \ind\{\tau(X_{n+1})>0\}\to 0$ in probability. 
We also claim that \[ \mathbb{P}\!\left( H\{s^*(X)\}=\alpha,\, \tau(X)>0 \right)=0. \] 
To see this, let $J_\alpha := \{t<0:H(t)=\alpha\}$.  Since $H$ is nondecreasing, $J_\alpha$ is an interval. 
If $a<b$ lie in the interior of this interval, then $0 = H(b)-H(a) = \mathbb{E}\!\left[ \gamma(X)  \ind\{a<s^*(X)\leq b\} \right]$. 
Since $\gamma(X)>0$ almost surely, it follows that $\mathbb{P}\{a<s^*(X)\leq b\}=0$. 
The endpoints of $J_\alpha$ also have probability zero because $s^*(X)$ has no point mass. This proves the claim. 
Recall that  $\widehat\pi_{\mathrm{welfare}}(X_{n+1}) = \ind\{p^{\mathrm{welfare}}_{n+1}\leq\alpha\}  \ind\{\widehat\tau(X_{n+1})>0\}$. 
The convergence $p^{\mathrm{welfare}}_{n+1} - H\{s^*(X_{n+1})\} = o_{\mathbb P}(1)$, the convergence of the treatment-effect indicator, and the preceding zero-probability claim imply  $\mathbb{P}\!\left\{ \widehat\pi_{\mathrm{welfare}}(X_{n+1}) \neq \pi_0(X_{n+1}) \right\}\to 0$, where \[ \pi_0(x) := \ind\{H(s^*(x))\leq\alpha\} \ind\{\tau(x)>0\}. \] 

It remains to identify $\pi_0$ with the oracle policy in Theorem~\ref{thm:welfare_optimal}. Let $r^*\leq0$ be the cutoff in that theorem. Suppose first that the safety constraint is nonbinding. Then $r^*=0$ and, since $\mathbb{P}\{\tau(X)=0\}=0$, \[ H(0) = \mathbb{E}\!\left[ \gamma(X)  \ind\{\tau(X)>0\} \right] \leq\alpha. \] For every $x$ such that $\tau(x)>0$, we have $s^*(x)<0$, and therefore $ H\{s^*(x)\} \leq H(0) \leq\alpha$. 
It follows that \[ \pi_0(X) =  \ind\{\tau(X)>0\} =  \ind\{s^*(X)\leq0\} = \pi^*_{\mathrm{welfare}}(X) \quad\text{almost surely}. \] Suppose instead that the safety constraint is binding. Then $r^*<0$. By the continuity of $H$ and the definition of $r^*$, we have $H(r^*)=\alpha$.  By the monotonicity of $H$ and the preceding zero-probability result for the level set $\{t<0:H(t)=\alpha\}$, we know 
\[  
\ind\{H(s^*(X))\leq\alpha\} =  \ind\{s^*(X)\leq r^*\} \quad\text{almost surely}. 
\] 
Moreover, the indicator $ \ind\{\tau(X)>0\}$ is redundant on the event $\{s^*(X)\leq r^*\}$ because $r^*<0$. Consequently, \[ \pi_0(X) =  \ind\{s^*(X)\leq r^*\} =  \ind\{s_{\mathrm{welfare}}(X)\leq r^*\} = \pi^*_{\mathrm{welfare}}(X) \quad\text{almost surely}. \] 
We have therefore shown that 
$\mathbb{P}\!\left\{ \widehat\pi_{\mathrm{welfare}}(X_{n+1}) \neq \pi^*_{\mathrm{welfare}}(X_{n+1}) \right\}\to 0$, which implies  
\$
& \left| \mathbb{E}\!\left[ Y\{\widehat\pi_{\mathrm{welfare}}(X_{n+1})\} \right] - \operatorname{Welfare} (\pi^*_{\mathrm{welfare}};P) \right| \leq \mathbb{P}\!\left\{ \widehat\pi_{\mathrm{welfare}}(X_{n+1}) \neq \pi^*_{\mathrm{welfare}}(X_{n+1}) \right\} \to 0. 
\$
This completes the proof.
\end{proof}

\subsection{Proof of Theorem~\ref{thm:valid_weighted}}
\label{app:subsec_weighted_validity}

\begin{proof}[Proof of Theorem~\ref{thm:valid_weighted}]
% For any monotone score, it holds that $V(X_i,Y_i^\dagger)\leq V(X_i,Y_i^*)$ since $Y_i^\dagger\leq Y_i^*$ holds by definition. 
The p-value~\eqref{eq:def_pval_weight} can be  equivalently written as 
\@\label{eq:V_pval_weight}
p_{n+1} = \frac{w(X_{n+1}) +\sum_{i\in \cI_{\calib}} w(X_i)  \ind\{V(X_i,Y_i^\dagger) \leq V(X_{n+1},0)\}}{\sum_{i\in \cI_\calib} w(X_i) + w(X_{n+1})}.
\@
By definition, we have $Y_i^\dagger = T_iY_i+(1-T_i)(1-Y_i) \leq \max\{Y_i(1),1-Y_i(0)\}=Y_i^*$. Thus, by the monotonicity of $V(x,y)$ in $y$, on the event $\{Y_{n+1}^*=0\}$, it holds deterministically that 
\$
p_{n+1} \geq p_{n+1}^* :=  \frac{w(X_{n+1}) +\sum_{i\in \cI_{\calib}} w(X_i)  \ind\{V(X_i,Y_i^*) \leq V(X_{n+1},Y_{n+1}^*)\}}{\sum_{i\in \cI_\calib} w(X_i) + w(X_{n+1})}.
\$
We thus have 
\@\label{eq:monotone_bd_weight}
\PP(p_{n+1}\leq \alpha, Y_{n+1}^* =0 \given \{G_i\}_{i=1}^n )
\leq \PP(p_{n+1}^*\leq \alpha, Y_{n+1}^* =0 \given \{G_i\}_{i=1}^n ) \leq \PP(p_{n+1}^*\leq \alpha  \given \{G_i\}_{i=1}^n ).
\@
The goal is then to prove the RHS of~\eqref{eq:monotone_bd_weight} is upper bounded by $\alpha$.

Conditional on $\{G_i\}_{i=1}^n$, the   data $(X_i,Y_i^*)$ for  $i\in\cI_{\calib}$ are i.i.d.~and follow the distribution 
\$
P^{\calib} :\stackrel{d}{=} P_{(X_i,Y_i^*)\given G_i=1}, 
\$
whereas the test point $(X_i,Y_i^*)$ follows the distribution 
\$
P^{\test} := P_{X_i,Y_i^*},
\$
and we recall that $P$ denotes the true underlying super-population distribution. Let $p(x,y^*)$  be the density function of $P_{X,Y^*}$  with respect to some base measure. The density ratio between the two distributions is 
\$
\frac{\ud P^{\test} }{\ud P^{\calib}} (x,y) = \frac{p(x,y^*)}{p(x,y^*\given G=1)} = \frac{p(x) p(y^*\given x)}{p(x\given G=1)p(y^*\given x,G=1)},
\$
where we define the random variable $G=\hat{g}(X,T)$, and $p(y^*\given x,G=1)$ is the conditional density of $Y_i^*$ given $X_i=x$ and $G_i=1$, etc. 
By unconfoundedness, $Y^*$ is independent of $T$ conditional on $X$. This leads to 
\$
\frac{\ud P^{\test} }{\ud P^{\calib}} (x,y)   = \frac{p(x) p(y^*\given x)}{p(x\given G=1)p(y^*\given x)} = \frac{p(x)  }{p(x\given G=1) }  = \frac{\PP(G=1)}{\PP(G=1\given X=x)}.
\$
Now, by the definition of the $G_i$'s, we know 
\$
\PP(G=1\given X=x) % &= \PP(\hat{g}(x,T)=1\given X=x) \\ 
& = 
\PP(G=1,T=1\given X=x) + \PP(G=1,T=0\given X=x) \\ 
&= \PP(G=1\given T=1, X=x) \PP(T=1\given X=x) + \PP(G=1\given T=0,X=x)\PP(T=0\given X=x) \\  
&= e(x) \hat{g}(x,1) + (1-e(x))\hat{g}(x,0).
\$
% \nec{As for $\hat{g}(x,1)$ $\hat{g}(x,0)$, can we keep them as G? I was not sure what \hat{g}(x,1) = 1 meant in the third line, but it seems that the proof goes through with just using G?} \yingcomment{yes!}
This implies the density ratio between $P^\test$ and $P^\calib$ is given by 
\$
\frac{\ud P^{\test} }{\ud P^{\calib}} (x,y) = \PP(G=1)\cdot w(x).
\$
In other words, $\{(X_i,Y_i^*\}_{i\in \cI_{\calib}\cup\{n+1\}}$ are weighted exchangeable~\citep{TibshiraniBCR19}, and thus following~\cite{TibshiraniBCR19}, we know that 
\$
\PP\Bigg(\frac{w(X_{n+1}) +\sum_{i\in \cI_{\calib}} w(X_i)  \ind\{V(X_i,Y_i^*) \leq V(X_{n+1},Y_{n+1}^*)\}}{\sum_{i\in \cI_\calib} w(X_i) + w(X_{n+1})} \leq \alpha \Bigggiven \{G_i\}_{i=1}^n\Bigg) \leq \alpha.
\$
This proves the desired upper bound on the RHS of~\eqref{eq:monotone_bd_weight} hence Theorem~\ref{thm:valid_weighted}.
\end{proof}

\subsection{Proof of Theorem~\ref{thm:opt_stratified_experiments}}
\label{app:subsec_opt_stratified}

\begin{proof}[Proof of Theorem~\ref{thm:opt_stratified_experiments}]
    The optimality follows from the asymptotic convergence of the power/welfare due to the law of large numbers, as well as the optimal results in Theorems~\ref{thm:global_opt} and~\ref{thm:welfare_optimal}. 

    Write $\tau(x):=\mu_1(x)-\mu_0(x)$ and $\widehat\tau(x):=\widehat\mu_1(x)-\widehat\mu_0(x)$, and recall that $\gamma(x)=\min\{1-\mu_1(x),\mu_0(x)\}$ and $\widehat\gamma(x) =\min\{1-\widehat\mu_1(x),\widehat\mu_0(x)\}$. Define \[ s^*(x) := -\frac{\tau(x)}{\gamma(x)} = s_{\mathrm{welfare}}(x), \qquad \widehat s(x) := -\frac{\widehat\tau(x)}{\widehat\gamma(x)}. \]
    Let $\mathcal T_n$ denote the sigma-field generated by the independent training process. Conditional on $\mathcal T_n$, all estimated functions are fixed. As in the proof of Theorem~\ref{thm:asymp_opt_welfare},  let $q_1(x):=1-\mu_1(x)$ and $q_0(x):=\mu_0(x)$, and let $\widehat a(x)\in\{0,1\}$ denote the arm selected by~\eqref{eq:def_G_welfare_optimal}. Thus, $\widehat a(x)=1$ when $1-\widehat\mu_1(x)\leq\widehat\mu_0(x)$ and $\widehat a(x)=0$ otherwise. Define 
    \[ \gamma_n(x) := q_1(x) \ind\{\widehat a(x)=1\} + q_0(x) \ind\{\widehat a(x)=0\}. 
    \] 
    The optimality of $\widehat a(x)$ for the estimated risks gives \[ 0 \leq \gamma_n(x)-\gamma(x) \leq 2\left\{ |\widehat\mu_1(x)-\mu_1(x)| + |\widehat\mu_0(x)-\mu_0(x)| \right\}, \] and hence $\|\gamma_n-\gamma\|_{L_1(P_X)}=o_{\mathbb P}(1)$. 
    The same argument as in the proof of Theorem~\ref{thm:asymp_opt_welfare} also gives $\widehat s(X)-s^*(X)=o_{\mathbb P}(1)$ and \[ \sup_{t\in\mathbb R} \mathbb{P}\!\left(  \ind\{\widehat s(X)\leq t\} \neq  \ind\{s^*(X)\leq t\} \,\middle|\,\mathcal T_n \right) = o_{\mathbb P}(1). \] We now examine the effect of weighting. Define the conditional inclusion probability \[ \rho_n(x) := e(x) \ind\{\widehat a(x)=1\} + \{1-e(x)\} \ind\{\widehat a(x)=0\}, \] so that the weight in~\eqref{eq:def_weight} can be written as $w_n(x)=\rho_n(x)^{-1}$. 
    For covariates $X$, the inclusion indicator is $G= \ind\{T=\widehat a(X)\}$. 
    Conditional on $(X,\mathcal T_n)$, note the two identities \[ \mathbb{E}[G w_n(X)\mid X,\mathcal T_n]=1, \qquad \mathbb{E}\!\left[ G w_n(X) \ind\{Y^\dagger=0\} \mid X,\mathcal T_n \right] = \gamma_n(X). \] 
    Indeed, if $\widehat a(X)=1$, then $G=T$, $w_n(X)=1/e(X)$, and the second conditional expectation is $1-\mu_1(X)$. 
    If $\widehat a(X)=0$, then $G=1-T$, $w_n(X)=1/\{1-e(X)\}$, and the second conditional expectation is $\mu_0(X)$. Define 
    \[ 
    F_n(t) := \frac{ w_n(X_{n+1}) + \sum_{i=1}^n G_iw_n(X_i)  \ind\{Y_i^\dagger=0\}  \ind\{\widehat s(X_i)\leq t\} }{ w_n(X_{n+1}) + \sum_{i=1}^nG_iw_n(X_i) }, \qquad t\in\mathbb R, 
    \] 
    so that $p^{\mathrm{str-opt}}_{n+1} =F_n\{\widehat s(X_{n+1})\}$. 
    The preceding identities imply that the training-conditional population counterpart of $F_n$ is \[ \overline H_n(t) := \mathbb{E}\!\left[ \gamma_n(X)  \ind\{\widehat s(X)\leq t\} \,\middle|\,\mathcal T_n \right]. \] 
    We now briefly verify the uniform law of large numbers. Let $Z=Gw_n(X)$. 
    Conditional on $\mathcal T_n$, $\mathbb{E}[Z\mid\mathcal T_n]=1$. 
    Moreover, for every $M>0$, we know 
    \$
    \mathbb{E}[(Z-M)_+\mid\mathcal T_n] &= \mathbb{E}[(1-M\rho_n(X))_+\mid\mathcal T_n] \leq \mathbb{E}[(1-Me(X))_+] + \mathbb{E}[(1-M\{1-e(X)\})_+], 
    \$ 
    which converges to zero as $M\to\infty$ because $e(X)\in(0,1)$ almost surely. Applying Lemma~\ref{lem:wcdf_conv} to the bounded truncations and then letting $M\to\infty$ therefore gives a uniform law of large numbers for the numerator of $F_n$; the denominator follows by the same argument. In addition, $w_n(X_{n+1})/n=o_{\mathbb P}(1)$ because $w_n(X)\leq\max\{e(X)^{-1},\{1-e(X)\}^{-1}\}<\infty$ almost surely. Consequently, \[ \sup_{t\in\mathbb R} |F_n(t)-\overline H_n(t)| = o_{\mathbb P}(1). \] 
    
    Define the oracle harm-cost curve $H(t) := \mathbb{E}\!\left[ \gamma(X) \ind\{s^*(X)\leq t\} \right]$. 
    Using $0\leq\gamma\leq1$, we have 
    \$ 
    \sup_{t\in\mathbb R} |\overline H_n(t)-H(t)| &\leq \|\gamma_n-\gamma\|_{L_1(P_X)} + \sup_{t\in\mathbb R} \mathbb{E}\!\left[ \gamma(X) \left|  \ind\{\widehat s(X)\leq t\} -  \ind\{s^*(X)\leq t\} \right| \,\middle|\,\mathcal T_n \right] = o_{\mathbb P}(1).
    \$ 
    It follows that \[ \sup_{t\in\mathbb R}|F_n(t)-H(t)| = o_{\mathbb P}(1). \] The function $H$ is continuous because its jump at any $t\in\mathbb R$ is $\mathbb{E}[\gamma(X) \ind\{s^*(X)=t\}]=0$. 
    Hence, for the independent test point, $p^{\text{str-opt}}_{n+1} - H\{s^*(X_{n+1})\} = o_{\mathbb P}(1)$. 
    Finally, the outcome-model consistency and $\mathbb{P}\{\tau(X)=0\}=0$ imply  $\ind\{\widehat\tau(X_{n+1})>0\} -  \ind\{\tau(X_{n+1})>0\} = o_{\mathbb P}(1)$. As shown in the proof of Theorem~\ref{thm:asymp_opt_welfare}, the no-point-mass condition also implies  $\mathbb{P}\!\left( H\{s^*(X)\}=\alpha,\, \tau(X)>0 \right)=0$. 
    Therefore, we have  
    $\mathbb{P}\!\left\{ \widehat\pi_{\text{str-opt}}(X_{n+1}) \neq \pi_0(X_{n+1}) \right\} \to 0$,  
    where $\pi_0(x) :=  \ind\{H(s^*(x))\leq\alpha\}  \ind\{\tau(x)>0\}$. The binding and nonbinding arguments in the proof of Theorem~\ref{thm:asymp_opt_welfare} show that $\pi_0(X)=\pi^*_{\mathrm{welfare}}(X)$ almost surely. Thus, \[ \mathbb{P}\!\left\{ \widehat\pi_{\mathrm{str-opt}}(X_{n+1}) \neq \pi^*_{\mathrm{welfare}}(X_{n+1}) \right\} \to 0. \] 
    With similar arguments as the end of the proof of Theorem~\ref{thm:asymp_opt_welfare}, we complete the proof.
    \end{proof}

\subsection{Proof of Theorem~\ref{thm:dr_obs}}
\label{app:subsec_dr_obs}

\begin{proof}[Proof of Theorem~\ref{thm:dr_obs}]
By Lemma~\ref{lem:bal_weight_approx}, the learn-then-balance weights
satisfy Assumption~\ref{assump:obs_bal} with $r_n=n^{-1/4}$ and
reference function 
$\hat\omega(\cdot)
=
\mathcal W_n(\hat\gamma,\tilde w)(\cdot)$. 

Under condition~(i), Lemma~\ref{lem:bal_weight_approx}  gives
$\|\hat\omega-w^\circ\|_{L_2(\PP_X)}=o_P(1)$, 
where $w^\circ=p_Gw\propto w$. Hence condition~(i) of
Theorem~\ref{thm:dr_obs_general} holds.

Under condition~(ii),  
Assumption~\ref{assump:bal_reg} verifies the overlap, interiority,
local-slope, and bounded-density conditions in condition~(ii) of
Theorem~\ref{thm:dr_obs_general}. The conclusion therefore follows
from Theorem~\ref{thm:dr_obs_general} in either case.
\end{proof}

\subsection{Proof of Theorem~\ref{thm:dr_obs_rate}}
\label{app:subsec_dr_obs_rate}

\begin{proof}[Proof of Theorem~\ref{thm:dr_obs_rate}]
Set $\hat\omega:=\mathcal W_n(\hat\gamma,\tilde w)$.
By Lemma~\ref{lem:bal_weight_approx}, the learn-then-balance weights
satisfy Assumption~\ref{assump:obs_bal} with $r_n=n^{-1/2}$ and 
$\|\hat\omega-w^\circ\|_{L_\infty(\PP_X)}
=
O_P(n^{-1/4})$, where 
$w^\circ(x):=p_Gw(x)$. 
Together with the assumed rate for $\hat\gamma$, this verifies
Assumption~\ref{assump:obs_slow}. Assumption~\ref{assump:bal_reg}
verifies the remaining regularity conditions in
Theorem~\ref{thm:dr_obs_rate_general}, which gives the result.
\end{proof}

\subsection{Proof of Theorem~\ref{thm:obs_welfare_opt}}
\label{app:subsec_optimality_obs}

\begin{proof}[Proof of Theorem~\ref{thm:obs_welfare_opt}]

Let
\[
H_n(t)
:=
\EE\!\left[
\hat\gamma(X)\ind\{\hat s(X)\leq t\}
\,\middle|\,\cT_n
\right],
\qquad
t_n^\circ
:=
\sup\{t\in\RR:H_n(t)\leq\alpha\},
\]
and define
\[
H_0(t)
:=
\EE\!\left[
\gamma(X)\ind\{s_{\welfare}(X)\leq t\}
\right].
\]
By the assumed convergence of $\hat\gamma$ and $\hat s$, the
no-point-mass condition on $s_{\welfare}(X)$, and
Lemma~\ref{lem:wecdf_hat},  
\[
\sup_{t\in\RR}|H_n(t)-H_0(t)|=o_P(1).
\]
Together with the local-crossing condition in
Assumption~\ref{assump:bal_reg}(iii), this implies 
$t_n^\circ-t_{\text{raw}}^*=o_P(1)$, 
where
$t_{\text{raw}}^*=\sup\{t:H_0(t)\leq\alpha\}$.
Indeed, for every fixed $0<\varepsilon\leq\eta$, with probability
tending to one, it holds that  
$H_0(t_n^\circ-\varepsilon)<\alpha
<
H_0(t_n^\circ+\varepsilon)$, 
and hence
$t_n^\circ-\varepsilon\leq t_{\text{raw}}^*
\leq t_n^\circ+\varepsilon$.

Set
$\hat\omega:=\mathcal W_n(\hat\gamma,\tilde w)$ and, writing
$I_t(x):=\ind\{\hat s(x)\leq t\}$, define
\[
F_n^{\hat\gamma}(t)
:=
\frac{
\EE[
\hat\omega(X)\hat g(X,T)\hat\gamma(X)I_t(X)
\mid\cT_n]
}{
\EE[\hat\omega(X)\hat g(X,T)\mid\cT_n]
}.
\]
By the defining balance equations for $\hat\omega$, we know  
$F_n^{\hat\gamma}(t_n^\circ)
=
H_n(t_n^\circ)
=
\alpha$. 
Moreover, Assumption~\ref{assump:bal_reg}(i)--(iii) implies that,
for some deterministic $\kappa>0$, with probability tending to one, 
$F_n^{\hat\gamma}(t_n^\circ-u)
\leq\alpha-\kappa u$ and 
$F_n^{\hat\gamma}(t_n^\circ+u)
\geq\alpha+\kappa u$ hold 
for every $0<u\leq\eta$.
Define
\[
\hat F_n^{\obs}(t)
:=
\frac{
\hat w_{n+1}
+
\sum_{i=1}^n
G_i\hat w_i
\ind\{Y_i^\dagger=0\}I_t(X_i)
}{
\hat w_{n+1}
+
\sum_{i=1}^nG_i\hat w_i
},
\]
so that
$p_{n+1}^{\obs}
=\hat F_n^{\obs}\{\hat s(X_{n+1})\}$. 
Similar to the arguments in the proof of Theorem~\ref{thm:dr_obs_general}, 
Lemma~\ref{lem:bal_weight_approx},
Lemmas~\ref{lem:w_to_wx} and~\ref{lem:wecdf_hat}, and
$\hat w_{n+1}/\sum_{i\in\cI_{\calib}}\hat w_i=O_P(n^{-1})$
give
\$
\sup_{t\in \RR} \big|\hat{F}_n^{\obs}(t) - F_n^\dagger(t) \big| = o_P(1), \quad 
F_n^\dagger(t):= \frac{\EE[\hat\omega(X)\hat{g}(X,T)m(X,T)\ind\{\hat{s}(X)\leq t\}\given \cT_n]}{\EE[\hat\omega(X)\hat{g}(X,T)\given \cT_n]},
\$  
where we define 
$m(x,t) = \PP(Y^\dagger=0\given X=x,T=t)=t(1-\mu_1(x))+(1-t)\mu_0(x)$. 
Furthermore,
by the definition of $g^*$, we know 
$\EE[
g^*(X,T)\ind\{Y^\dagger=0\}
\given X]=\EE[
g^*(X,T)m(X,T)
\given X]
=
\EE[g^*(X,T)\gamma(X)\given X]$. 
Thus, 
\[
\sup_{t\in\RR}
|\hat F_n^{\obs}(t)-F_n^{\hat\gamma}(t)|
\leq
o_P(1)
+
C\|\hat g-g^*\|_{L_1(\PP_{X,T})}
+
C\|\hat\gamma-\gamma\|_{L_1(\PP_X)}
=
o_P(1).
\]

Let \[ \widehat\pi_{\mathrm{obs,raw}}(X_{n+1}) :=  \ind\{p^{\mathrm{obs}}_{n+1}\leq\alpha\}. \]
Fix $0<\varepsilon\leq\eta$. The preceding  results imply that, with probability tending to
one,
\$ 
\widehat s(X_{n+1}) \leq t_n^\circ-\varepsilon \quad\Rightarrow\quad p^{\mathrm{obs}}_{n+1}\leq\alpha, \\
\widehat s(X_{n+1}) \geq t_n^\circ+\varepsilon \quad\Rightarrow\quad p^{\mathrm{obs}}_{n+1}>\alpha.  
\$
Therefore, we have 
\[
\left| \widehat\pi_{\mathrm{obs,raw}}(X_{n+1}) -  \ind\{\widehat s(X_{n+1})\leq t_n^\circ\} \right| \leq  \ind\{ |\widehat s(X_{n+1})-t_n^\circ| \leq\varepsilon \} + o_{\mathbb P}(1).
\]
Assumption~\ref{assump:bal_reg}(iv) and the arbitrariness of $\varepsilon>0$ then imply 
$\mathbb{E}\!\left[ \left| \widehat\pi_{\mathrm{obs,raw}}(X_{n+1}) -  \ind\{\widehat s(X_{n+1})\leq t_n^\circ\} \right| \right] \to 0$. Since multiplication by an indicator cannot increase the absolute difference, the preceding display gives
\$
\mathbb{E}\!\left[ \left| \widehat\pi_{\mathrm{obs-welfare}}(X_{n+1}) -  \ind\{\widehat s(X_{n+1})\leq t_n^\circ\}  \ind\{\widehat s(X_{n+1})<0\} \right| \right] \to 0.
\$
We next pass to the population limit. The assumptions $\|\widehat s-s_{\mathrm{welfare}}\|_{L_2(P_X)} =o_{\mathbb P}(1)$ and $t_n^\circ-t_{\mathrm{raw}}^*=o_{\mathbb P}(1)$, together with the no-point-mass condition on $s_{\mathrm{welfare}}(X)$, imply 
\$ 
\mathbb{E}\Big[ \big| &  \ind\{\widehat s(X_{n+1})\leq t_n^\circ\}  \ind\{\widehat s(X_{n+1})<0\}  -  \ind\{s_{\mathrm{welfare}}(X_{n+1}) \leq t_{\mathrm{raw}}^*\}  \ind\{s_{\mathrm{welfare}}(X_{n+1})<0\} \big| \Big] \to 0.
\$ 
Indeed, for any $\varepsilon>0$, disagreement between the first threshold indicators is contained in the union of $\{|\widehat s-s_{\mathrm{welfare}}|>\varepsilon\}$, $\{|t_n^\circ-t_{\mathrm{raw}}^*|>\varepsilon\}$, and $\{|s_{\mathrm{welfare}}-t_{\mathrm{raw}}^*|\leq2\varepsilon\}$. The corresponding argument at the threshold zero establishes convergence of the second indicators. 

It remains to identify the limiting rule with the oracle policy in Theorem~\ref{thm:welfare_optimal}. Let $r^*\leq0$ denote the cutoff therein. We consider two cases:
\begin{itemize} 
\item If the safety constraint is binding, then $H_0(0)>\alpha$, so $t_{\mathrm{raw}}^*<0$ and $t_{\mathrm{raw}}^*=r^*$. In this case, $ \ind\{s_{\mathrm{welfare}}<0\}$ is redundant on $\{s_{\mathrm{welfare}}\leq t_{\mathrm{raw}}^*\}$, and hence $\ind\{s_{\mathrm{welfare}}(X) \leq t_{\mathrm{raw}}^*\}  \ind\{s_{\mathrm{welfare}}(X)<0\} =  \ind\{s_{\mathrm{welfare}}(X)\leq r^*\} = \pi^*_{\mathrm{welfare}}(X)$  almost surely. 
\item 
If the safety constraint is nonbinding, then $H_0(0)\leq\alpha$, so $r^*=0$ and $t_{\mathrm{raw}}^*\geq0$. Therefore,  $\ind\{s_{\mathrm{welfare}}(X) \leq t_{\mathrm{raw}}^*\}  \ind\{s_{\mathrm{welfare}}(X)<0\} =  \ind\{s_{\mathrm{welfare}}(X)<0\}$. Since $s_{\mathrm{welfare}}(X)$ has no point mass at zero, this equals $ \ind\{s_{\mathrm{welfare}}(X)\leq0\} = \pi^*_{\mathrm{welfare}}(X)$ almost surely. 
\end{itemize}
Combining the preceding results gives  $\mathbb{E} [\left| \widehat\pi_{\mathrm{obs-welfare}}(X_{n+1}) - \pi^*_{\mathrm{welfare}}(X_{n+1}) \right|] \to 0 $. Finally, since the outcomes are binary, we know 
\$ 
\left| \mathbb{E}\!\left[ Y\{\widehat\pi_{\mathrm{obs-welfare}}(X_{n+1})\} \right] - \operatorname{Welfare}(\pi^*_{\mathrm{welfare}};P) \right| &\leq \mathbb{E}\!\left[ \left| \widehat\pi_{\mathrm{obs-welfare}}(X_{n+1}) - \pi^*_{\mathrm{welfare}}(X_{n+1}) \right| \right] \to 0,
\$
which completes the proof.
\end{proof}

\section{Proof of general balancing theory}
\subsection{Proof of Theorem~\ref{thm:dr_obs_general}}
\label{app:subsec_dr_obs_general}

\begin{proof}[Proof of Theorem~\ref{thm:dr_obs_general}] 

    Throughout this proof, we condition on the training process of the functions, so $\hat{s}$ and $\hat{g}$ are viewed as fixed, as well as the function $\hat{w}(\cdot)$ in the first condition of Assumption~\ref{assump:obs_bal}. 
    Once we prove $\PP(Y_{n+1}(\pi_{\textnormal{obs}}(X_{n+1})<Y_{n+1}(0)\given \hat{s},\hat{g},\hat{w})\leq \alpha + o_P(1)$ conditional on these randomness, since $R$ is uniformly bounded we also have the marginal harm rate $\PP(Y_{n+1}(\pi_{\textnormal{obs}}(X_{n+1})<Y_{n+1}(0)) \leq  \alpha + o_P(1)$. For notational simplicity, we shall omit the conditioning in the probability/expectation throughout the proof.     
    Define 
    \$
    \hat\tau:= \sup\{t\in \RR\colon \hat{F}_n(t)\leq \alpha\},\quad \hat{F}_n(t):=  \frac{ \hat{w}_{n+1} +\sum_{i\in \cI_{\calib}} \hat{w}_i \ind\{Y_i^\dagger = 0\}\ind\{\hat{s}(X_i)\leq t\}}{\sum_{i\in \cI_\calib} \hat{w}_i + \hat{w}_{n+1}},
    \$
    so that $p_{n+1} = \hat{F}_n(\hat{s}(X_{n+1}))$. 
    Thus, $p_{n+1}\leq \alpha$ is equivalent to $\hat{s}(X_{n+1})\leq \hat\tau$. 
    
    We now construct a random variable $Y_i^{**}$ whose conditional expectation is the worst-case harm rate function $\fna(X_i)$ and upper bounds $Y_i^\dagger$. 
    For $t\in\{0,1\}$, define
\[
D_i(t):=tY_i(1)+(1-t)\{1-Y_i(0)\},\qquad
q_t(x):=\PP(D_i(t)=0\mid X_i=x),
\]
so that $D_i(T_i)=Y_i^\dagger$, $q_1(x)=1-\mu_1(x)$, $q_0(x)=\mu_0(x)$, and our previously defined oracle label satisfies 
\[
Y_i^*=\max\{D_i(0),D_i(1)\}.
\]
Let $\nu(x):=\PP(Y_i^*=0\mid X_i=x)$. Since
$\{Y_i^*=0\}\subseteq\{D_i(t)=0\}$ for each $t\in\{0,1\}$, we have 
\[
\nu(x)\leq \gamma(x)=\min\{q_0(x),q_1(x)\}\leq q_t(x).
\]
Define
\[
\rho_t(x):=
\begin{cases}
\dfrac{\gamma(x)-\nu(x)}{q_t(x)-\nu(x)}, & q_t(x)>\nu(x),\\[6pt]
0, & q_t(x)=\nu(x).
\end{cases}
\]
Then $\rho_t(x)\in[0,1]$; moreover, if $q_t(x)=\nu(x)$, then necessarily
$\gamma(x)=\nu(x)$. On an enlarged probability space, let
$U_i\stackrel{\mathrm{i.i.d.}}{\sim}\operatorname{Unif}(0,1)$ be independent
of all existing random variables, and define
\[
Y_i^{**}
:=
Y_i^*
-
\ind\left\{
Y_i^*=1,\,
D_i(T_i)=0,\,
U_i\leq \rho_{T_i}(X_i)
\right\}.
\]
By construction, $Y_i^{**}\leq Y_i^*$. Furthermore, if $Y_i^\dagger=1$,
then $D_i(T_i)=1$, so the indicator in the preceding display vanishes and
$Y_i^{**}=Y_i^*=1$. Hence
\[
Y_i^\dagger\leq Y_i^{**}\leq Y_i^*
\qquad\text{almost surely}.
\]
By unconfoundedness, for each $t\in\{0,1\}$,
\[
\begin{aligned}
\PP(Y_i^{**}=0\mid X_i=x,T_i=t)
&=
\nu(x)
+
\rho_t(x)
\PP\bigl(Y_i^*=1,D_i(t)=0\mid X_i=x\bigr)\\
&=
\nu(x)+\rho_t(x)\{q_t(x)-\nu(x)\}
=
\gamma(x).
\end{aligned}
\]
Thus, in particular, $\PP(Y_i^{**}=0\mid X_i)=\gamma(X_i)$.

% Accordingly, we define the (unobservable) oracle counterpart 
%     \$
%     \hat\tau^* = \sup\{t\in \RR\colon \hat{F}_n^*(t)\leq \alpha\},\quad \hat{F}^*_n(t):=  \frac{ \hat{w}_{n+1} +\sum_{i\in \cI_{\calib}} \hat{w}_i \ind\{Y_i^{**} = 0\}\ind\{\hat{s}(X_i)\leq t\}}{\sum_{i\in \cI_\calib} \hat{w}_i + \hat{w}_{n+1}}.
%     \$
%     Since $Y_i^{**}\geq Y_i^\dagger$ almost surely, we have $\hat{F}_n^*(t) \leq \hat{F}_n(t)$ for every $t\in \RR$ almost surely, hence 
%     \$
%     \hat\tau \leq \hat\tau^*\quad \text{almost surely}.
%     \$
%     The worst-case bound on the harm rate (e.g.,~\cite{kallus2022s}) then implies 
%     \@\label{eq:risk_eq1}
%     \hspace{-2em}
%     \PP(Y_{n+1}(T_{n+1})<Y_{n+1}(0)) &= \PP( Y_{n+1}(1)=0, Y_{n+1}(0)=1, \hat{s}(X_{n+1})\leq \hat\tau) \notag \\
%     & 
%     \leq \EE[\fna(X_{n+1})\ind\{\hat{s}(X_{n+1})\leq \hat\tau\}]
%     \leq \EE[\fna(X_{n+1})\ind\{\hat{s}(X_{n+1})\leq \hat\tau^*\}],
%     \@
%     where the first inequality invokes the tower property of expectations, and the last inequality uses $\hat\tau\leq \hat\tau^*$. 
%     Above, $\hat\tau$ and $\hat\tau^*$ are random cutoffs whose randomness depends on the labeled data and the test point. 

    To avoid introducing a cutoff that depends on the test point, define
$ 
S_n:=\sum_{i=1}^n G_i\hat w_i,
$ 
and, on the event $\{S_n>0\}$, define the calibration-only functions
\@\label{eq:def_hatF}
\hat F_{n,0}(t):=
\frac{\sum_{i=1}^n G_i\hat w_i\ind\{Y_i^\dagger=0\}
\ind\{\hat s(X_i)\leq t\}}{S_n},\quad
\hat F_{n,0}^*(t):=
\frac{\sum_{i=1}^n G_i\hat w_i\ind\{Y_i^{**}=0\}
\ind\{\hat s(X_i)\leq t\}}{S_n}.
\@
Under the assumptions of the theorem, $S_n>0$ with probability tending to one. Moreover, assuming that the estimated weights are nonnegative, both $\hat F_{n,0}$ and $\hat F_{n,0}^*$ are nondecreasing functions taking values in $[0,1]$.

The observed p-value thus obeys 
\[
p_{n+1}^{\mathrm{obs}}
=
\frac{\hat w_{n+1}
+S_n\hat F_{n,0}\bigl(\hat s(X_{n+1})\bigr)}
{\hat w_{n+1}+S_n} \geq \hat{F}_{n,0}(\hat{s}(X_{n+1}))
\] 
since $\hat{F}_{n,0}(t)\in [0,1]$ for any $t\in \RR$. 
Furthermore, since $Y_i^\dagger\leq Y_i^{**}$ almost surely, we have 
\[
\hat F_{n,0}(t)\geq\hat F_{n,0}^*(t)
\qquad\text{for every }t\in\RR.
\]
It follows that
\[
\{p_{n+1}^{\mathrm{obs}}\leq\alpha\}
\subseteq
\big\{
\hat F_{n,0} (\hat s(X_{n+1}) )\leq\alpha
\big\}
\subseteq
\big\{
\hat F_{n,0}^* (\hat s(X_{n+1}) )\leq\alpha
\big\}.
\]

For later use, define the calibration-only cutoffs
\[
\hat\tau_0:=
\sup\{t\in\RR:\hat F_{n,0}(t)\leq\alpha\},\qquad
\hat\tau_0^*:=
\sup\{t\in\RR:\hat F_{n,0}^*(t)\leq\alpha\}.
\]
The preceding pointwise ordering implies $\hat\tau_0\leq\hat\tau_0^*$. Moreover,
\[
\{p_{n+1}^{\mathrm{obs}}\leq\alpha\}
\subseteq
\{\hat s(X_{n+1})\leq\hat\tau_0\}
\subseteq
\{\hat s(X_{n+1})\leq\hat\tau_0^*\}.
\]
Importantly, both $\hat\tau_0$ and $\hat\tau_0^*$ depend only on the calibration data, together with the auxiliary randomness used to construct $Y_i^{**}$, and do not depend on the test point. Let $\mathcal A_n$ denote the $\sigma$-field generated by these quantities. Conditional on $\mathcal A_n$, and writing
\[
R_n:=
\PP\left(
Y_{n+1}(1)=0,\,
Y_{n+1}(0)=1,\,
p_{n+1}^{\mathrm{obs}}\leq\alpha
\,\middle|\,\mathcal A_n
\right),
\]
the sharp conditional upper bound
$\PP(Y(1)=0,Y(0)=1\mid X)\leq\gamma(X)$ gives
\[
\begin{aligned}
R_n
&\leq
\EE\left[
\gamma(X_{n+1})
\ind\big\{
\hat F_{n,0}^* (\hat s(X_{n+1}) )\leq\alpha
\big\}
\,\middle|\,\mathcal A_n
\right] \leq
\EE\left[
\gamma(X_{n+1})
\ind\{\hat s(X_{n+1})\leq\hat\tau_0^*\}
\,\middle|\,\mathcal A_n
\right].
\end{aligned}
\]

    Since $\hat{F}^*_{n,0}(t)$ is a right-continuous and non-decreasing step function, letting $\hat{s}(X_{[1]})\leq \hat{s}(X_{[2]})\leq \cdots \leq \hat{s}(X_{[n]})$ be the order statistics and $[1],[2],\dots,[n]$ is a permutation of $(1,\dots,n)$, we consider two cases:
    \begin{itemize}
        \item When there exists some $k\in [n]$ such that $\hat{F}_{n,0}^*(\hat{s}(X_{[k]})) = \alpha$, we know $\hat\tau_0^* = \hat{s}(X_{[k]})$, thus $\hat{F}_{n,0}^*(\hat\tau_0^*) = \alpha$. 
        \item When there exists some $k\in [n]$ such that $\hat{F}^*_{n,0}(\hat{s}(X_{[k-1]}))<\alpha< \hat{F}_{n,0}^*(\hat{s}(X_{[k]}))$, we know $\hat\tau_0^* = \hat{s}(X_{[k]})$, and $\alpha < \hat{F}_{n,0}^*(\hat\tau^*) \leq \alpha + \hat{F}_{n,0}^*(X_{[k]}) - \hat{F}_{n,0}^*(X_{[k-1]}) 
        = \alpha + \frac{\hat{w}_{[k-1]}}{\sum_{i\in \cI_{\calib}}\hat{w}_i+\hat{w}_{n+1}} = \alpha + O_P(1/n)$ by Lemma~\ref{lem:bal_weight_approx}. 
    \end{itemize}
    Combining the two cases yields $\alpha \leq \hat{F}^*_{n,0}(\hat\tau_0^*)\leq \alpha+O_P(1/n)$.

    The above arguments yield
    \@
    R_n &\leq \EE[\fna(X_{n+1})\ind\{\hat{s}(X_{n+1})\leq \hat\tau_0^*\}\given \cA_n] + \alpha - \hat{F}_{n,0}^*(\hat\tau_0^*) + O_P(1/n)  \label{eq:risk_bd20} \\ 
    % & \leq \alpha + \EE[\fna(X_{n+1})\ind\{\hat{s}(X_{n+1})\leq \hat\tau_0^*\}\given \cA_n] - \frac{  \sum_{i=1}^n G_i \hat{w}_i \ind\{Y_i^{**} = 0\}\ind\{\hat{s}(X_i)\leq \hat\tau_0^*\}}{\sum_{i=1}^n G_i \hat{w}_i }\notag  \\ 
    & \leq \alpha + \EE[\fna(X_{n+1})\ind\{\hat{s}(X_{n+1})\leq \hat\tau_0^*\}\given \cA_n] - \frac{  \sum_{i=1}^n G_i \hat{w}_i \ind\{Y_i^{**} = 0\}\ind\{\hat{s}(X_i)\leq \hat\tau_0^*\}}{\sum_{i=1}^n G_i \hat{w}_i  } + O_P(1/n) , \quad \label{eq:risk_bd2}
    \@ 
    We shall repeatedly invoke the following lemma, whose proof is in Appendix~\ref{subsec:lemma_w_to_wx}. 
    
    \begin{lemma}\label{lem:w_to_wx}
        Under Assumption~\ref{assump:obs_bal}, for any random variable $\{Z_i\}$ such that $(X_i,Y_i,T_i,Z_i)$ are i.i.d.~across $i\in [n]$ and $\EE[Z^2]<\infty$, and any (possibly random) function $f\colon \cX\to \RR$, it holds that 
        \$
    \sup_{t\in \RR}\bigg|\frac{1}{n}\sum_{i=1}^n \hat{w}_iZ_i\ind\{f(X_i)\leq t\} - \frac{1}{n}\sum_{i=1}^n \hat{\omega}(X_i) Z_i\ind\{f(X_i)\leq t\} \bigg| = O_P(r_n).
        \$
    \end{lemma}

    Lemma~\ref{lem:w_to_wx} with $r_n = o(1)$ implies $\frac{1}{n}\sum_{i=1}^n G_i\hat{w}_i = \frac{1}{n}\sum_{i=1}^n G_i \hat{\omega}(X_i) + O_P(r_n)$, and together with $Z_i = G_i\ind\{Y_i^{**}=0\}$  and $f=\hat{s}$ it implies 
    \@\label{eq:bd_w_wx_Y**}
    \sup_{t\in \RR}\bigg| \frac{  \sum_{i=1}^n G_i \hat{w}_i \ind\{Y_i^{**} = 0\}\ind\{\hat{s}(X_i)\leq t\}}{\sum_{i=1}^n G_i \hat{w}_i } - \frac{  \sum_{i=1}^n G_i \hat{\omega}(X_i) \ind\{Y_i^{**} = 0\}\ind\{\hat{s}(X_i)\leq t\}}{\sum_{i=1}^n G_i \hat{\omega}(X_i) } \bigg| = O_P(r_n).
    \@
    Taking $t=\hat\tau_0^*$ in the above display, and following~\eqref{eq:risk_bd2}, we have 
    \@\label{eq:risk_bd22}
    R_n\leq  \alpha + \EE[\fna(X_{n+1})\ind\{\hat{s}(X_{n+1})\leq \hat\tau_0^*\}\given \cA_n] - \frac{  \sum_{i=1}^n G_i \hat{\omega}(X_i) \ind\{Y_i^{**} = 0,\hat{s}(X_i)\leq \hat\tau_0^*\}}{\sum_{i=1}^n G_i \hat{\omega}(X_i)  } +O_P(r_n+1/n).\quad~ ~
    \@ 

    Consider
\( 
Z_i
:=
G_i\hat \omega(X_i)
\{\ind(Y_i^{**}=0)-\gamma(X_i)\}.
\)
Since $G_i$ is generated independently of $(Y_i(0),Y_i(1),U_i)$
conditional on $(X_i,T_i)$, and
$\PP(Y_i^{**}=0\mid X_i,T_i)=\gamma(X_i)$, we have
\[
\begin{aligned}
\EE[Z_i\mid X_i]
&=
\hat \omega(X_i)
\EE\left[
\EE\left[
G_i\{\ind(Y_i^{**}=0)-\gamma(X_i)\}
\mid X_i,T_i
\right]
\middle|X_i
\right]\\
&=
\hat \omega(X_i)
\EE\left[
\hat g(X_i,T_i)
\{\PP(Y_i^{**}=0\mid X_i,T_i)-\gamma(X_i)\}
\middle|X_i
\right]
=0.
\end{aligned}
\]  
    % We now define  
    % \$
    % \hat{q}(x):= &\EE[ \ind\{Y_i^\dagger=0\}\given X_i=x,G_i=1]
    % = \frac{\PP(Y_i^\dagger =0, G_i=1\given X_i=x)}{\PP(G_i=1\given X=x)} \\ 
    % &= \frac{\PP(Y_i^\dagger =0, G_i=1, T_i=1\given X_i=x)+\PP(Y_i^\dagger =0, G_i=1, T_i=0\given X_i=x)}{\PP(T_i=1, G_i=1\given X=x) + \PP(T_i=0, G_i=1\given X=x)} \\ 
    % &= (1-\mu_1(x)) \cdot \frac{ e(x) \hat{g}(x,1)}{e(x)\hat{g}(x,1)+(1-e(x))\hat{g}(x,0)} +  \mu_0(x) \cdot \frac{ (1-e(x)) \hat{g}(x,0)}{e(x)\hat{g}(x,1)+(1-e(x))\hat{g}(x,0)},
    % \$
    % We now write $G_i=\hat{g}(X_i,T_i)$ 
    % Recall that $\hat{g}(x,t)$ is the sampling probability of $G_i=1$ given $(X_i,T_i)=(x,t)$. 
    % This function is independent of the $n+1$ data points (only depending on the training process of $\hat{g}$).  
    % Consider $Z_i=G_i\hat{w}(X_i)[\ind\{Y_i^{**}=0\}-\fna(X_i)]$ and $s(\cdot)=\hat{s}(\cdot)$ in Lemma~\ref{lem:wcdf_conv}. By definition, because of the unconfoundedness, we know  $\EE[Z_i\given X_i] = \EE[\hat{g}(X_i,T_i)\given X_i] \hat{w}(X_i) \cdot \{\PP(Y_i^{**}=0\given X_i) - \fna(X_i)\} =0$ almost surely, and thus $\EE[Z\ind\{s(X)\leq t\}]=0$. 
    Thus, invoking Lemma~\ref{lem:wcdf_conv} with this $Z_i$ and $s(\cdot)=\hat{s}(\cdot)$  yields
    \@\label{eq:bd_Y_to_q}
    \sup_{t\in \RR}\bigg| \frac{  \sum_{i=1}^n G_i \hat{\omega}(X_i)  [\ind\{Y_i^{**}=0\} - \fna(X_i)] \ind\{\hat{s}(X_i)\leq t\}}{\sum_{i=1}^n G_i \hat{\omega}(X_i)  }  \bigg| = O_P(1/\sqrt{n}).
    \@
    % Since $\fna(x) = \min \{ 1-\mu_1(x), \mu_0(x)\}$, by the definition of $\hat{q}(x)$ we know 
    % \$
    % \hat{q}(x) \geq \fna(x),\quad \text{for all }x\in \cX,
    % \$
    % no matter the choice of $e(x)$ and $\hat{g}(x,t)$. 
    Continuing with~\eqref{eq:risk_bd22}, this implies 
    \@\label{eq:risk_bd23}
    R_n\leq  \alpha + \EE[\fna(X_{n+1})\ind\{\hat{s}(X_{n+1})\leq \hat\tau_0^*\}\given\cA_n] - \frac{  \sum_{i=1}^n G_i \hat{\omega}(X_i)\fna(X_i)\ind\{\hat{s}(X_i)\leq \hat\tau_0^*\}}{\sum_{i=1}^n G_i \hat{\omega}(X_i)  } +O_P(1/\sqrt{n}+r_n).\qquad
    \@ 

    \paragraph{Case 1: $\|\hat\omega-w^\circ\|_{L_2(\PP_X)}=o_P(1)$.} 
    Let $I_t(x):=\ind\{\hat s(x)\leq t\}$. Conditional on $\cT_n$, uniform empirical convergence gives
\[
\sup_{t\in\RR}
\left|
\frac{
\sum_{i=1}^nG_i\hat\omega(X_i)\gamma(X_i)I_t(X_i)
}{
\sum_{i=1}^nG_i\hat\omega(X_i)
}
-
\frac{
\EE[q(X)\hat\omega(X)\gamma(X)I_t(X)\mid\cT_n]
}{
\EE[q(X)\hat\omega(X)\mid\cT_n]
}
\right|
=o_P(1).
\]
Since $0\leq q,\gamma,I_t\leq1$, we have
\[
\sup_{t\in\RR}
\left|
\EE[q(X)\{\hat\omega(X)-w^\circ(X)\}
\gamma(X)I_t(X)\mid\cT_n]
\right|
=o_P(1),
\]
and
$\EE[q(X)\{\hat\omega(X)-w^\circ(X)\}\mid\cT_n]=o_P(1)$.
Moreover, $q(x)w^\circ(x)=p_G$, and therefore
\[
\frac{
\EE[q(X)w^\circ(X)\gamma(X)I_t(X)\mid\cT_n]
}{
\EE[q(X)w^\circ(X)\mid\cT_n]
}
=
\EE[\gamma(X)I_t(X)\mid\cT_n].
\]
Since $p_G$ is bounded away from zero, it follows that
\[
\sup_{t\in\RR}
\left|
\frac{
\sum_{i=1}^nG_i\hat\omega(X_i)\gamma(X_i)I_t(X_i)
}{
\sum_{i=1}^nG_i\hat\omega(X_i)
}
-
\EE[\gamma(X)I_t(X)\mid\cT_n]
\right|
=o_P(1).
\]
Evaluating this display at $t=\hat\tau_0^*$ in~\eqref{eq:risk_bd23}
gives $R_n\leq\alpha+o_P(1)$.

    % Since our p-value is invariant to constant scaling of the weights, without loss of generality we assume $\bar{w} (\cdot)=w(\cdot)$. 
    % Recall that $G_i\sim \text{Bern}(\hat{g}(X_i,T_i))$ conditional on the data for the pre-trained function $\hat{g}$ which is viewed as fixed throughout the proof. Invoking Lemma~\ref{lem:wecdf_hat} with $Z_i=T_i$, $\hat{f}_n(X_i,Z_i) = G_i\hat{w}(X_i) \fna(X_i)$ and $\hat{s}_n(X_i)=\hat{s}(X_i)$, we have 
    % \$
    % \sup_{t\in \RR} \bigg| \frac{  \sum_{i=1}^n G_i \hat{w}(X_i)\fna(X_i)\ind\{\hat{s}(X_i)\leq t\}}{\sum_{i=1}^n G_i \hat{w}(X_i)  } - \frac{  \EE[G_i  \bar{w}(X_i)\fna(X_i)\ind\{\hat{s}(X_i)\leq t\}]}{\EE[G_i \bar{w}(X_i)]  } \bigg| = o_P(1),
    % \$
    % and since $\bar{w}=w$ is the correct density ratio, we know 
    % \$
    % &\EE[G_i\bar{w}(X_i) ] = \EE[\hat{g}(X,T)\bar{w}(X)]=1,  \\
    % &\EE[G_i\bar{w}(X_i)\fna(X_i)\ind\{\hat{s}(X_i)\leq t\}] = \EE[\gamma(X)\ind\{\hat{s}(X)\leq t\}],\quad \forall ~ t\in \RR. 
    % \$
    % Taking $t=\hat\tau_0^*$ yields the desired bound $R\leq \alpha + o_P(1)$ with $r_n=o(1)$.  

    \paragraph{Case 2: $\|\hat\gamma-\gamma\|_{L_2(\PP_X)}=o_P(1)$.} 
    For $\xi\in\{\hat\gamma,\gamma\}$, define the training-conditional
population curves
\[
H_n^\xi(t)
:=
\EE[\xi(X)\ind\{\hat s(X)\leq t\}\mid\cT_n],
\qquad
F_n^\xi(t)
:=
\frac{
\EE[\hat\omega(X)\hat g(X,T)\xi(X)
\ind\{\hat s(X)\leq t\}\mid\cT_n]
}{
\EE[\hat\omega(X)\hat g(X,T)\mid\cT_n]
}.
\]
Let $t_n^\circ:=\sup\{t\in\RR:H_n^{\hat\gamma}(t)\leq\alpha\}$ and
recall
\[
\hat Q_n(t)
:=
\frac{
\sum_{i=1}^nG_i\hat w_i\hat\gamma(X_i)
\ind\{\hat s(X_i)\leq t\}
}{
\sum_{i=1}^nG_i\hat w_i
}.
\]
Lemmas~\ref{lem:w_to_wx} and~\ref{lem:wecdf_hat}, applied to the
numerators and denominators, give
\@
\sup_{t\in\RR}|\hat Q_n(t)-F_n^{\hat\gamma}(t)|
&=o_P(1), \quad 
\sup_{t\in\RR}|\hat F_{n,0}^*(t)-F_n^\gamma(t)|
 =o_P(1).
\@
Here the second display also uses
$\PP(Y_i^{**}=0\mid X_i,T_i)=\gamma(X_i)$. Moreover, since the
denominator of $F_n^\xi$ is bounded away from zero,
\[
\sup_{t\in\RR}|F_n^\gamma(t)-F_n^{\hat\gamma}(t)|
+
\sup_{t\in\RR}|H_n^\gamma(t)-H_n^{\hat\gamma}(t)|
\leq
C\|\hat\gamma-\gamma\|_{L_1(\PP_X)}
=o_P(1).
\]

By Lemma~\ref{lem:wecdf_hat} and the assumed local slope,
$\hat t-t_n^\circ=O_P(n^{-1/2})$. The balancing condition in
Assumption~\ref{assump:obs_bal}, together with the definition of
$\hat t$ and the fact that the maximal jump of its defining empirical
curve is at most $n^{-1}$, gives
$\hat Q_n(\hat t)=\alpha+o_P(1)$. Hence
$F_n^{\hat\gamma}(\hat t)=\alpha+o_P(1)$ and, by the bounded-density
condition,
\[
F_n^{\hat\gamma}(t_n^\circ)=\alpha+o_P(1).
\]

The same curve has a positive local slope at $t_n^\circ$. Indeed, for
$t_1<t_2$ in a sufficiently small neighborhood of $t_n^\circ$, the
lower bounds on $q$ and $\hat\omega$, together with their boundedness,
imply
\[
F_n^{\hat\gamma}(t_2)-F_n^{\hat\gamma}(t_1)
\geq
c\{H_n^{\hat\gamma}(t_2)-H_n^{\hat\gamma}(t_1)\}
\]
for some deterministic $c>0$. Therefore, for every fixed
$\varepsilon>0$, with probability tending to one,
\[
F_n^\gamma(t_n^\circ-\varepsilon)
<
\alpha
<
F_n^\gamma(t_n^\circ+\varepsilon).
\]
Combining this strict separation with
$\sup_t|\hat F_{n,0}^*(t)-F_n^\gamma(t)|=o_P(1)$ and the
monotonicity of $\hat F_{n,0}^*$ yields
\[
\hat\tau_0^*-t_n^\circ=o_P(1).
\]

Finally, conditional on $\cA_n$, the test point is independent of the
training and calibration data. Thus, by~\eqref{eq:risk_bd23} and
the preceding uniform convergence,
\[
R_n
\leq
\alpha
+
H_n^\gamma(\hat\tau_0^*)
-
F_n^\gamma(\hat\tau_0^*)
+
o_P(1).
\]
The bounded-density condition, $\hat\tau_0^*-t_n^\circ=o_P(1)$, and
the uniform differences between the $\gamma$- and
$\hat\gamma$-curves imply 
$H_n^\gamma(\hat\tau_0^*)=\alpha+o_P(1)$ and 
$F_n^\gamma(\hat\tau_0^*)=\alpha+o_P(1)$. 
Consequently, $R_n\leq\alpha+o_P(1)$. Marginalizing over $\cA_n$
completes the proof in this case.
\end{proof}

\subsection{Proof of Theorem~\ref{thm:dr_obs_rate_general}}
\label{app:subsec_obs_dr_rate_general}
    % \paragraph{Proof of rate-double robustness.} 
\begin{proof}[Proof of Theorem~\ref{thm:dr_obs_rate_general}]
    
We reuse the proof of Theorem~\ref{thm:dr_obs} up to~\eqref{eq:risk_bd23}  to prove the rate double robustness result. Note that the results up to~\eqref{eq:risk_bd23} did not use any assumptions on the fitted models.  Define
\[
\hat\omega^\dagger(x):=\frac{\hat\omega(x)}{p_G}.
\]
Since $q(x)\geq c$, we have $p_G=\EE[q(X)\mid\cT_n]\geq c$.
Therefore, Assumption~\ref{assump:obs_slow} implies
\[
\|\hat{\omega}^\dagger-w\|_{L_\infty(\PP_X)}
= 
\|\hat\omega-w^\circ\|_{L_\infty(\PP_X)} / p_G
=
O_P(n^{-1/4}).
\]
Moreover, replacing $\hat\omega$ by $\hat{\omega}^\dagger$ does not change
any weighted empirical or population ratio.
Recall that $\cA_n$ is the $\sigma$-field that includes the randomness in the labeled data, the training process, and generating $\{G_i\}$, and we denote the $\cA_n$-conditional harm rate  
$R_n := \PP(Y_{n+1}(\pi_{\text{obs}}(X_{n+1}))<Y_{n+1}(0)\given \cA_n)$. 
The balancing condition in Assumption~\ref{assump:obs_bal} with $r_n=O(n^{-1/2})$, together with~\eqref{eq:risk_bd23}, implies 
\$
R_n &\leq \alpha + \EE[\fna(X_{n+1})\ind\{\hat{s}(X_{n+1})\leq \hat\tau_0^*\}\given \cA_n] - \frac{  \sum_{i=1}^n G_i \hat\omega^\dagger(X_i)\fna(X_i)\ind\{\hat{s}(X_i)\leq \hat\tau_0^*\}}{\sum_{i=1}^n G_i \hat\omega^\dagger(X_i)  }  \\ 
&\qquad + \frac{\sum_{i=1}^n G_i \hat\omega^\dagger(X_i) \hat\fna(X_i) \ind\{\hat{s}(X_i)\leq \hat{t}\}}{\sum_{i=1}^n G_i\hat\omega^\dagger(X_i)} - \frac{1}{n}\sum_{i=1}^n \hat\fna(X_i) \ind\{\hat{s}(X_i)\leq \hat{t}\} +O_P(1/\sqrt{n}),
\$
where the second and the third terms invoked Lemma~\ref{lem:w_to_wx} twice. 
By Assumption~\ref{assump:obs_bal}, we have
\[
\begin{aligned}
\frac1n\sum_{i=1}^nG_i\hat\omega(X_i)
&=
\frac1n\sum_{i=1}^nG_i\hat w_i
+
O_P(1/\sqrt{n}) =
\frac{m_n}{n}\{1+O_P(1/\sqrt{n})\}
+
O_P(1/\sqrt{n}) =
p_G+O_P(1/\sqrt{n}).
\end{aligned}
\]
Consequently, 
$\frac1n\sum_{i=1}^nG_i\hat\omega^\dagger(X_i)
=
1+O_P(1/\sqrt{n})$, which further yields
\$
R_n &\leq \alpha + \EE[\fna(X_{n+1})\ind\{\hat{s}(X_{n+1})\leq \hat\tau_0^*\}\given \cA_n] - \frac{1}{n}\sum_{i=1}^n \hat\fna(X_i) \ind\{\hat{s}(X_i)\leq \hat{t}\}  \\ 
&\quad + \frac{1}{n}\sum_{i=1}^n G_i \hat\omega^\dagger(X_i) \hat\fna(X_i) \ind\{\hat{s}(X_i)\leq \hat{t}\}  - \frac{1}{n} \sum_{i=1}^n G_i \hat\omega^\dagger(X_i)\fna(X_i)\ind\{\hat{s}(X_i)\leq \hat\tau_0^*\} +O_P(1/\sqrt{n}) \\
 &\leq \alpha + \EE[\fna(X)\ind\{\hat{s}(X)\leq \hat\tau_0^*\}\given \cA_n] - \EE[\hat\fna(X ) \ind\{\hat{s}(X )\leq \hat{t}\}\given \cA_n ] \\ 
&\quad + \frac{1}{n}\sum_{i=1}^n \hat{g}(X_i,T_i) \hat\omega^\dagger(X_i) \hat\fna(X_i) \ind\{\hat{s}(X_i)\leq \hat{t}\}  - \frac{1}{n} \sum_{i=1}^n \hat{g}(X_i,T_i)  \hat\omega^\dagger(X_i)\fna(X_i)\ind\{\hat{s}(X_i)\leq \hat\tau_0^*\} +O_P(1/\sqrt{n}),
\$
where $X\sim \PP_X$ is an independent copy. 
Here, the second inequality uses Lemma~\ref{lem:wcdf_conv} for $Z_i=(\hat{G}_i-\hat{g}(X_i,T_i))\hat{\omega}^\dagger(X_i)\fna(X_i)$ or $Z_i =(\hat{G}_i-\hat{g}(X_i,T_i))\hat{\omega}^\dagger(X_i)\hat\fna(X_i)$, as well as Lemma~\ref{lem:wcdf_conv} for $Z_i=\hat\fna(X_i)$ and $H_n(t)=\EE[\hat\gamma(X)\ind\{\hat{s}(X)\leq t\}\given \cT_n]$ evaluated at $t=\hat{t}$ which is adapted to the $\sigma$-field $\cA_n$. 

For the oracle weight $w(X)$, we know $\EE[G_iw(X_i)]=\EE[\hat{g}(X,T)w(X)]=1$, and due to the covariate shift between calibration data conditional on $G_i=1$ and the test data, 
\$
& \EE[\fna(X_{n+1})\ind\{\hat{s}(X_{n+1})\leq \hat\tau_0^*\}\given \cA_n] =  \EE[\hat{g}(X,T) w(X)\fna(X)\ind\{\hat{s}(X )\leq \hat\tau_0^*\}\given \cA_n],\\
 & \EE[\hat\fna(X_{n+1})\ind\{\hat{s}(X_{n+1})\leq \hat t\}\given \cA_n] =  \EE[\hat{g}(X,T) w(X)\hat\fna(X)\ind\{\hat{s}(X )\leq \hat t\}\given \cA_n],
\$
where the expectations are over an independent copy $(X,T)\sim \PP_{X,T}$. Thus, denoting 
\$
\hat{r}(x,t) = \hat{g}(x,t) \cdot \big\{ \hat\gamma(x)\ind\{\hat{s}(x)\leq \hat{t}\} - \gamma(x) \ind\{\hat{s}(x)\leq \hat\tau_0^*\}\big\}
\$
which is a function adapted to $\cA_n$ and obeys $\hat{r}(x,t)\in [-1,1]$ for any value of $(x,t)$, since $\hat{g}(x,t)\in [0,1]$ and $\hat\fna(x),\fna(x)\in [0,1]$. 
Then, for an independent copy $(X,T)\sim \PP_{X,T}$, 
\$
R_n & \leq \alpha + \frac{1}{n}\sum_{i=1}^n  \hat{\omega}^\dagger(X_i) \hat{r}(X_i,T_i) - \EE[w(X)\hat{r}(X,T)\given \cA_n] + O_P(1/\sqrt{n}) \\
&\leq \alpha + \underbrace{\frac{1}{n}\sum_{i=1}^n \big\{\hat{\omega}^\dagger(X_i)-w(X_i)\big\} \hat{r}(X_i,T_i)}_{\text{(a)}} + \underbrace{\frac{1}{n}\sum_{i=1}^n \Big\{ w(X_i)\hat{r}(X_i,T_i) - \EE[w(X)\hat{r}(X,T)\given \cA_n] \Big\}}_{\text{(b)}} + O_P(1/\sqrt{n}).
\$
We now proceed to control the two terms separately via uniform convergence. 
Let $\cT_n$ denote the $\sigma$-field generated by the
training processes, conditional on which $\hat g$, $\hat\gamma$, and
$\hat s$ are fixed. Writing $P_n$ for the empirical measure of
$(X_i,T_i)_{i=1}^n$ and $P$ for the expectation over an independent copy,
define
\[
f_{1,u}(x,t)
:=
w(x)\hat g(x,t)\hat\gamma(x)\ind\{\hat s(x)\leq u\},
\qquad
f_{2,v}(x,t)
:=
w(x)\hat g(x,t)\gamma(x)\ind\{\hat s(x)\leq v\}.
\]
By definition, as the cutoffs $\hat\tau_0^*$ and $\hat{t}$ in the definition of $\hat{r}(\cdot,\cdot)$ is fixed given $\cA_n$, we  know 
\[
\begin{aligned}
|\text{(b)}|
&=
\left|
(P_n-P)f_{1,\hat t}
-
(P_n-P)f_{2,\hat\tau_0^*}
\right| \leq
\sup_{u\in\RR}|(P_n-P)f_{1,u}|
+
\sup_{v\in\RR}|(P_n-P)f_{2,v}| = O_P(1/\sqrt{n}),
\end{aligned}
\]
where the last equality follows from applying Lemma~\ref{lem:wcdf_conv} twice,  
provided that
\( \|w\hat g\hat\fna\|_{L_2(\PP_{X,T})}
+
\|w\hat g\fna\|_{L_2(\PP_{X,T})}
=
O_P(1)
\)
based on the boundedness of $w(\cdot)$, $\hat{g}$, $\fna$, and $\hat\fna$. 
For term $(a)$, we define for any $u,v\in\RR$ the function 
\[
r_{u,v}(x,t)
:=
\hat g(x,t)\left\{
\hat\gamma(x)\ind\{\hat s(x)\leq u\}
-
\gamma(x)\ind\{\hat s(x)\leq v\}
\right\},
\]
so that $\hat r=r_{\hat t,\hat\tau_0^*}$. We then have 
\$
|\text{(a)}|
&=
\big|P_n\{(\hat{\omega}^\dagger - w)r_{\hat t,\hat\tau_0^*}\}\big|
\leq
\big|P\{(\hat{\omega}^\dagger-w)r_{\hat t,\hat\tau_0^*}\}\big|
+
\sup_{u,v\in\RR}
\big|
(P_n-P)\{(\hat{\omega}^\dagger-w)r_{u,v}\}
\big|.
\$
By the H\"older's inequality,
\[
\left|P\{(\hat{\omega}^\dagger-w)r_{\hat t,\hat\tau_0^*}\}\right|
\leq
\|\hat{\omega}^\dagger-w\|_{L_\infty(\PP_X)}
\|\hat r\|_{L_1(\PP_{X,T})}.
\]
Moreover,
\$
&\sup_{u,v\in\RR}
\big|
(P_n-P)\{(\hat{\omega}^\dagger-w)r_{u,v}\}
\big| \\ 
&\leq
\sup_{u\in\RR}
\big|
(P_n-P)
\left[
(\hat{\omega}^\dagger-w)\hat g\hat\gamma
\ind\{\hat s\leq u\}
\right]
\big|+
\sup_{v\in\RR}
\big|
(P_n-P)
\left[
(\hat{\omega}^\dagger-w)\hat g\gamma
\ind\{\hat s\leq v\}
\right]
\big|.
\$
Conditional on $\cT_n$, both terms are indexed by nested threshold
classes. Since $\hat g,\hat\gamma,\gamma\in[0,1]$ and
$\|\hat{\omega}^\dagger-w\|_{L_2(\PP_X)}=O_P(1)$, applying
Lemma~\ref{lem:wcdf_conv} twice gives
$\sup_{u,v\in\RR}
\left|
(P_n-P)\{(\hat{\omega}^\dagger-w)r_{u,v}\}
\right|
=
O_P(n^{-1/2})$. 
Thus,
\[
|(a)|
\leq
\|\hat{\omega}^\dagger-w\|_{L_\infty(\PP_X)}
\|\hat r\|_{L_1(\PP_{X,T})}
+
O_P(n^{-1/2}).
\]
Combining the above two bounds gives 
\@\label{eq:R_bd_5}
R_n \leq \alpha + \|\hat{\omega}^\dagger-w\|_{L_\infty(\PP_X)}
\|\hat r\|_{L_1(\PP_{X,T})}
+
O_P(n^{-1/2}).
\@

We now proceed to show $\|\hat{r}\|_{L_1(\PP_{X,T})}=O_P(n^{-1/4})$. By definition, as $\hat{g}\in [0,1]$ and $\fna\in [0,1]$, 
\$
\|\hat{r}\|_{L_1(\PP_{X,T})} &\leq \big\| \hat\fna(\cdot)\ind\{\hat{s}(\cdot)\leq \hat{t}\} - \fna(\cdot) \ind\{\hat{s}(\cdot)\leq \hat\tau_0^*\}\big\|_{L_1(\PP_{X,T})} \\ 
&\leq \big\| \hat\fna(\cdot)\ind\{\hat{s}(\cdot)\leq \hat{t}\} - \fna(\cdot) \ind\{\hat{s}(\cdot)\leq \hat t\}\big\|_{L_1(\PP_{X,T})} + \big\| \fna(\cdot)\ind\{\hat{s}(\cdot)\leq \hat{t}\} - \fna(\cdot) \ind\{\hat{s}(\cdot)\leq \hat\tau_0^*\}\big\|_{L_1(\PP_{X,T})} \\ 
&\leq \|\hat\fna-\fna\|_{L_1(\PP_{X,T})} +\big\|  \ind\{\hat{s}(\cdot)\leq \hat{t}\} -  \ind\{\hat{s}(\cdot)\leq \hat\tau_0^*\}\big\|_{L_1(\PP_{X,T})} 
\leq O_P(n^{-1/4}) + O(|\hat\tau_0^*-\hat{t}|),
\$
where the last inequality uses the bounded-density condition on $\hat{s}(\cdot)$. 

In the following, we show that $|\hat\tau_0^*-\hat{t}|=O_P(n^{-1/4})$ under the given conditions. 
Recall that 
$H_n(t)
:=
\EE[
\hat\gamma(X)\ind\{\hat s(X)\leq t\}
\mid\cT_n
]$ and  
$t_n^\circ
:=
\sup\{t\in\RR:H_n(t)\leq\alpha\}$. 
By Lemma~\ref{lem:wecdf_hat},
\[
\sup_{t\in\RR}
\left|
\frac1n\sum_{i=1}^n
\hat\gamma(X_i)\ind\{\hat s(X_i)\leq t\}
-
H_n(t)
\right|
=
O_P(n^{-1/2}).
\]
Together with condition~(ii) of
Theorem~\ref{thm:dr_obs_rate_general}, the cutoff conclusion of
Lemma~\ref{lem:wecdf_hat} gives 
$|\hat t-t_n^\circ|=O_P(n^{-1/2})$. 
 
Since Assumption~\ref{assump:obs_bal} holds with
$r_n=O(n^{-1/2})$, Lemmas~\ref{lem:w_to_wx}
and~\ref{lem:wecdf_hat}, together with
$\PP(Y_i^{**}=0\mid X_i,T_i)=\gamma(X_i)$, give
\[
\sup_{t\in\RR}
|\hat F_{n,0}^*(t)-F_n^\circ(t)|
=
O_P(n^{-1/2}),\quad \text{where}\quad F_n^\circ(t)
:=
\frac{
\EE\!\left[
\hat\omega^\dagger(X)\hat g(X,T)\gamma(X)
\ind\{\hat s(X)\leq t\}
\,\middle|\,\cT_n
\right]
}{
\EE\!\left[
\hat\omega^\dagger(X)\hat g(X,T)
\,\middle|\,\cT_n
\right]
}.
\]
Write 
$D_n:=\EE\!\left[
\hat\omega^\dagger(X)\hat g(X,T)
\given\cT_n
\right]$. 
For every integrable function $h$, we have
$\EE\!\left[
w(X)\hat g(X,T)h(X)
\given \cT_n
\right]
=
\EE[h(X)\given\cT_n]$. 
It follows that 
$|D_n-1|
\leq
\|\hat\omega^\dagger-w\|_{L_\infty(\PP_X)}
=
O_P(n^{-1/4})$, 
and
\[
\begin{aligned}
\sup_{t\in\RR}
\bigg|
&\EE\!\left[
\hat\omega^\dagger(X)\hat g(X,T)\gamma(X)
\ind\{\hat s(X)\leq t\}
\,\middle|\,\cT_n
\right]
-
H_n(t)
\bigg| \leq
\|\hat\omega^\dagger-w\|_{L_\infty(\PP_X)}
+
\|\hat\gamma-\gamma\|_{L_1(\PP_X)}
=
O_P(n^{-1/4}).
\end{aligned}
\]
Since $D_n=1+o_P(1)$ and $0\leq H_n(t)\leq1$, on the event
$D_n\geq1/2$, it holds that  
$\sup_{t\in\RR}|F_n^\circ(t)-H_n(t)|
=
O_P(n^{-1/4})$. 
Therefore, by the triangle inequality,
\[
\begin{aligned}
\varepsilon_n
&:=
\sup_{t\in\RR}
|\hat F_{n,0}^*(t)-H_n(t)|  \leq
\sup_{t\in\RR}
|\hat F_{n,0}^*(t)-F_n^\circ(t)|
+
\sup_{t\in\RR}
|F_n^\circ(t)-H_n(t)| =
O_P(n^{-1/4}).
\end{aligned}
\]

Let $\mathcal E_n$ denote the event on which 
$H_n(t_n^\circ-u)\leq\alpha-cu$ and 
$H_n(t_n^\circ+u)\geq\alpha+cu$ 
for every $0<u\leq\eta$. By condition~(ii) of
Theorem~\ref{thm:dr_obs_rate_general}, we know 
$\PP(\mathcal E_n)\to1$.
Set 
$\bar\varepsilon_n:=\varepsilon_n+n^{-1}$ and 
$u_n:=\frac{2\bar\varepsilon_n}{c}$. 
Since $\varepsilon_n=O_P(n^{-1/4})$, we have
$u_n=O_P(n^{-1/4})$ and 
$\PP\{\mathcal E_n\cap\{u_n\leq\eta\}\}\to1$. 
On this event,
\[
\begin{aligned}
\hat F_{n,0}^*(t_n^\circ-u_n)
&\leq
H_n(t_n^\circ-u_n)+\varepsilon_n \leq
\alpha-cu_n+\varepsilon_n
<
\alpha, \\
\hat F_{n,0}^*(t_n^\circ+u_n)
&\geq
H_n(t_n^\circ+u_n)-\varepsilon_n \geq
\alpha+cu_n-\varepsilon_n
>
\alpha.
\end{aligned}
\]
Because $\hat F_{n,0}^*$ is non-decreasing, it follows that
$t_n^\circ-u_n
\leq
\hat\tau_0^*
\leq
t_n^\circ+u_n$. 
Therefore, we have  
$|\hat\tau_0^*-t_n^\circ|
=
O_P(n^{-1/4})$, 
and hence
\[
|\hat\tau_0^*-\hat t|
\leq
|\hat\tau_0^*-t_n^\circ|
+
|\hat t-t_n^\circ|
=
O_P(n^{-1/4}).
\]
Putting it back to~\eqref{eq:R_bd_5}, we conclude the proof of Theorem~\ref{thm:dr_obs_rate_general}. 
\end{proof}

\subsection{Proof of Lemma~\ref{lem:bal_weight_approx}}
\label{app:subsec_lem_bal_weight_approx}

\begin{proof}[Proof of Lemma~\ref{lem:bal_weight_approx}]
% For completeness, define
% $\bar\psi_n(t):=n^{-1}\sum_{i=1}^n\psi_t^{\hat\gamma,\tilde w}(X_i)$ and
% \[
% \hat M_n(t)
% :=
% \frac1{m_n}\sum_{i\in\cI_{\calib}}
% \psi_t^{\hat\gamma,\tilde w}(X_i)
% \psi_t^{\hat\gamma,\tilde w}(X_i)^\top .
% \]
% Let
% $B_n(t)
% :=
% \left\{
% b\in\RR^3:
% b_1=1,\ 
% |b_k-\bar\psi_{n,k}(t)|\leq\delta_n,\ k\in\{2,3\}
% \right\}$, 
% and, whenever $\hat M_n(t)$ is invertible, define
% \$
% \hat b_n(t)\in\argmin_{b\in B_n(t)}
% b^\top\hat M_n(t)^{-1}b, 
% \quad \hat a_n(t):=\hat M_n(t)^{-1}\hat b_n(t), \quad \hat\omega_t(x):=
% \psi_t^{\hat\gamma,\tilde w}(x)^\top\hat a_n(t).
% \$

Let $\cN:=\{t\in\RR:|t-t^\circ|\leq\eta\}$ and write
$P_{n,G}f:=m_n^{-1}\sum_{i\in\cI_{\calib}}f(X_i)$ and
$P_Gf:=\EE[f(X)\mid G=1,\cT_n]$. We work conditional on $\cT_n$
and on the event in Assumption~\ref{assump:bal_reg}, whose probability
tends to one. Conditional on $\cT_n$, the pairs $(X_i,G_i)$ are
i.i.d., with
$\PP(G_i=1\mid X_i=x,\cT_n)=q(x)$ and
$\PP(G_i=1\mid\cT_n)=p_G$. All the conditional stochastic bounds below
therefore also hold marginally.

For brevity, throughout the first part of the proof write
$\psi_t:=\psi_t^{\hat\gamma,\tilde w}$,
$M(t):=M_{\hat\gamma,\tilde w}(t)$,
$b(t):=b_{\hat\gamma,\tilde w}(t)$, and
$\omega_t:=\omega_t^{\hat\gamma,\tilde w}$.
The coordinates of $\psi_t$ and their pairwise products are uniformly
bounded and indexed by the nested sets $\{\hat s(X)\leq t\}$.
Applying Lemma~\ref{lem:wecdf_hat} coordinatewise gives
\@\label{eq:bal_unif_concentration}
\sup_{t\in\cN}\|\bar\psi_n(t)-b(t)\|
&=O_P(n^{-1/2}),\qquad
\frac{m_n}{n}=p_G+O_P(n^{-1/2}),\\
\sup_{t\in\cN}\|\hat M_n(t)-M(t)\|_{\text{op}}
&=O_P(n^{-1/2}), \qquad 
\sup_{u\in\RR}
\left|
P_{n,G}\ind\{\hat s(X)\leq u\}
-
P_G\ind\{\hat s(X)\leq u\}
\right|
 =O_P(n^{-1/2}).
\@
Since $p_G$ and
$\inf_{t\in\cN}\lambda_{\min}\{M(t)\}$ are bounded away from zero,
$\hat M_n(t)$ is invertible uniformly over $t\in\cN$ with probability
tending to one.

Define $H_n(t):=n^{-1}\sum_{i=1}^n
\hat\gamma(X_i)\ind\{\hat s(X_i)\leq t\}$ and
$H(t):=\EE[\hat\gamma(X)\ind\{\hat s(X)\leq t\}\mid\cT_n]$.
Lemma~\ref{lem:wecdf_hat} gives
$\sup_t|H_n(t)-H(t)|=O_P(n^{-1/2})$. The local-slope condition in
Assumption~\ref{assump:bal_reg} therefore yields
\@\label{eq:bal_cutoff_rate}
|\hat t-t^\circ|=O_P(n^{-1/2}).
\@
In particular, $\hat t\in\cN$ with probability tending to one.

We next characterize the optimizer. Fix $t\in\cN$ and let $\Psi_t$
be the $m_n\times3$ matrix whose row indexed by
$i\in\cI_{\calib}$ is $\psi_t(X_i)^\top$. For any $b\in B_n(t)$,
the unique vector $u\in\RR^{m_n}$ minimizing $\|u\|^2$ subject to
$m_n^{-1}\Psi_t^\top u=b$ is
$u=\Psi_t\hat M_n(t)^{-1}b$, with objective value
$m_nb^\top\hat M_n(t)^{-1}b$. Hence the unconstrained solution to the
balancing program is
$\hat w_i(t)=\hat\omega_t(X_i)$.
By the definition of $B_n(t)$, equation~\eqref{eq:bal_unif_concentration},
and $\delta_n=O(n^{-1/2})$, we know 
$\sup_{t\in\cN}\|\hat b_n(t)-b(t)\|=O_P(n^{-1/2})$. The identity
$\hat M_n^{-1}-M^{-1}
=\hat M_n^{-1}(M-\hat M_n)M^{-1}$ then gives
\@\label{eq:bal_coefficient_rate}
\sup_{t\in\cN}\|\hat a_n(t)-M(t)^{-1}b(t)\|
=O_P(n^{-1/2}),
\qquad
\sup_{ t\in\cN, x\in\cX}
|\hat\omega_t(x)-\omega_t(x)|
=O_P(n^{-1/2}).
\@
By Assumption~\ref{assump:bal_reg}(ii), we have
$\inf_{t\in\cN,x\in\cX}\omega_t(x)\geq c_1$.
Thus, with probability tending to one,
$\hat\omega_t(x)\geq0$ uniformly over $t\in\cN$ and $x\in\cX$.
The unconstrained solution is therefore feasible for the nonnegative
balancing program and, by strict convexity, is its unique solution.
Taking $t=\hat t$ proves
$\hat w_i=\hat\omega_{\hat t}(X_i)$.

We now compare the empirical weights with
$\mathcal W_n(\hat\gamma,\tilde w)=\omega_{t^\circ}$.
Since $q(x)\leq1$ and $p_G$ is bounded away from zero,
Assumption~\ref{assump:bal_reg}(iv) implies
$P_G\{a<\hat s(X)\leq b\}\leq C(b-a)$ for all $a<b$.
Consequently, uniformly over $t,t'\in\cN$,
\[
\|b(t)-b(t')\|+\|M(t)-M(t')\|_{\text{op}}\leq C|t-t'|,
\qquad
\|M(t)^{-1}b(t)-M(t')^{-1}b(t')\|
\leq C|t-t'|.
\]
Write $M(t)^{-1}b(t)=(a_0(t),a_h(t),a_w(t))^\top$. Then
\[
\omega_t(x)-\omega_{t'}(x)
=
\psi_t(x)^\top
\{M(t)^{-1}b(t)-M(t')^{-1}b(t')\}
+
a_h(t')\{h_t^{\hat\gamma}(x)-h_{t'}^{\hat\gamma}(x)\},
\]
and hence
\@\label{eq:bal_weight_cutoff_lipschitz}
|\omega_t(x)-\omega_{t'}(x)|^2
\leq
C|t-t'|^2+
C\ind\{\min(t,t')<\hat s(x)\leq\max(t,t')\}.
\@
Equations~\eqref{eq:bal_unif_concentration} and
\eqref{eq:bal_cutoff_rate} imply
\[
P_{n,G}
\ind\{\min(\hat t,t^\circ)<\hat s(X)
\leq\max(\hat t,t^\circ)\}
=O_P(n^{-1/2}).
\]
It follows from \eqref{eq:bal_weight_cutoff_lipschitz} that
\@\label{eq:bal_population_cutoff_error}
P_{n,G}
\{\omega_{\hat t}(X)-\omega_{t^\circ}(X)\}^2
=O_P(n^{-1/2}).
\@
On the other hand, equation~\eqref{eq:bal_coefficient_rate} gives
\@\label{eq:bal_empirical_coefficient_error}
P_{n,G}
\{\hat\omega_{\hat t}(X)-\omega_{\hat t}(X)\}^2
=
\{\hat a_n(\hat t)-M(\hat t)^{-1}b(\hat t)\}^\top
\hat M_n(\hat t)
\{\hat a_n(\hat t)-M(\hat t)^{-1}b(\hat t)\}
=O_P(n^{-1}).
\@
Combining
\eqref{eq:bal_population_cutoff_error}--\eqref{eq:bal_empirical_coefficient_error}
proves
\[
\frac1{m_n}\sum_{i\in\cI_{\calib}}
\left\{
\hat w_i-\mathcal W_n(\hat\gamma,\tilde w)(X_i)
\right\}^2
=O_P(n^{-1/2}).
\]

The boundedness of $\psi_t$, together with
Assumption~\ref{assump:bal_reg}(ii), implies
\[
\sup_{ t\in\cN,x\in\cX}
|\omega_t(x)|
\leq
\sup_{ t\in\cN,x\in\cX}
\|\psi_t(x)\|
\sup_{t\in\cN}\|M(t)^{-1}\|_{\text{op}}
\sup_{t\in\cN}\|b(t)\|
\leq C.
\]
Together with~\eqref{eq:bal_coefficient_rate} and
$\PP(\hat t\in\cN)\to1$, this gives 
$\sup_{x\in\cX}
|\hat\omega_{\hat t}(x)|
=
O_P(1)$. 
Since
$\hat w_i=\hat\omega_{\hat t}(X_i)$ for
$i\in\cI_{\calib}$ and
$\hat w_{n+1}=\hat\omega_{\hat t}(X_{n+1})$, it follows that
$\hat w_{n+1}
\vee
\max_{i\in\cI_{\calib}}\hat w_i
=
O_P(1)$. 
The normalization constraint gives
$\sum_{i\in\cI_{\calib}}\hat w_i=m_n$, while
$m_n/n=p_G+O_P(n^{-1/2})$ and $p_G$ is bounded away from zero.
Therefore,
\[
\frac{
\hat w_{n+1}
\vee
\max_{i\in\cI_{\calib}}\hat w_i
}{
\sum_{i\in\cI_{\calib}}\hat w_i
}
=
O_P(n^{-1}).
\]

Suppose now that
$\|\tilde w-w\|_{L_2(\PP_X)}=o_P(1)$.
Let $a^\circ:=(0,0,p_G)^\top$. Since $q(x)w(x)=1$, we have,
uniformly over $t\in\cN$, that
$b(t)
=
\EE[q(X)\psi_t(X)w(X)\given \cT_n]$. 
On the other hand, since
$\psi_t(x)^\top a^\circ=p_G\tilde w(x)$, we have
\[
M(t)a^\circ
=
\EE[q(X)\psi_t(X)\tilde w(X)\given \cT_n].
\]
Consequently, we have 
$b(t)-M(t)a^\circ
=
\EE\!\left[
q(X)\psi_t(X)\{w(X)-\tilde w(X)\}
\given \cT_n
\right]$. 
By Assumption~\ref{assump:bal_reg}, $\psi_t$ is uniformly bounded and
$M(t)^{-1}$ is uniformly bounded over $t\in\cN$. Hence,
by the Cauchy--Schwarz inequality,
\[
\sup_{t\in\cN}
\|M(t)^{-1}b(t)-a^\circ\|
\leq
C\|\tilde w-w\|_{L_2(\PP_X)}
=o_P(1).
\]
Moreover, since $w^\circ(x)=p_Gw(x)$, we have 
\@
\omega_t(x)-w^\circ(x)
&=
\psi_t(x)^\top
\{M(t)^{-1}b(t)-a^\circ\}
+
p_G\{\tilde w(x)-w(x)\}.
\@
It follows that 
$\sup_{t\in\cN}
\|\omega_t-w^\circ\|_{L_2(\PP_X)}
=o_P(1)$. 
Taking $t=t^\circ$ and recalling that
$\mathcal W_n(\hat\fna,\tilde w)=\omega_{t^\circ}$ gives
\[
\|\mathcal W_n(\hat\fna,\tilde w)
-w^\circ\|_{L_2(\PP_X)}
=o_P(1).
\]

Finally, suppose $\varepsilon_n
:=
\|\tilde w-w\|_{L_\infty(\PP_X)}
=
O_P(n^{-1/4})$. 
The preceding argument now yields
\@
\sup_{t\in\cN}
\|M(t)^{-1}b(t)-a^\circ\|
&\leq C\varepsilon_n, \quad 
\sup_{t\in\cN}|a_h(t)|
\leq C\varepsilon_n,
\@
where $a_h(t)$ denotes the coefficient on
$h_t^{\widehat\gamma}$ in $M(t)^{-1}b(t)$. Therefore, 
$\sup_{t\in\cN}
\|\omega_t-w^\circ\|_{L_\infty(\PP_X)}
=O_P(n^{-1/4})$, 
and in particular,
\[
\|\mathcal W_n(\hat\gamma,\tilde w)
-w^\circ\|_{L_\infty(\PP_X)}
=O_P(n^{-1/4}).
\]

To sharpen the empirical approximation rate, recall that, uniformly
over $t,t'\in\cN$,
\[
\begin{aligned}
|\omega_t(x)-\omega_{t'}(x)|^2
\leq{}&
C|t-t'|^2 +
C\left\{\sup_{u\in\cN}|a_h(u)|^2\right\}
\ind\{\min(t,t')<\widehat s(x)\leq\max(t,t')\}.
\end{aligned}
\]
The bounds already established above imply
$|\widehat t-t^\circ|=O_P(n^{-1/2})$ and
\[
P_{n,G}
\ind\{\min(\widehat t,t^\circ)<\widehat s(X)
\leq\max(\widehat t,t^\circ)\}
=O_P(n^{-1/2}).
\]
Since
$\sup_{u\in\cN}|a_h(u)|^2=O_P(n^{-1/2})$, it follows that
$P_{n,G}
\{\omega_{\widehat t}(X)-\omega_{t^\circ}(X)\}^2
=O_P(n^{-1})$. 
Combining this with
$P_{n,G}
\{\widehat\omega_{\widehat t}(X)
-\omega_{\widehat t}(X)\}^2
=O_P(n^{-1})$ 
gives
\[
\frac1{m_n}\sum_{i\in\cI_{\calib}}
\left\{
\widehat w_i-
\mathcal W_n(\widehat\gamma,\tilde w)(X_i)
\right\}^2
=O_P(n^{-1}),
\]
which completes the proof.
\end{proof}

\section{Technical proof for results in appendix}

\subsection{Proof of Theorem~\ref{thm:global_opt_general}}
\label{app:subsec_proof_global_opt_general}

\begin{proof}[Proof of Theorem~\ref{thm:global_opt_general}]
    Following~\cite{kallus2022s} and~\cite{li2023trustworthy},  for any decision rule $\phi\colon \cX\to \{0,1\}$, 
    \$
\max_{P\in \cP} ~\text{Err}(\phi;P)  = \EE[\phi(X)\fna(X)],
    \$
    where $\fna(x)=\min\{\PP(Y(1)=0\given X=x), \PP(Y(0)=1\given X=x)\}$ solely relies on the observed distribution, and $\EE[\cdot]$ is with respect to the observed distribution. On the other hand, the power is the expectation of $\phi(X)$ under the observed covariate distribution. The optimization problem (with randomized policy) is 
    \$
    \maximize_{\phi\colon \cX\to [0,1]} ~~&\EE[\phi(X)]\\ 
    \text{subject to} ~~& \EE[\phi(X)\fna(X)] \leq \alpha.
    \$
    If $\EE[\fna(X)]\leq \alpha$, then $\phi_{\text{power}}^*=1$ maximizes the power since its power is one. 
    
    Otherwise, suppose $\EE[\fna(X)]>\alpha$. 
    Consider any decision rule $\psi(\cdot)$ obeying $\EE[\psi(X)\fna(X)]\leq \alpha$, and write $\phi^*(\cdot)=\phi_{\text{power}}^*(\cdot)$ for notational convenience. Note that 
    \$
     &\fna^*\cdot \big(\EE[\psi(X)]-\EE[\phi^*(X)]\big)\\  
    &= \EE\big[ \fna^*\cdot (\psi(X)-\phi^*(X))\ind\{\fna(X)<\fna^*\} \big]\\ 
    &\qquad \quad + \EE\big[ \fna^*\cdot (\psi(X)-\phi^*(X))\ind\{\fna(X)>\fna^*\} \big] + \EE\big[ \fna^*\cdot (\psi(X)-\phi^*(X))\ind\{\fna(X)=\fna^*\} \big] .
    \$
    The first term, as $\phi^*(X)=1$ when $\fna(X)<\fna^*$ and hence $\psi(X)-\phi^*(X)\leq 0$, obeys 
    \$
    \EE\big[ \fna^*\cdot (\psi(X)-\phi^*(X))\ind\{\fna(X)<\fna^*\} \big] \leq \EE\big[ \fna(X)(\psi(X)-\phi^*(X))\ind\{\fna(X)<\fna^*\} \big].
    \$
    Similar arguments yield 
    \$
    &\EE\big[ \fna^*\cdot (\psi(X)-\phi^*(X))\ind\{\fna(X)>\fna^*\} \big] \leq \EE\big[ \fna(X)(\psi(X)-\phi^*(X))\ind\{\fna(X)>\fna^*\} \big], \\ 
    &\EE\big[ \fna^*\cdot (\psi(X)-\phi^*(X))\ind\{\fna(X)=\fna^*\} \big] = \EE\big[ \fna(X)(\psi(X)-\phi^*(X))\ind\{\fna(X)=\fna^*\} \big].
    \$
    Adding the three terms up, we know 
    \$
\fna^*\cdot \big(\EE[\psi(X)]-\EE[\phi^*(X)]\big) 
\leq \EE[\fna(X)\psi(X)] - \EE[\fna(X)\phi^*(X)] \leq \alpha-\alpha=0
    \$
    since $\psi(X)$ is a feasible solution and the satefy constraint is binding for $\phi^*$. Finally, as $\fna^*\geq 0$, this implies $\EE[\psi(X)]\leq \EE[\phi^*(X)]$. The arbitrariness of $\psi$ implies the optimality of $\phi^*$. 
\end{proof}

\subsection{Proof of Theorem~\ref{thm:power_opt_str}}
\label{app:subsec_opt_power_ours}

\begin{proof}[Proof of Theorem~\ref{thm:power_opt_str}]
Let $\widehat a(x)\in\{0,1\}$ denote the arm selected by the
inclusion rule, and define
\[
\overline\fna_n(x)
:=
\{1-\mu_1(x)\}\mathbf{1}\{\widehat a(x)=1\}
+
\mu_0(x)\mathbf{1}\{\widehat a(x)=0\}.
\]
As in the proof of Theorem~\ref{thm:opt_stratified_experiments}, the optimality of
$\widehat a(x)$ for the estimated proxy-label risks gives
\[
0
\leq
\overline\fna_n(x)-\fna(x)
\leq 2\left\{
|\widehat\mu_1(x)-\mu_1(x)|
+ |\widehat\mu_0(x)-\mu_0(x)| \right\}.
\]
Hence
$\|\overline\fna_n-\fna\|_{L_1(P_X)}
=o_{\mathbb P}(1)$. The assumed outcome-model consistency also
implies
$\|\widehat\fna-\fna\|_{L_1(P_X)}
=o_{\mathbb P}(1)$.

Conditional on the training process, let
\[
\rho_n(x)
:=
e(x)\mathbf{1}\{\widehat a(x)=1\}
+
\{1-e(x)\}\mathbf{1}\{\widehat a(x)=0\},
\qquad
w_n(x):=\rho_n(x)^{-1}.
\]
The weighted identities in the proof of Theorem~\ref{thm:opt_stratified_experiments} give 
$\mathbb{E}[G w_n(X)\mid X]=1$  and 
$\mathbb{E}\!\left[
G w_n(X)\mathbf{1}\{Y^\dagger=0\}
\,\middle|\,X
\right]=\overline\fna_n(X)$. 
It follows by the same uniform-law-of-large-numbers argument that
the weighted empirical conformal curve converges uniformly to 
$H(t)
:=
\mathbb{E} [
\gamma(X)\mathbf{1}\{\gamma(X)\leq t\}]$. 
Consequently,
\[
p^{\mathrm{str-power}}_{n+1}
-
H\{\gamma(X_{n+1})\}
=
o_{\mathbb P}(1).
\]

Since $\gamma(X)$ has no point mass, the same threshold argument
as in the proof of Theorem~\ref{thm:opt_stratified_experiments} gives
\[
\mathbb{P}\!\left\{
\widehat\pi_{\mathrm{str-power}}(X_{n+1})
\neq
\pi^*_{\mathrm{power}}(X_{n+1})
\right\} \to 0,
\]
where the binding and nonbinding cases are characterized by
Theorem~\ref{thm:global_opt_general}. Therefore, 
$| 
\mathbb{E} [
\widehat\pi_{\mathrm{str-power}}(X_{n+1})
 ]
-
\text{Power}(\pi^*_{\mathrm{power}};P)|\to 0$, which completes the proof.
\end{proof}

\subsection{Proof of Theorem~\ref{thm:obs_power_opt}}
\label{app:subsec_obs_power_opt}

\begin{proof}[Proof of Theorem~\ref{thm:obs_power_opt}]
Define
\[
H_n(t)
:=
\EE\!\left[
\hat\gamma(X)\ind\{\hat\gamma(X)\leq t\}
\,\middle|\,\cT_n
\right],
\qquad
t_n^\circ
:=
\sup\{t\in\RR:H_n(t)\leq\alpha\},
\]
and
\[
H_0(t)
:=
\EE[
\gamma(X)\ind\{\gamma(X)\leq t\}
].
\]
By $\|\hat\gamma-\gamma\|_{L_2(\PP_X)}=o_P(1)$, the
no-point-mass condition on $\gamma(X)$, and
Lemma~\ref{lem:wecdf_hat}, we have 
\[
\sup_{t\in\RR}|H_n(t)-H_0(t)|=o_P(1).
\]
Together with the local-crossing condition in
Assumption~\ref{assump:bal_reg}(iii), this implies
$t_n^\circ-t_{\power}^*=o_P(1)$, 
where 
$t_{\power}^*
:=
\sup\{t\in\RR:H_0(t)\leq\alpha\}$. 
Set
$\hat\omega:=\mathcal W_n(\hat\gamma,\tilde w)$ and define
\[
F_n^{\hat\gamma}(t)
:=
\frac{
\EE[
\hat\omega(X)\hat g(X,T)\hat\gamma(X)
\ind\{\hat\gamma(X)\leq t\}
\mid\cT_n]
}{
\EE[\hat\omega(X)\hat g(X,T)\mid\cT_n]
}.
\]
By the defining population balance equations, we know  
$F_n^{\hat\gamma}(t_n^\circ)
=
H_n(t_n^\circ)
=
\alpha$. 
Moreover, Assumption~\ref{assump:bal_reg}(i)--(iii) implies that
$F_n^{\hat\gamma}$ crosses $\alpha$ at $t_n^\circ$ with a slope
bounded away from zero, with probability tending to one.
Define
\[
\hat F_n^{\obs}(t)
:=
\frac{
\hat w_{n+1}
+
\sum_{i=1}^n
G_i\hat w_i\ind\{Y_i^\dagger=0\}
\ind\{\hat\gamma(X_i)\leq t\}
}{
\hat w_{n+1}
+
\sum_{i=1}^nG_i\hat w_i
},
\]
so that
$p_{n+1}^{\obs}
=
\hat F_n^{\obs}\{\hat\gamma(X_{n+1})\}$.
By Lemma~\ref{lem:bal_weight_approx},
Lemmas~\ref{lem:w_to_wx} and~\ref{lem:wecdf_hat}, and 
$\frac{\hat w_{n+1}}
{\sum_{i\in\cI_{\calib}}\hat w_i}
=
O_P(n^{-1})$, we obtain the uniform convergence of the empirical curve $\hat{F}_n^{\obs}(t)$ to its  
training-conditional population counterpart.
Furthermore, by the definition of $g^*$, we know  
\$
&\EE[
g^*(X,T)\ind\{Y^\dagger=0\}
\given X] \\
&= \EE\big[\EE[\ind\{Y^\dagger=0\}\given X]T\ind\{1-\mu_1(X)\leq \mu_0(X)\} \biggiven X] + \EE\big[\EE[\ind\{Y^\dagger=0\}\given X](1-T)\ind\{1-\mu_1(X)> \mu_0(X)\}\biggiven X]  \\
&= \EE\big[(1-\mu_1(X))T\ind\{1-\mu_1(X)\leq \mu_0(X)\} \biggiven X] + \EE\big[\mu_0(X)(1-T)\ind\{1-\mu_1(X)> \mu_0(X)\}\biggiven X] \\
&=\EE[g^*(X,T)\gamma(X)\given X]
\$
Therefore,
\[
\sup_{t\in\RR}
|\hat F_n^{\obs}(t)-F_n^{\hat\gamma}(t)|
\leq
o_P(1)
+
C\|\hat g-g^*\|_{L_1(\PP_{X,T})}
+
C\|\hat\gamma-\gamma\|_{L_1(\PP_X)}
=
o_P(1).
\]

Fix $0<\varepsilon\leq\eta$. The local-crossing and uniform
convergence results imply that, with probability tending to one,
\[
\hat\gamma(X_{n+1})\leq t_n^\circ-\varepsilon
\quad\Rightarrow\quad 
p_{n+1}^{\obs}\leq\alpha, 
\qquad \text{whereas}  \qquad 
\hat\gamma(X_{n+1})\geq t_n^\circ+\varepsilon
\quad\Rightarrow\quad
p_{n+1}^{\obs}>\alpha.
\]
Consequently,
\[
\left|
\hat\pi_{\obs}(X_{n+1})
-
\ind\{\hat\gamma(X_{n+1})\leq t_n^\circ\}
\right|
\leq
\ind\{
|\hat\gamma(X_{n+1})-t_n^\circ|\leq\varepsilon
\}
+o_P(1).
\]
Assumption~\ref{assump:bal_reg}(iv) and the arbitrariness of
$\varepsilon>0$ yield
\[
\EE\left[
\left|
\hat\pi_{\obs}(X_{n+1})
-
\ind\{\hat\gamma(X_{n+1})\leq t_n^\circ\}
\right|
\right]
\to0.
\]

Finally, since
$\hat\gamma\to\gamma$ in $L_2(\PP_X)$,
$t_n^\circ\to t_{\power}^*$, and
$\gamma(X)$ has no point mass,
\[
\EE\left[
\left|
\hat\pi_{\obs}(X_{n+1})
-
\ind\{\gamma(X_{n+1})\leq t_{\power}^*\}
\right|
\right]
\to0.
\]
By Theorem~\ref{thm:global_opt}, we know  
$\pi_{\power}^*(x)
=
\ind\{\gamma(x)\leq t_{\power}^*\}$. 
Hence
\[
\left|
\EE[\hat\pi_{\obs}(X_{n+1})]
-
\EE[\pi_{\power}^*(X_{n+1})]
\right|
\leq
\EE\left[
\left|
\hat\pi_{\obs}(X_{n+1})
-
\pi_{\power}^*(X_{n+1})
\right|
\right]
\to0,
\]
which proves 
$\EE[\hat\pi_{\obs}(X_{n+1})]
\to
\text{Power}(\pi_{\power}^*;\PP)$. 
\end{proof}

\subsection{Proof of Theorem~\ref{thm:welfare_optimal}}
\label{app:subsec_proof_welfare_opt}

\begin{proof}[Proof of Theorem~\ref{thm:welfare_optimal}]

The result follows directly from Theorem~\ref{thm:welfare_optimal_general}. In particular, since
$s_{\mathrm{welfare}}(X)$ has no point mass,
$\mathbb{P}\{s_{\mathrm{welfare}}(X)=r^*\}=0$, so the randomization
at the cutoff in Theorem~\ref{thm:welfare_optimal_general} is immaterial. Therefore, the optimal
randomized rule in Theorem~\ref{thm:welfare_optimal_general} reduces almost surely to
$\pi^*_{\mathrm{welfare}}(x)
= \ind\{s_{\mathrm{welfare}}(x)\leq r^*\}$.
The binding/nonbinding characterizations follow from the two cases in Theorem~\ref{thm:welfare_optimal_general}.
\end{proof}

\begin{proof}[Proof of Theorem~\ref{thm:welfare_optimal_general}]

We first characterize the worst-case safety constraint. For any $P'\in\mathcal P$, define \[ f_{P'}(x) := P'\{Y(1)=0,Y(0)=1\mid X=x\}. \] For a randomized policy $\phi$, whose randomization is independent of the potential outcomes conditional on $X$, its harm rate under $P'$ is $\operatorname{Err}(\phi;P') =\mathbb{E}[\phi(X)f_{P'}(X)]$. Since $f_{P'}(x)\leq P'\{Y(1)=0\mid X=x\}=1-\mu_1(x)$ and $f_{P'}(x)\leq P'\{Y(0)=1\mid X=x\}=\mu_0(x)$, we have $f_{P'}(x)\leq\gamma(x)$ for every $P'\in\mathcal P$. This upper bound is sharp. Indeed, for every $x$, consider the following conditional distribution of the potential outcomes: \[ \begin{aligned} P'\{Y(1)=0,Y(0)=1\mid X=x\} &=\gamma(x),\\ P'\{Y(1)=0,Y(0)=0\mid X=x\} &=1-\mu_1(x)-\gamma(x),\\ P'\{Y(1)=1,Y(0)=1\mid X=x\} &=\mu_0(x)-\gamma(x),\\ P'\{Y(1)=1,Y(0)=0\mid X=x\} &=\mu_1(x)-\mu_0(x)+\gamma(x). \end{aligned} \] All four probabilities are nonnegative. 
The first three nonnegativity claims follow directly from $\gamma(x)=\min\{1-\mu_1(x),\mu_0(x)\}$, while the last follows from $\gamma(x)\geq\mu_0(x)-\mu_1(x)$. They sum to one and yield the conditional marginals $P'\{Y(1)=1\mid X=x\}=\mu_1(x)$ and $P'\{Y(0)=1\mid X=x\}=\mu_0(x)$. 
Combining this conditional coupling with the original distribution of $X$ and the same treatment-assignment mechanism therefore defines a distribution in $\mathcal P$. Consequently, for every randomized policy $\phi$, \[ \max_{P'\in\mathcal P}\operatorname{Err}(\phi;P') = \mathbb{E}[\gamma(X)\phi(X)]. \] Write $\tau(x):=\mu_1(x)-\mu_0(x)$. 
The welfare of a randomized policy is 
\$ 
\operatorname{Welfare}(\phi;P) &= \mathbb{E}\!\left[ \mu_1(X)\phi(X) + \mu_0(X)\{1-\phi(X)\} \right]\\ &= \mathbb{E}[\mu_0(X)] + \mathbb{E}[\tau(X)\phi(X)]. 
\$
Thus, up to the constant $\mathbb{E}[\mu_0(X)]$, the randomized extension of~\eqref{eq:def_opt_welfare_phi} is equivalent to 
\[ 
\underset{\phi:\mathcal X\to[0,1]}{\operatorname{maximize}} \quad \mathbb{E}[\tau(X)\phi(X)] \qquad \text{subject to} \qquad \mathbb{E}[\gamma(X)\phi(X)]\leq\alpha. 
\] 
We next verify that the cutoff and randomization probability in the theorem are well defined. Recall that  $s_{\mathrm{welfare}}(x) = -\frac{\tau(x)}{\gamma(x)}$  under the stated ratio conventions, and define, for $r\leq0$, 
\[ 
H(r) := \mathbb{E}\!\left[ \gamma(X)  \ind\{s_{\mathrm{welfare}}(X)<r\} \right]. 
\] 
The function $H$ is nondecreasing and left-continuous on $(-\infty,0]$. Indeed, if $r_k\uparrow r$, then $ \ind\{s_{\mathrm{welfare}}(X)<r_k\}$ increases pointwise to $ \ind\{s_{\mathrm{welfare}}(X)<r\}$, so the conclusion follows from the dominated convergence theorem. Moreover, $H(r)\to0$ as $r\to-\infty$. To see this, if $s_{\mathrm{welfare}}(X)$ is finite, then $ \ind\{s_{\mathrm{welfare}}(X)<r\}\to0$. If $s_{\mathrm{welfare}}(X)=-\infty$, then necessarily $\gamma(X)=0$, so $\gamma(X) \ind\{s_{\mathrm{welfare}}(X)<r\}=0$ for every $r$. The dominated convergence theorem therefore gives the claim, and hence the set defining $r^*$ is nonempty. Suppose first that $H(0)>\alpha$. Since $H(r)\uparrow H(0)$ as $r\uparrow0$, there exists some $r_0<0$ such that $H(r_0)>\alpha$. It follows that $r^*<0$. By the definition of $r^*$, there exists a sequence $r_k\uparrow r^*$ such that $H(r_k)\leq\alpha$. The left-continuity of $H$ therefore implies $H(r^*)\leq\alpha$. 
On the other hand, the right limit of $H$ at $r^*$ is \[ \lim_{r\downarrow r^*}H(r) = \mathbb{E}\!\left[ \gamma(X)  \ind\{s_{\mathrm{welfare}}(X)\leq r^*\} \right] = H(r^*) + \mathbb{E}\!\left[ \gamma(X)  \ind\{s_{\mathrm{welfare}}(X)=r^*\} \right]. \] This limit must be at least $\alpha$. 
Otherwise, there would exist some $r>r^*$ sufficiently close to $r^*$ such that $H(r)<\alpha$, contradicting the definition of $r^*$ as the supremum. Therefore, \[ H(r^*) \leq \alpha \leq H(r^*) + \mathbb{E}\!\left[ \gamma(X)  \ind\{s_{\mathrm{welfare}}(X)=r^*\} \right]. \] This establishes the existence of $\eta^*\in[0,1]$ satisfying the equality in the theorem. If the expectation multiplying $\eta^*$ is zero, the preceding inequalities imply $H(r^*)=\alpha$, and we may take $\eta^*=0$. If instead $H(0)\leq\alpha$, the theorem sets $r^*=0$ and $\eta^*=0$. Under the stated ratio conventions, $s_{\mathrm{welfare}}(x)<0$ if and only if $\tau(x)>0$. Hence, in this case, \[ \phi^*_{\mathrm{welfare}}(x) =  \ind\{\tau(x)>0\}, \qquad \mathbb{E}\!\left[ \gamma(X)\phi^*_{\mathrm{welfare}}(X) \right] = H(0) \leq\alpha. \] It follows in both cases that $\phi^*_{\mathrm{welfare}}$ is feasible and satisfies  $(-r^*) \left\{ \mathbb{E}\!\left[ \gamma(X)\phi^*_{\mathrm{welfare}}(X) \right] -\alpha \right\} =0$ by complementary slackness. 

It remains to establish optimality. Let $\psi:\mathcal X\to[0,1]$ be any feasible randomized policy. By the definition of $\phi^*_{\mathrm{welfare}}$, it holds pointwise that \[ \{\tau(x)+r^*\gamma(x)\} \{\psi(x)-\phi^*_{\mathrm{welfare}}(x)\} \leq0. \] Indeed, if $s_{\mathrm{welfare}}(x)<r^*$, then $\tau(x)+r^*\gamma(x)\geq0$ and $\phi^*_{\mathrm{welfare}}(x)=1$, so the second factor is nonpositive. If $s_{\mathrm{welfare}}(x)>r^*$, then $\tau(x)+r^*\gamma(x)\leq0$ and $\phi^*_{\mathrm{welfare}}(x)=0$, so the second factor is nonnegative. If $s_{\mathrm{welfare}}(x)=r^*$, the first factor is zero. The same conclusions continue to hold when $\gamma(x)=0$ under the stated ratio conventions. Consequently, 
\$ 
& \operatorname{Welfare}(\psi;P) - \operatorname{Welfare}(\phi^*_{\mathrm{welfare}};P)\\ &\quad= \mathbb{E}\!\left[ \tau(X) \{\psi(X)-\phi^*_{\mathrm{welfare}}(X)\} \right]\\ &\quad= \mathbb{E}\!\left[ \{\tau(X)+r^*\gamma(X)\} \{\psi(X)-\phi^*_{\mathrm{welfare}}(X)\} \right]  -r^* \mathbb{E}\!\left[ \gamma(X) \{\psi(X)-\phi^*_{\mathrm{welfare}}(X)\} \right]\\ &\quad\leq -r^* \left\{ \mathbb{E}[\gamma(X)\psi(X)] - \mathbb{E}[\gamma(X)\phi^*_{\mathrm{welfare}}(X)] \right\}\\ &\quad\leq -r^* \left\{ \alpha - \mathbb{E}[\gamma(X)\phi^*_{\mathrm{welfare}}(X)] \right\} =0. 
\$ 
The first inequality follows from the preceding pointwise inequality. The second follows from the feasibility of $\psi$ and the fact that $-r^*\geq0$. The final equality follows from complementary slackness. Since $\psi$ was arbitrary, this proves the optimality of $\phi^*_{\mathrm{welfare}}$.
\end{proof}

\section{Simulation details}

\subsection{Details for experiments in Figure~\ref{fig:rct_fna}}
\label{app:subsec_simu_rct_details}

This section includes the omitted details that produce Figure~\ref{fig:rct_fna}.

\vspace{-0.5em}
\paragraph{Data generating process.}  
Across all the four settings, we set the feature dimension as $X\in \RR^{20}$. In Settings 1-2,   the linear coefficients are obtained by a random i.i.d.~draw from Unif$[(1,2)]$ and fixed before running all the experiments. 

In Setting 1, we draw $X\iid N(0,\mathbf{I})$ and  set $\mu_1(x)=\mathrm{logit}^{-1}(0.0835x_1+0.097x_2+0.0885x_3+0.0575x_4+0.0675x_5+1)$ and  $\mu_0(x)=1-0.8\mu_1(x)-0.2\,\mathrm{logit}^{-1}(-0.87  x_4-0.58 x_5-0.675 x_6-0.54  x_7-0.57  x_8)$. Both functions are then truncated to $[0.05,0.95]$. 

In Setting 2, we draw $X\iid \cN(0, \Sigma)$ with  $\Sigma=0.5 \ind \ind^\top+0.5\mathbf{I}$ and set $\mu_1(x)=\mathrm{logit}^{-1}(x^\top\theta_1+0.8)$ where $[\theta_1]_{1:10}=(1.66,1.71,1.93,1.72,1.99,1.98,1.52,1.57,1.37,1.53)\times 0.02$, and $\mu_0(x)=1-0.9\mu_1(x)-0.4\,\mathrm{logit}^{-1}(x^\top\theta_0)$ where   $[\theta_0]_{4:13}=-(1.73,1.92,1.97,1.70,1.78,1.37,1.78,1.76,1.59,1.32)\times 0.2 $; both functions are then truncated to $[0.05,0.95]$.  Finally,  we take $X_{1:7}$ as the observed features.  

In the nonlinear Setting 3, we draw entries of $X$ i.i.d.~from $\mathrm{Unif}(0,1)$ and set $\mu_1(x)=\mathrm{logit}^{-1}(f_1(x))$, where $f_1(x)=\mathrm{signal}\cdot\{\sin(3\pi x_1)+\cos(2\pi x_2)+x_3^2- \ind(x_4>0.5)x_4^2+0.5\, \ind(x_5>0.3)+e^{-x_4}\}+1$; we then truncate $\mu_1(x)$ to $[0.05,0.95]$. Next, we define $\mu_0(x)=1-0.2\mu_1(x)-0.8\mu_1(x)^2$. 

In Setting 4, we draw $X$ from the same correlated Gaussian design as in Setting 2, and set $\mu_1(x)=\mathrm{logit}^{-1}(f_1(x))$ with $f_1(x)=0.2 \cdot\{\sin(3\pi x_1)+\cos(2\pi x_8)+0.5x_3^2- \ind(x_4>0.5)x_5^2+0.5\, \ind(x_6>0.3)+e^{-x_7}\}+0.8+0.1x_8$, truncated to $[0.05,0.95]$. We then define $\mu_0(x)=1-\mu_1(x)+0.1\, \ind(x_6>0.5)-0.2x_8^2$, truncated to $[0.08,0.95]$. Finally, after generating the outcomes and treatments, we take $X_{1:5}$ as the observed features.

\vspace{-0.5em}
\paragraph{Additional implementation details.} For \texttt{Li et al.}, we implement the doubly-robust estimator $\hat{u}_{\text{FNA}}(\pi)$ in~\citet[Lemma 6.1]{li2023trustworthy} with uniform cost $c(X)=1$, where the $\mu_t(\cdot)$ are fitted using the same random forest classifier in \texttt{scikit-learn} Python library with 2-fold cross-fitting. The policy class $\pi\in\Pi$ is depth-2 policy trees from the \texttt{econml} Python library. Given these estimators we fit the optimization problem (2) in \cite{li2023trustworthy} via the dual form and a bi-search to find the dual variable. 
The \texttt{Policy-Tree} baseline follows the standard use of \texttt{econml} library, where we use $\cD_{\text{label}}$ to learn the policy and apply to $\cD_{\test}$.  

\subsection{Details for additional experiments in Section~\ref{subsec:simu_rct}}
\label{app:subsec_rct_aux}

% \paragraph{Data generating process.}
For all the additional experiments in Section~\ref{subsec:simu_rct}, we sample $X  \in [0,1]^{20}\iid\mathrm{Unif}([0,1]^{20})$ and assign treatment independently as $T \sim \mathrm{Bernoulli}(1/2)$.  We state the three sets of DGPs below. 

\vspace{-0.5em}
\paragraph{Arm informativeness and selective calibration.} In this set of experiments, we vary the scale of $\mu_1(x)$ and $\mu_0(x)$ to vary whether the harm rate relies on the treated or control outcome, i.e., whether $\rho(x)=1-\mu_1(x)$ or $\rho(x) = \mu_0(x)$. 
Let $q \in \{0,0.25,0.5,0.75,1\}$ be a parameter that controls the comparison of the two outcome models. We define the latent score $r_{\text{arm}}(x)
=
\sin(\pi x_1x_2)
+x_3
-0.5x_4
+0.3\,\ind\{x_5>0.4\}$, and  $q_r\in \RR$ be the $q$-th population quantile of $r_{\text{arm}}(X)$. For $r_{\text{arm}}(x)\leq q_r$, we set
\[
\mu_1(x)=1-\rho(x),
\qquad
\mu_0(x)=\rho(x)+\Delta(x),
\]
whereas for $r_{\text{arm}}(x)> q_r$, we set
\[
\mu_0(x)=\rho(x),
\qquad
\mu_1(x)=1 - \rho(x)-\Delta(x).
\]
Thus, the parameter $q$ controls the fraction of units for which the treated arm is the informative arm. Here we use a nonlinear harm rate function $\rho(x)=0.08+0.22\,\text{expit}\{\sin(2\pi x_1)
+0.7\cos(2\pi x_2)
+0.5(x_3-0.5)^2
-0.6\,\ind\{x_4>0.6\}
+0.5e^{-x_5}\},$ and the treatment effect function $\Delta(x) = \{0.35 + 0.35+0.25\,\text{expit}(r_\delta(x))\}\cdot \max\{1-2\rho(x)-0.02,\;0.02\}$, where $r_\delta(x) = \cos(2\pi x_3)
-0.6\sin(2\pi x_5)
+0.8x_6
-0.4\,\ind\{x_7>0.5\}$, and $\text{expit}(u)=1/(1+e^{-u})$.

\paragraph{Subgroup heterogeneity.}
For this experiment we use a three-group construction. Let $\tilde x_1=x_1-0.5$ and $\tilde x_2=x_2-0.5$, and define
\[
v(x)=\sigma\!\left(\sin(2\pi x_3)-0.8\cos(2\pi x_4)+0.6x_5\right).
\]
We partition the covariate space into three groups: the first group $G_1$ consists of samples with $\tilde{x}_1>0$; the second group $G_2$ consists of samples with $\tilde{x}_1\leq 0$ and $\tilde{x}_2\leq 0$; the third group $G_3$ consists of samples with $\tilde{x}_1\leq 0$ and $\tilde{x}_2>0$. 
For $x \in G_1$, we set
\[
\mu_0(x)=0.35,\qquad \mu_1(x)=0.65.
\]
For $x \in G_2 \cup G_3$, we define a small heterogeneous treatment effect 
$ 
\tau (x)
=
\max\{0.004, \min\{ 0.03, 
0.008+0.01125\bigl(0.3+0.7 v(x) \bigr)\}. 
$ 
Then for $x \in G_2$, we set
\[
\mu_1(x)=\max\{0.94, ~\min\{ 0.995, ~0.95+0.02 v(x)\}\},
\qquad
\mu_0(x)=\max\{0.9, ~\min\{ 0.99,~ \mu_1(x)-\tau (x)\}\},
\]
so both outcomes are near one and the treatment effect is small and positive. Note that in $G_2$, the harm rate is small, and one recognizes this only when using the treated outcomes for calibration. For $x \in G_3$, we set
\[
\mu_0(x)=\max\{0.005, ~\min\{ 0.08,~0.02+0.03 v(x)\}\}
\qquad
\mu_1(x)=\max\{0.001, ~\min\{ 0.06,~\mu_0(x)-\tau (x)\}\},
\]
so both arms are near zero and the treatment effect is small and negative. In $G_3$, the harm rate is close to zero, and the method recognizes this only when using the control outcomes for calibration. 

% \paragraph{Power--welfare tradeoff.} 
% In this set of experiments, we design artificial settings where the ranking of scores (and hence p-values) using the fitted harm rate function and the fitted reward-harm rate ratio differ.
% Let $s(x)=\textrm{expit}(\cos(2\pi x_1)
% -\sin(2\pi x_3)
% +0.7x_6
% -0.5\,\ind\{x_8>0.5\}.)$ and let $\gamma \in \{0.1,0.2,0.3,0.4,0.5\}$. We set
% \[
% \mu_0(x)=0.05+0.25\,s(x),
% \qquad
% \tau(x)=0.5\,s(x)^{1+\gamma}+0.05\,s(x)^\gamma,
% \qquad
% \mu_1(x)=\min\{\mu_0(x)+\tau(x),\,1\}.
% \]
% This construction creates units with different values of the worst-case harm rate score $\rho(x)=\min\{1-\mu_1(x),\mu_0(x)\}$ and the welfare-to-risk trade-off. Increasing $\gamma$ amplifies heterogeneity in the ranking induced by the power-oriented and welfare-oriented selection rules.

\subsection{Details for experiments in Section~\ref{subsec:simu_obs}}
\label{app:subsec_simu_obs}

\paragraph{Data generating process for Figure~\ref{fig:obs_fna}.} In the four settings, the process that generates $(X, Y(1),Y(0))$ is the same as those in Figure~\ref{fig:rct_fna} for RCT experiments. 
Given the features $\{X_i\}$, we sample the treatments $\{T_i\}$ independently from Bernoulli$(e(X_i))$  according to a propensity score $e(x)=\mathbb{P}(T=1\mid X=x)=\mathrm{logit}^{-1}(\eta_e(x))$. In the linear settings 1-2, we set $\eta_e(x)=0.25 \cdot (1.1x_1+0.9x_2+0.6x_3-0.4x_4+0.5x_5+0.3x_6)$. In the nonlinears 3-4, we set $\eta_e(x)=0.25\cdot (\sin(3\pi x_1)+\cos(2\pi x_8)+0.1x_3^2-0.3x_4)$.

\vspace{-0.5em}
\paragraph{Data generating process for double robustness results.}
We generate covariates $X\in \RR^{20}$ with entries i.i.d.~from $\mathrm{Unif}(0,1)$. We define three latent threshold indicators by $q_1(x)= \ind\{x_1>0.5\}$, $q_2(x)= \ind\{x_3>0.5\}$, and $q_3(x)= \ind\{x_5>0.5\}$. These induce three hidden subgroups: $G_1(x)= \ind\{q_1(x)=1,q_2(x)=1\}$, $G_2(x)= \ind\{q_1(x)=1,q_2(x)=0\}$, and $G_3(x)= \ind\{q_1(x)=0,q_3(x)=1\}$. The potential outcomes satisfy $Y(t)\mid X=x \sim \mathrm{Bernoulli}(\mu_t(x))$ for $t\in\{0,1\}$. We set $\mu_0(x)=0.07+0.02x_7+0.01(x_8-0.5)+0.38\, G_1(x)+0.03\,G_3(x)$ and $\tau(x)=\mu_1(x)-\mu_0(x)=0.03+0.01(x_8-x_7)-0.48\,G_1(x)+0.42\,G_2(x)+0.14\,G_3(x)$, and then truncate $\mu_0(x)$ to $[0.03,0.82]$ and set $\mu_1(x)=\mu_0(x)+\tau(x)$ truncated to $[\mu_0(x)+0.01,0.95]$.  

Treatment is assigned observationally according to the propensity score $e(x)=\mathrm{logit}^{-1}(-1.2+\beta\,G_1(x)+1.2\,G_2(x)+0.6\,G_3(x))$, truncated to $[0.1,0.9]$, and we sample $T\mid X=x \sim \mathrm{Bernoulli}(e(x))$. Here the parameter $\beta\in \{2.8, 3.4, 4.0, 4.4 \}$ controls the strength of confounding  in four levels shown in the $x$-axis of Figure~\ref{fig:obs_aux}. 
The observed outcome is $Y=TY(1)+(1-T)Y(0)$. We generate $(Y(1),Y(0))$ under the negative coupling scheme. 

\vspace{-0.5em}
\paragraph{Methods for double robustness results.}
To study double robustness, we compare five nuisance-model regimes. In \texttt{both\_correct}, both the propensity and outcome models are fit by logistic regression on the transformed features $(q_1,q_2,q_3,G_1 ,G_2,G_3,x_7,x_8)$. In \texttt{both\_wrong}, both are fit by logistic regression on the raw covariates $X$. In \texttt{ps\_correct\_only}, the propensity model uses the transformed features while the outcome models use the raw covariates, and in \texttt{outcome\_correct\_only} the reverse is used. In the \texttt{rf} regime, both the propensity and outcome models are fit by random forests on the raw covariates. The other data splitting and model fitting details are the same as other experiments.

\subsection{Real data analysis details}
\label{app:econml-welfare}

For the real-data analysis, we define the binary outcome as $Y_i = \mathbb{I}(\texttt{Post\_Belief\_Specific}_i < 50)$ and the treatment indicator as $T_i = \mathbb{I}(\texttt{messageType}_i = \texttt{Active})$. Because the study is a randomized experiment, treatment assignment is known by design, so we do not estimate a propensity score model. 

We split the sample in two stages. First, we reserve 20\% of observations as the test set. We then split the remaining 80\% evenly into a 40\% training set and a 40\% calibration set. All prediction models are fit on the training set only, and predictions are generated for the calibration and test sets.

The feature set includes demographic covariates (\texttt{Education\_Cat}, \texttt{AgeYears}, \texttt{Race\_*}, \texttt{Gender\_*}, \texttt{religion}), political and psychological covariates (\texttt{Extremism}, \texttt{AOT}, \texttt{IH}, \texttt{Party\_*}), AI-related covariates (\texttt{genai\_fam\_1}, \texttt{genai\_use\_1}, \texttt{genai\_trust}, \texttt{Sureness\_1}), and baseline state variables (\texttt{Pre\_Belief\_Specific}, \texttt{LowConfidence}). Missing values in the structured covariates   are imputed using mean imputation. These variables are then standardized using scaling parameters estimated on the training set.
To incorporate the text information in the stated conspiracy, we use a pre-computed embedding for the (pre-treatment) conspiracy for each observation. We apply principal components analysis (PCA) to the training-sample embeddings and retain the first 20 principal components. The final feature vector is the concatenation of the standardized structured covariates, the standardized dialogue-length variables, and 20 PCA components.

We estimate the conditional mean outcomes under treatment and control separately using the regression forest in \texttt{econml} Python package.  Specifically, one regression forest is fit on the treated training subsample to estimate $\mu_1(x) = \mathbb{E}[Y \mid X=x, T=1]$, and a second regression forest is fit on the control training subsample to estimate $\mu_0(x) = \mathbb{E}[Y \mid X=x, T=0]$. These two forests produce predictions $\hat\mu_1(x)$ and $\hat\mu_0(x)$ on the calibration and test sets.
To obtain a direct estimate of treatment heterogeneity, we additionally fit a causal forest via the \texttt{econml} package on the full training sample using the same features. This model produces a direct estimate of the conditional average treatment effect, $\hat\tau_{\mathrm{CF}}(x)$, where $\tau(x) = \mu_1(x) - \mu_0(x)$.

Our welfare-based selection rule combines the direct causal-forest estimate of treatment benefit with a conservative proxy for treatment risk. Define $\hat\rho(x) = \min\{1-\hat\mu_1(x), \hat\mu_0(x)\}$.  We then define the welfare score as $S_{\mathrm{welfare}}(x) = -\hat\tau_{\mathrm{CF}}(x) / \max\{\hat\rho(x), c\}$, where $c > 0$ is a small numerical floor to avoid instability when $\hat\rho(x)$ is close to zero. In our implementation, we set $c = 0.025$. 

Our power-based selection rule directly uses the $\hat\rho(x)$ above in the score.

\section{Auxiliary lemmas}

    \begin{lemma}\label{lem:conv_cdf_diff}
        Suppose the distribution of a random variable $X\in \RR$ has no point mass. Then $\sup_t \PP(t-\delta<X\leq t)$ as a function of $\delta$ converges to zero as $\delta\to 0$. 
    \end{lemma}
    \begin{proof}[Proof of Lemma~\ref{lem:conv_cdf_diff}]
     Since $X$ has no point mass, the c.d.f.~$F(t):= \PP(X\leq t)$ is continuous and non-decreasing. 
     Consider any constant $\epsilon>0$. There exists constants $a,b\in \RR$ such that $F(a)\leq \epsilon$ and $1-F(b)\leq \epsilon$. 
     Second, since $F$ is continuous on the compact set $[a,b]$, the Heine--Cantor theorem implies $F$ is uniformly continuous on $[a,b]$, which means there exists some $\delta>0$ such that $\sup_{x,y\in [a,b], |x-y|\leq \delta} |F(x)-F(y)|\leq \epsilon$. 
     Now we take any sequence 
     \$
     a = t_1<t_2<\cdots <t_M=b,
     \$
     such that $t_{i+1}-t_i\leq \delta$. By the arguments above, we know  the following holds:
    \$
    \PP(t_i < X \leq t_{i+1})\leq \epsilon,\quad \PP(X\leq t_1)\leq \epsilon,\quad \PP(X> t_{M})\leq \epsilon.
    \$
    Therefore, for any  $t\in [t_1+\delta,t_M+\delta]$, there exists some $i$ such that $\delta_{i}\leq t-\delta < t \leq \delta_{i+2}$ and thus $\PP(t-\delta<X\leq t)\leq 2\epsilon$. For any $t \leq t_1+\delta$, we know   $\PP(t-\delta<X\leq t)\leq \PP(X\leq t_1+\delta) $  Therefore, for any $\epsilon>0$, we know $\sup_t  \PP( X\leq t <  \hat{s}(X))  \leq o_P(1)+ \epsilon$. The arbitrariness of $\epsilon>0$ the implies the desired result. 
    \end{proof}

\begin{lemma}\label{lem:bd_cdf}
    Suppose two fixed functions $s_1,s_2\colon\cX\to \RR$ obeys $\|s_1(X)-s_2(X)\|_{L_2} = o(1)$ and one of them has no point mass. Then for any fixed $t\in \RR$, we have $\PP(s_1(X)\leq t) - \PP(s_2(X)\leq t) = o(1)$. 
\end{lemma}

\begin{proof}[Proof of Lemma~\ref{lem:bd_cdf}]
    Without loss of generality we assume $s_1(X)$ has no point mass, so the mapping $t\mapsto \PP(s_1(X)\leq t)$ is continuous on $t\in \RR$.
    The $L_2$ convergence implies the convergence in probability. That is, $\PP(|s_1(X)-s_2(X)|>\epsilon)\to 0$ for any fixed $\epsilon>0$. Due to the continuity of $\PP(s_1(X)\leq t)$ in $t\in \RR$ and the above convergence,  for any $\delta>0$,    we can find a sufficiently small $\epsilon>0$ such that $\PP(s_1(X)\leq t-\epsilon) \geq \PP(s_1(X)\leq t) -\delta$, $\PP(s_1(X)\leq t+ \epsilon)\leq \PP(s_1(X)\leq t)+\delta$,  and $ \PP(  |s_1(X)-s_2(X)|>\epsilon)\leq \delta$.  Therefore,
    \$
     \PP(s_2(X)\leq t) & \leq  \PP(s_2(X)\leq t, |s_1(X)-s_2(X)|>\epsilon) + \PP(s_2(X)\leq t, |s_1(X)-s_2(X)|\leq \epsilon) \\ 
     &\leq  \PP(  |s_1(X)-s_2(X)|>\epsilon) + \PP(s_1(X)\leq t+\epsilon) \leq \PP(s_1(X)\leq t) + 2\delta. 
    \$
    By the same arguments, 
    \$
    \PP(s_1(X)\leq t-\epsilon) \leq  \PP(  |s_1(X)-s_2(X)|>\epsilon) + \PP(s_2(X)\leq t ),
    \$
    which further implies
    \$
    \PP(s_2(X)\leq t)\geq \PP(s_1(X)\leq t) - 2\delta. 
    \$
    The arbitrariness of $\delta>0$ thus implies the desired result.
\end{proof}

\begin{lemma}\label{lem:wcdf_conv}
Let $\cT_n$ be a $\sigma$-field representing a possibly random training process. Conditional on $\cT_n$, let $\{(X_i,Z_i)\}_{i=1}^n$ be i.i.d.\ copies of $(X,Z)$, where $Z\in\RR$, and let $s\colon\cX\to\RR$ be fixed. Suppose
$ 
 \EE[Z^2\mid\cT_n]=O_P(1).
$ 
Define
\[
H_n(t):=\frac{1}{n}\sum_{i=1}^n Z_i\ind\{s(X_i)\leq t\},\qquad
H(t):=\EE[Z\ind\{s(X)\leq t\}\mid\cT_n],\qquad
\Delta_n:=\sup_{t\in\RR}|H_n(t)-H(t)|.
\]
Then $\Delta_n=O_P(n^{-1/2})$, where the stochastic order is with respect to both the training process and the evaluation sample.
In addition, suppose that $Z\geq0$ almost surely conditional on $\cT_n$, and define
\[
\hat\tau:=\sup\bigg\{t\in\RR:\frac{1}{n}\sum_{i=1}^nZ_i\ind\{s(X_i)\leq t\}\leq\alpha\bigg\},\qquad
\tau^*:=\sup\left\{t\in\RR:\EE[Z\ind\{s(X)\leq t\}\mid\cT_n]\leq\alpha\right\}.
\]
Assume that there exist constants $c>0$ and $\delta>0$ such that, with probability tending to one over the training process,
\[
H(\tau^*-u)\leq\alpha-cu,\qquad
H(\tau^*+u)\geq\alpha+cu
\]
for every $0<u\leq\delta$. Then $\hat\tau-\tau^*=O_P(n^{-1/2})$.
\end{lemma}

\begin{proof}[Proof of Lemma~\ref{lem:wcdf_conv}]
For $t\in\RR$, set $g_t(x,z):=z\ind\{s(x)\leq t\}$. Conditional on $\cT_n$, the functions $g_t$ are fixed and $(X_i,Z_i)_{i=1}^n$ are i.i.d. By the conditional symmetrization inequality,
\[
\EE[\Delta_n\mid\cT_n]
\leq\frac{2}{n}\EE\left[
\sup_{t\in\RR}\left|\sum_{i=1}^n\varepsilon_i
Z_i\ind\{s(X_i)\leq t\}\right|
\,\middle|\,\cT_n\right],
\]
where $\varepsilon_1,\ldots,\varepsilon_n$ are i.i.d.\ Rademacher random variables independent of everything else.

Condition further on $(X_i,Z_i)_{i=1}^n$, reorder the observations so that
$s(X_{(1)})\leq\cdots\leq s(X_{(n)})$, and let $a_j:=Z_{(j)}$. Since the sets
$\{i:s(X_i)\leq t\}$ are nested,
\[
\sup_{t\in\RR}\left|\sum_{i=1}^n\varepsilon_i
Z_i\ind\{s(X_i)\leq t\}\right|
\leq
\max_{0\leq k\leq n}\left|\sum_{j=1}^k\widetilde\varepsilon_j a_j\right|,
\]
where $(\widetilde\varepsilon_j)_{j=1}^n$ is again a Rademacher sequence. Setting
$M_k:=\sum_{j=1}^k\widetilde\varepsilon_j a_j$, Doob's $L_2$ maximal inequality gives
\[
\EE_\varepsilon\left[
\max_{0\leq k\leq n}|M_k|^2
\,\middle|\,(X_i,Z_i)_{i=1}^n,\cT_n\right]
\leq4\sum_{j=1}^n a_j^2.
\]
Consequently,
\[
\EE[\Delta_n\mid\cT_n]
\leq\frac{4}{n}\EE\left[
\left(\sum_{i=1}^n Z_i^2\right)^{1/2}
\,\middle|\,\cT_n\right]
\leq\frac{4}{\sqrt n}
\left\{\EE[Z^2\mid\cT_n]\right\}^{1/2}.
\]
Let $V_n:=\{\EE[Z^2\mid\cT_n]\}^{1/2}$. For any $K,M>0$, conditional Markov's inequality yields
\[
\PP(\sqrt n\,\Delta_n>M)
\leq\PP(V_n>K)+\frac{4K}{M}.
\]
Since $V_n=O_P(1)$, this proves $\Delta_n=O_P(n^{-1/2})$.

For the second claim, let $\mathcal E_n$ denote the event on which the local-slope condition holds, and define
\[
A_n:=\mathcal E_n\cap\left\{\Delta_n\leq\frac{c\delta}{4}\right\}.
\]
Then $\PP(A_n)\to1$. On $A_n$, let $u_n:=4\Delta_n/c$, so that $u_n\leq\delta$. If $\Delta_n=0$, then $H_n\equiv H$ and hence $\hat\tau=\tau^*$. Otherwise,
\[
H_n(\tau^*+u_n)
\geq H(\tau^*+u_n)-\Delta_n
\geq\alpha+cu_n-\Delta_n
=\alpha+3\Delta_n>\alpha,
\]
whereas
\[
H_n(\tau^*-u_n)
\leq H(\tau^*-u_n)+\Delta_n
\leq\alpha-cu_n+\Delta_n
=\alpha-3\Delta_n<\alpha.
\]
Because $Z_i\geq0$, the function $H_n$ is nondecreasing. It follows that
\[
\tau^*-u_n\leq\hat\tau\leq\tau^*+u_n,
\]
and therefore $|\hat\tau-\tau^*|\leq4\Delta_n/c$ on $A_n$. Since $\Delta_n=O_P(n^{-1/2})$ and $\PP(A_n)\to1$, we conclude that $\hat\tau-\tau^*=O_P(n^{-1/2})$.
\end{proof}

\begin{lemma}\label{lem:wecdf_hat}
    Consider a sequence of random functions
    $\hat{f}_n\colon \cX\times\cZ\to \RR^+$
    and
    $\hat{s}_n\colon \cX\to \RR$
    obeying
    \[
    \|\hat{f}_n-f\|_{L_2(\PP_{X,Z})}=o_P(1),
    \qquad
    \|\hat{s}_n-s\|_{L_2(\PP_{X,Z})}=o_P(1),
    \]
    where $f\in L_2(\PP_{X,Z})$ is nonnegative and the distribution
    of $s(X)$ has no point masses. Let
    $\{(X_i,Z_i)\}_{i=1}^n$ be i.i.d.\ samples from $\PP_{X,Z}$ and
    independent of the training processes of $\hat{f}_n$ and
    $\hat{s}_n$. Define
    \[
    \begin{aligned}
    \hat{H}_n(t)
    &=
    \frac{1}{n}\sum_{i=1}^n
    \hat{f}_n(X_i,Z_i)
    \ind\{\hat{s}_n(X_i)\leq t\}, \quad 
    H_n(t)
    =
    \EE\left[
    \hat{f}_n(X,Z)
    \ind\{\hat{s}_n(X)\leq t\}
    \right],\\
    \widetilde{H}_n(t)
    &=
    \EE\left[
    f(X,Z)
    \ind\{\hat{s}_n(X)\leq t\}
    \right],\quad 
    H(t)
    =
    \EE\left[
    f(X,Z)
    \ind\{s(X)\leq t\}
    \right],
    \end{aligned}
    \]
    and
    \[
    \hat{\Delta}_n
    =
    \sup_{t\in\RR}
    |\hat{H}_n(t)-H(t)|,
    \qquad
    \Delta_n
    =
    \sup_{t\in\RR}
    |H_n(t)-H(t)|.
    \]
    Here, the expectations defining $H_n$, $\widetilde{H}_n$, and
    $H$ are taken over an independent copy
    $(X,Z)\sim\PP_{X,Z}$.
    Then
    \[
    \sup_{t\in\RR}
    |H_n(t)-\widetilde{H}_n(t)|
    =
    o_P(1),
    \qquad
    \sup_{t\in\RR}
    |\hat{H}_n(t)-\widetilde{H}_n(t)|
    =
    o_P(1),
    \]
    and
    $ 
    \Delta_n=o_P(1)$m $
    \hat{\Delta}_n=o_P(1).
    $
    In addition, suppose that
    \( 
    H(t^*-\varepsilon)
    <
    \alpha
    <
    H(t^*+\varepsilon)\) holds for every $\varepsilon>0$,  
    where
    \( 
    \hat{t}
    =
    \sup\{t\in\RR:\hat{H}_n(t)\leq\alpha\}\) and 
    \(t^*
    =
    \sup\{t\in\RR:H(t)\leq\alpha\}.
    \)
    Then
    \[
    \hat{t}-t^*=o_P(1).
    \]
\end{lemma}

\begin{proof}[Proof of Lemma~\ref{lem:wecdf_hat}]
Write $P_ng:=n^{-1}\sum_{i=1}^n g(X_i,Z_i)$ and $Pg:=\EE[g(X,Z)]$, and let $\cT_n$ be the $\sigma$-field generated by the training processes of $\hat f_n$ and $\hat s_n$. Define
\[
A_n:=\sup_{t\in\RR}\left|(P_n-P)\bigl(\hat f_n\ind\{\hat s_n\leq t\}\bigr)\right|,\quad
C_n:=\sup_{t\in\RR}|H_n(t)-\widetilde H_n(t)|,\quad
D_n:=\sup_{t\in\RR}|\widetilde H_n(t)-H(t)|.
\]
Then $\Delta_n\leq C_n+D_n$, $\hat\Delta_n\leq A_n+C_n+D_n$, and
\[
\sup_{t\in\RR}|\hat H_n(t)-\widetilde H_n(t)|\leq A_n+C_n.
\]

We first control $A_n$. Conditional on $\cT_n$, the functions $\hat f_n$ and $\hat s_n$ are fixed and the evaluation observations are i.i.d. By conditional symmetrization,
\[
\EE[A_n\mid\cT_n]\leq \frac{2}{n}\EE\left[\sup_{t\in\RR}\left|\sum_{i=1}^n\varepsilon_i\hat f_n(X_i,Z_i)\ind\{\hat s_n(X_i)\leq t\}\right|\;\middle|\;\cT_n\right],
\]
where $\varepsilon_1,\ldots,\varepsilon_n$ are i.i.d. Rademacher random variables independent of everything else. Conditional further on the evaluation sample, reorder the observations so that $\hat s_n(X_{(1)})\leq\cdots\leq\hat s_n(X_{(n)})$ and set $a_j:=\hat f_n(X_{(j)},Z_{(j)})$. Since the sets $\{i:\hat s_n(X_i)\leq t\}$ are nested, they are prefixes of this ordering. Thus, by Doob's $L_2$ maximal inequality and Jensen's inequality,
\[
\EE_\varepsilon\left[\sup_{t\in\RR}\left|\sum_{i=1}^n\varepsilon_i\hat f_n(X_i,Z_i)\ind\{\hat s_n(X_i)\leq t\}\right|\;\middle|\;(X_i,Z_i)_{i=1}^n,\cT_n\right]
\leq 2\left(\sum_{j=1}^n a_j^2\right)^{1/2}.
\]
Consequently,
\[
\EE[A_n\mid\cT_n]\leq \frac{4}{n}\EE\left[\left(\sum_{i=1}^n\hat f_n(X_i,Z_i)^2\right)^{1/2}\;\middle|\;\cT_n\right]
\leq \frac{4}{\sqrt n}\|\hat f_n\|_{L_2(\PP_{X,Z})}.
\]
Since
\[
\|\hat f_n\|_{L_2(\PP_{X,Z})}\leq \|f\|_{L_2(\PP_{X,Z})}+\|\hat f_n-f\|_{L_2(\PP_{X,Z})}=O_P(1),
\]
conditional Markov's inequality gives $A_n=O_P(n^{-1/2})=o_P(1)$.

Next,
\[
C_n=\sup_{t\in\RR}\left|P\bigl((\hat f_n-f)\ind\{\hat s_n\leq t\}\bigr)\right|
\leq P|\hat f_n-f|
\leq \|\hat f_n-f\|_{L_2(\PP_{X,Z})}
=o_P(1).
\]

It remains to control $D_n$. For each $t\in\RR$, let
\[
E_{n,t}:=\left\{\ind\{\hat s_n(X)\leq t\}\neq\ind\{s(X)\leq t\}\right\}.
\]
By the Cauchy--Schwarz inequality,
\[
D_n\leq \|f\|_{L_2(\PP_{X,Z})}\left(\sup_{t\in\RR}P(E_{n,t})\right)^{1/2}.
\]
For any $\eta>0$,
\[
E_{n,t}\subseteq \{|s(X)-t|\leq\eta\}\cup\{|\hat s_n(X)-s(X)|>\eta\},
\]
and therefore
\[
\sup_{t\in\RR}P(E_{n,t})
\leq \omega(\eta)+P(|\hat s_n(X)-s(X)|>\eta)
\leq \omega(\eta)+\eta^{-2}\|\hat s_n-s\|_{L_2(\PP_{X,Z})}^2,
\]
where $\omega(\eta):=\sup_{t\in\RR}P(|s(X)-t|\leq\eta)$. Since the distribution of $s(X)$ has no point masses, its distribution function is continuous and hence uniformly continuous, which implies $\omega(\eta)\to0$ as $\eta\downarrow0$. For every fixed $\eta>0$, the second term is $o_P(1)$. Letting first $n\to\infty$ and then $\eta\downarrow0$ yields $\sup_tP(E_{n,t})=o_P(1)$, and hence $D_n=o_P(1)$.

We have thus shown
\[
\sup_{t\in\RR}|H_n(t)-\widetilde H_n(t)|=C_n=o_P(1)
\]
and
\[
\sup_{t\in\RR}|\hat H_n(t)-\widetilde H_n(t)|
\leq A_n+C_n=o_P(1).
\]
Moreover,
\[
\Delta_n\leq C_n+D_n=o_P(1),\qquad
\hat\Delta_n\leq A_n+C_n+D_n=o_P(1).
\]

Finally, fix $\varepsilon>0$ and define
\[
\gamma_\varepsilon:=\frac12\min\left\{\alpha-H\left(t^*-\frac{\varepsilon}{2}\right),
H\left(t^*+\frac{\varepsilon}{2}\right)-\alpha\right\}>0.
\]
On the event $\{\hat\Delta_n<\gamma_\varepsilon\}$,
\[
\hat H_n\left(t^*-\frac{\varepsilon}{2}\right)
\leq H\left(t^*-\frac{\varepsilon}{2}\right)+\hat\Delta_n<\alpha
\]
and
\[
\hat H_n\left(t^*+\frac{\varepsilon}{2}\right)
\geq H\left(t^*+\frac{\varepsilon}{2}\right)-\hat\Delta_n>\alpha.
\]
Since $\hat f_n\geq0$, the function $\hat H_n$ is nondecreasing. Hence
\[
t^*-\frac{\varepsilon}{2}\leq\hat t\leq t^*+\frac{\varepsilon}{2},
\]
so
\[
\PP(|\hat t-t^*|\geq\varepsilon)
\leq \PP(\hat\Delta_n\geq\gamma_\varepsilon)\to0.
\]
Therefore, $\hat t-t^*=o_P(1)$.
\end{proof}

\subsection{Proof of Lemma~\ref{lem:w_to_wx}}
\label{subsec:lemma_w_to_wx}

    \begin{proof}[Proof of Lemma~\ref{lem:w_to_wx}]
    By the Cauchy-Schwarz inequality, it holds deterministically for any $t\in \RR$ that  
    \@\label{eq:replace_w_with_wX}
     &\bigg|\frac{1}{n}\sum_{i=1}^n \hat{w}_iZ_i\ind\{f(X_i)\leq t\} - \frac{1}{n}\sum_{i=1}^n \hat{w}(X_i) Z_i\ind\{f(X_i)\leq t\} \bigg| \notag \\ 
     &= \bigg|\frac{1}{n}\sum_{i=1}^n Z_i (\hat{w}_i-\hat{w}(X_i))  \ind\{f(X_i)\leq t\}  \bigg| \notag \\ 
     & \leq \sqrt{\frac{1}{n}\sum_{i=1}^n Z_i^2 \ind\{f(X_i)\leq t\}}\sqrt{\frac{1}{n}\sum_{i=1}^n (\hat{w}_i-\hat{w}(X_i))^2} \leq  \sqrt{\frac{1}{n}\sum_{i=1}^n Z_i^2  }\sqrt{\frac{1}{n}\sum_{i=1}^n (\hat{w}_i-\hat{w}(X_i))^2} .  
    \@
    Since $\EE[Z_i^2]<\infty$, we know $\frac{1}{n}\sum_{i=1}^n Z_i^2 = O_P(1)$ by the Markov's inequality, and therefore 
    \$
     \sup_{t\in \RR}\bigg|\frac{1}{n}\sum_{i=1}^n \hat{w}_iZ_i\ind\{f(X_i)\leq t\} - \frac{1}{n}\sum_{i=1}^n \hat{w}(X_i) Z_i\ind\{f(X_i)\leq t\} \bigg|  \sqrt{\frac{1}{n}\sum_{i=1}^n Z_i^2  }\sqrt{\frac{1}{n}\sum_{i=1}^n (\hat{w}_i-\hat{w}(X_i))^2} =O_P(r_n)
    \$
    by Assumption~\ref{assump:obs_bal}.
    \end{proof}

\end{document}